\documentclass[aps,english,10pt,a4,superscriptaddress]{revtex4}

\newcommand{\Z}{{\mathbb{Z}}}

\usepackage{amsmath}
\usepackage{amsthm}
\usepackage{amssymb}
\usepackage{amsfonts}
\usepackage{bbm}
\usepackage{epsfig}
\usepackage{times}
\usepackage{babel}
\usepackage{color}
\usepackage{hyperref}
\usepackage{framed}

\newcommand{\R}{\mathbb{R}}
\newcommand{\C}{\mathbb{C}}
\newcommand{\N}{\mathbb{N}}
\newcommand{\tr}{{\mathrm{tr}}}

\newcommand{\Tr}{\operatorname{Tr}}

\newtheorem{theorem}{Theorem}
\newtheorem{definition}[theorem]{Definition}
\newtheorem{lemma}[theorem]{Lemma}
\newtheorem{example}[theorem]{Example}

\begin{document}

\title{Thermalization and canonical typicality in translation-invariant quantum lattice systems}

\author{Markus P.\ M\"uller}
\affiliation{Institut f\"ur Theoretische Physik, Universit\"at Heidelberg, Philosophenweg 19, D-69120 Heidelberg, Germany}
\affiliation{Department of Applied Mathematics, Department of Philosophy, University of Western Ontario, 1151 Richmond Street, London, ON N6A 5B7, Canada}
\affiliation{Perimeter Institute for Theoretical Physics, 31 Caroline Street North, Waterloo, ON N2L 2Y5, Canada}
\author{Emily Adlam}
\affiliation{Perimeter Institute for Theoretical Physics, 31 Caroline Street North, Waterloo, ON N2L 2Y5, Canada}
\affiliation{Centre for Quantum Information and Foundations, DAMTP, Centre for Mathematical Sciences, University of Cambridge, Wilberforce Road, Cambridge, CB3 0WA, U.K.}
\author{Llu\'is Masanes}
\affiliation{H.\ H.\ Wills Physics Laboratory, University of Bristol, Bristol BS8 1TL, UK}
\affiliation{University College London, Department of Physics \& Astronomy, London WC1E 6BT, UK}
\author{Nathan Wiebe}
\affiliation{Quantum Architectures and Computation Group, Microsoft Research, Redmond, WA 98052, USA}
\affiliation{Institute for Quantum Computing, University of Waterloo, Waterloo, Ontario, N2L 3G1, Canada}
\affiliation{Department of Combinatorics \& Opt., University of Waterloo, Waterloo, Ontario, N2L 3G1, Canada}

\date{September 9, 2015, \textbf{correction added on March 30, 2021}}

\begin{abstract}
It has previously been suggested that small subsystems of closed quantum systems thermalize under some assumptions;
however, this has been rigorously shown so far only for systems with very weak interaction between subsystems.
In this work, we give rigorous analytic results on
thermalization for translation-invariant quantum lattice systems with finite-range interaction of arbitrary strength,
in all cases where there is a unique equilibrium state at the corresponding temperature.
We clarify the physical picture by showing that subsystems relax towards the reduction of the global Gibbs state, not the local Gibbs state,
if the initial state has close to maximal population entropy and certain non-degeneracy conditions on the spectrum are satisfied.
Moreover, we show that almost all pure states with support on a small energy window are locally thermal in the sense of canonical typicality.
We derive our results from a statement on equivalence of ensembles generalizing earlier results by Lima, and
give numerical and analytic finite-size bounds, relating the Ising model to the finite de Finetti theorem. Furthermore,
we prove that global energy eigenstates are locally close to diagonal in the local energy eigenbasis, which
constitutes a part of the eigenstate thermalization hypothesis that is valid regardless of the integrability of the model.
\end{abstract}

\maketitle

\tableofcontents

\parskip .75ex

\section{Introduction}

How do closed quantum systems thermalize?
The last few years have seen a resurgence of interest in this old question,
motivated by new experimental~\cite{Trotzky} and numerical~\cite{Banuls} methods, relying on new ideas and methods from quantum information
theory~\cite{Goldstein,Popescu,Reimann,LindenPopescu,Short,Riera,ShortFarrelly,UWE13}.
Clearly, closed quantum systems in any given pure initial state cannot literally thermalize: unitary time evolution enforces that the global
state remains pure and will never become thermal, unless there is at least a tiny interaction with some environment.
However, small subsystems of closed quantum systems can equilibrate in a certain sense, as entanglement between the subsystem
and its remainder will lead to locally mixed states, and one may hope that these will in many cases resemble the ensembles of statistical physics.

Along these lines, it was suggested in~\cite{Goldstein} that typical pure quantum states in many-body systems resemble thermal states
on small subsystems due to entanglement, a property called ``canonical typicality''. However, no rigorous mathematical formulation of this was given in~\cite{Goldstein}.
Almost at the same time, it was rigorously proven in~\cite{Popescu} that typical pure quantum states in subspaces of bipartite Hilbert spaces are locally
close to some equilibrium state. However, this equilibrium state is not thermal in general. This raises the question what conditions are needed to
ensure that the local equilibrium state will be thermal, i.e.\ a Gibbs state.

In addition to these kinematical results, there has been major progress in understanding how closed quantum systems equilibrate
dynamically~\cite{Reimann,LindenPopescu,Short,ShortFarrelly,UWE13}. Regarding the emergence of the Gibbs state, the situation is
similar to the kinematical case: the subsystems approach some equilibrium state (for most times in some time interval), which is however not thermal in general.
The question is thus the same: under what conditions will the equilibrium state be thermal?

Important progress on this question was made in~\cite{Riera}: a rigorous bound on the distance $\mathcal{D}$ between the local equilibrium state and a thermal state was
established. This result has two drawbacks, however. First, the given bound is rather cumbersome, which is due to the great generality
of considering arbitrary Hamiltonians. Second, and more importantly, the upper bound on the distance $\mathcal{D}$ grows with the operator norm
of the interaction Hamiltonian which couples the subsystem to its surroundings. Thus, the bound becomes trivial as soon as the boundary of the
subsystem becomes moderately large, or the interaction becomes strong.

In this work, we give rigorous analytic proofs of dynamical and kinematic formulations of thermalization for interactions of
finite range, but arbitrary strength. By restricting to the special case of translation-invariant lattice systems as in Fig.~\ref{fig_setup},
we are able to prove the common belief that small subsystems are indeed close to thermal, under various natural conditions
on the spectrum and the initial state that depend on the specific setup and boundary conditions.
Our work also clarifies how thermalization should generally be formalized by showing that the resulting state will in general not be the local Gibbs state; rather,
it is the reduction of the global system's Gibbs state. This identification is made clear from the fact that the expected distance between the local reduced state and the
thermal state goes to zero in the thermodynamic limit.  In contrast, we show that boundary effects cause the local Gibbs state in general
to remain distinct from the thermal state even in the thermodynamic limit. This shows why earlier work led to bounds on the distance that necessarily grow with the interaction strength.

\begin{figure}[!hbt]
\begin{center}
\includegraphics[angle=0, width=7cm]{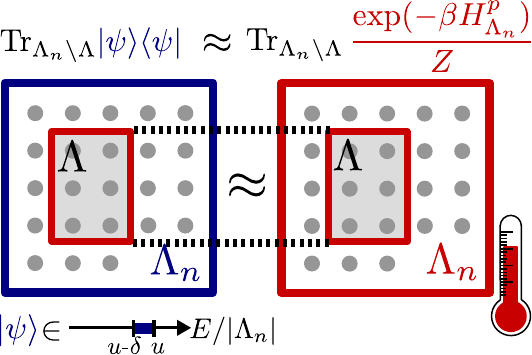}
\caption{Canonical typicality. A rectangular lattice $\Lambda_n$ evolves according to a translation-invariant finite-range interaction Hamiltonian $H_{\Lambda_n}^p$,
where ``$p$'' is for periodic boundary conditions (the case of arbitrary boundary conditions is treated in the Section~\ref{Sec3}).
If $|\psi\rangle$ is a generic state occupying only energies $E$ with $u-\delta\leq E/|\Lambda_n|\leq u$,
then small subsystems $\Lambda\subset\Lambda_n$ will, for large $n$, behave as if the full system was in a Gibbs state of the
corresponding temperature, for all possible measurements in the subsystem. Dynamically, the same will be true for $|\psi(t)\rangle$ for most times $t$ if
the initial state $|\psi(0)\rangle$ has close to maximal population entropy, and the spectrum satisfies certain non-degeneracy conditions.
}
\label{fig_setup}
\end{center}
\end{figure}

We are further able to provide concrete finite-size bounds, rather than asymptotic bounds, for two cases of interest.
We give tight analytic bounds for the distance between the reduction of a typical global pure state and
the local Gibbs state in the non-interacting case, which already turns out
to be a non-trivial problem, and we give numerical finite-size estimates for interacting models in one lattice dimension.
Building on the results by Low~\cite{Low}, we also show that the kinematical result on canonical typicality
remains true if global pure states are not drawn with respect to the unitarily invariant measure (which is hard to implement),
but according to an approximation of this measure (an ``$8$-design'') that can be sampled efficiently.

Finally, we address the question whether the given thermalization results can hold even on the level of single
energy eigenstates, as conjectured in the eigenstate thermalization hypothesis (ETH)~\cite{Deutsch, Srednicki}.
In a nutshell, the ETH claims that global energy eigenstates are locally close to a thermal state.
It is easy to see that the ETH cannot be true for all the models that we consider, and that additional assumptions (along the
lines of nonintegrability) are needed. However, we prove a result that constitutes
a part of the ETH which is true for all models with finite interaction range: global energy eigenstates are locally close to diagonal in the local energy eigenbasis.
We hope that this result (proven via Lieb-Robinson bounds) may serve as a first step towards a complete resolution of the ETH in future work.

\section{Summary of the main results}
\label{SecSummary}

We provide a self-contained summary of the main results of this paper in this section, focusing on periodic boundary conditions. The case of arbitrary boundary
conditions will be treated in Section~\ref{Sec3}.
While the detailed definitions will be given in Section~\ref{Sec3} (and are close to~\cite{Simon}), here we describe the setup and notation in a less formal way.

Our work considers the thermalization of interacting $d$-dimensional systems in a cubic or rectangular lattice in $\nu$ spatial dimensions.
These spins are constrained to interact with each other via finite-range translationally invariant Hamiltonians with arbitrary boundary conditions.
Although these restrictions are stringent, many models relevant to condensed matter physics, such as the Ising and Heisenberg models, satisfy these requirements.

We introduce the following notation to describe the lattice. We define the set of lattice sites to be $\Lambda:=[\lambda_1,\mu_1]\times\ldots
\times[\lambda_\nu,\mu_\nu]$, where $[\lambda,\mu]\subset\Z$ denotes the interval of integers between $\lambda$ and $\mu\geq\lambda$.
In particular, we consider sequences of regions $\Lambda_1\subset\Lambda_2\subset\Lambda_3\ldots$
that converge to the full infinite lattice $\mathbb{Z}^\nu$; for example, we may have the sequence of hypercubes $\Lambda_n=[-n,n]^\nu$. The physical interpretation is that
a region $\Lambda_n$ describes the actual physical system in the laboratory, and a subregion $\Lambda\subset\Lambda_n$ describes a small subsystem, cf.\ Fig.~\ref{fig_setup}.
The number of sites in a region $\Lambda$ is denoted $|\Lambda|$.  The ``particles'' located at each of these sites carry a $d$-dimensional Hilbert space $\C^d$.

Time evolution in $\Lambda_n$ is determined by a Hamiltonian $H_{\Lambda_n}^{BC}$ where the superscript explicitly denotes the type of boundary conditions that the Hamiltonian satisfies.
The choice of Hamiltonian is subject to some conditions defined as follows.
To every finite region $X\subset\Z^\nu$, we associate a self-adjoint operator $h_X$, and define the Hamiltonian with open boundary conditions
to be $H_\Lambda:=\sum_{X\subset\Lambda}h_X$.
We assume translation-invariance, i.e.\ $h_{X+y}$ equals $h_X$ (translated to the corresponding lattice sites), and finite-range of interaction,
i.e.\ there is some $r<\infty$ such that $h_X=0$ whenever the diameter of $X$ is larger than $r$. In the following,
we will exclude the case that the map $X\mapsto h_X$ is, up to physical equivalence~\cite{Simon},
everywhere identically zero.
As a simple example in one dimension, the Heisenberg model $H_{[1,n]}=-J\sum_{i=1}^{n-1} \vec \sigma_i \cdot \vec \sigma_{i+1}-h\sum_{i=1}^n \sigma_i^Z$,
with Pauli matrices $\vec\sigma=(\sigma^X,\sigma^Y,\sigma^Z)$, fits into  this framework, if we define $h_X$ as $-h\sigma_i^Z$ if $X=\{i\}$ for some integer $i$,
as $-J\vec\sigma_i\cdot \vec\sigma_{i+1}$ if $X=\{i,i+1\}$, and as zero for all other $X$.

The Hamiltonian with open boundary conditions, $H_{\Lambda_n}$, can be augmented with additional non-translationally invariant terms on the boundary of $\Lambda_n$ to obtain some $H_{\Lambda_n}^{BC}$.  The case of periodic boundary conditions is of particular importance to the remainder of the discussion and we denote such Hamiltonians by $H_{\Lambda_n}^p$.  More general boundary conditions are also permitted.  The only assumption will be that $\|H_{\Lambda_n}^{BC}-H_{\Lambda_n}\|_\infty/|\Lambda_n|\to 0$
as $n\to\infty$, where $\|\cdot\|_\infty$ is the operator norm. That is, the boundary terms only contribute a vanishing energy density.

While we aim at statements for \emph{finite} regions $\Lambda_n$, the thermodynamic limit $n\to\infty$ becomes important as a proof tool and
an indicator of phase transitions~\cite{Simon,BratteliRobinson}. We make extensive use of the following properties, which characterize the system's behavior in the thermodynamic limit.   States $\omega$ on the infinite lattice $\Z^\nu$ are given by families of density matrices
$(\omega_{\Lambda})_{\Lambda\subset\Z^\nu {\rm finite}}$, with $\omega_{\Lambda}=\Tr_{\Lambda'\setminus\Lambda}\omega_{\Lambda'}$ if $\Lambda\subseteq\Lambda'$.
Translation-invariant states $\omega$ on $\Z^\nu$ have entropy density $s(\omega):=\lim_{n\to\infty} \frac 1 {|\Lambda_n|} S(\omega_{\Lambda_n})$, with $S(\rho)=-\tr(\rho\log\rho)$ the
von Neumann entropy, and energy density $u(\omega):=\lim_{n\to\infty} \frac 1 {|\Lambda_n|} \tr(\omega_{\Lambda_n} H_{\Lambda_n})$.
A characteristic quantity for any given model and $\beta > 0$ is the equilibrium Helmholtz free energy density
$f_{\rm th}(\beta):=(-1/\beta)\lim_{n\to\infty}\frac 1 {|\Lambda_n|}\log\tr\exp(-\beta H_{\Lambda_n})$.
It holds
\[
   f_{\rm th}(\beta)= \inf\{f(\omega)\,\,|\,\, \omega\mbox{ translation-invariant state}\},
\]
where $f(\omega):=u(\omega)-s(\omega)/\beta$ is the Helmholtz free energy density~\cite{Simon} of state $\omega$.
For any finite region $\Lambda$, the Gibbs state at inverse temperature $\beta$ is $\gamma_\Lambda^{BC}(\beta):=\exp(-\beta H_{\Lambda}^{BC})/Z$,
with $Z$ the partition function. {Gibbs states on the infinite lattice can be defined in several different equivalent ways; here we use a variational principle:}
a translation-invariant state $\omega$ on the infinite lattice is by definition a Gibbs state at inverse temperature $\beta$
if it minimizes the free energy density, i.e.\ if $f(\omega)=f_{\rm th}(\beta)$. This definition is equivalent to the well-known KMS condition~\cite{Araki74}.

For every inverse temperature $\beta$, there is at least one Gibbs state $\omega_\beta$ on the infinite lattice; however,
the possibility of finite-temperature phase transitions implies that there may be more than one Gibbs state at the same $\beta$.
Consequently, we say that \emph{there is a unique equilibrium state around inverse temperature $\beta$} if there is a small interval around $\beta$ such that
for all $\beta'$ in that interval, there is only one Gibbs state at inverse temperature $\beta'$. This is true, for example,
if $\beta$ is smaller than some model-dependent critical inverse temperature~\cite{Kliesch}, and it is true for all $\beta$ if the lattice dimension is $\nu=1$~\cite{Araki75}.
A given energy density value $u$ will be called \emph{thermal} if it is strictly larger than the ground state energy density $u_{\min}$,
and strictly smaller than the infinite-temperature energy density $u_{\max}$. These are given by
$u_{\min}=\lim_{n\to\infty} \lambda_{\min}(H_{\Lambda_n})/|\Lambda_n|$ with $\lambda_{\min}$ the smallest eigenvalue,
and $u_{\max}:=\lim_{n\to\infty}\tr(H_{\Lambda_n})/(|\Lambda_n|d^{|\Lambda_n|})$.
If $u$ is thermal, then there is exactly one positive inverse temperature $\beta\equiv \beta(u)$ such
that the energy density $u(\omega_\beta)$ of the corresponding Gibbs state $\omega_\beta$ equals $u$~\cite{Simon}.

\subsection{Canonical typicality}
As suggested in~\cite{Goldstein}, we show that the Gibbs state arises in translation-invariant quantum lattice systems due to entanglement
between small subsystems and the remainder. Consider any model with a given thermal energy density $u$ such that there is a unique equilibrium state
around the corresponding inverse temperature $\beta=\beta(u)$. For $\delta>0$, define the \emph{microcanonical subspace}
\begin{equation}
   T_n^p:={\rm span}\left\{|E\rangle\,\,\left|\,\, u-\delta\leq E/|\Lambda_n|\leq u\right.\right\},
   \label{eqMicrocanDelta}
\end{equation}
where $H_{\Lambda_n}^p |E\rangle=E|E\rangle$ denotes the periodic boundary condition energy eigenstates on the global region $\Lambda_n$.
Choose any pure state $|\psi\rangle\in T_n^p$ at random according to the unitarily invariant measure. Then, with high probability,
this state will locally in $\Lambda\subset\Lambda_n$ be very close to the reduction of the global Gibbs state, as depicted in Fig.~\ref{fig_setup}:
\begin{theorem}[Summary of Theorem~\ref{TheCanonicalTypicalityPeriodic}]
\label{TheCanonTypMainText}
Fix $\delta>0$ and $u$ thermal. Then for every $\varepsilon\geq 0$,
the probability $p$ that a state $|\psi\rangle\in T_n^p$ sampled according to the unitarily invariant measure satisfies
\[
   \left\|\Tr_{\Lambda_n\setminus\Lambda}|\psi\rangle\langle\psi| - \Tr_{\Lambda_n\setminus\Lambda} \frac{\exp(-\beta H_{\Lambda_n}^p)}{Z}\right\|_1
   \geq \varepsilon+\Delta_{n,\Lambda}
\]
is doubly-exponentially small in the lattice size $|\Lambda_n|$; that is, $p\leq \exp\left( -\varepsilon^2 \exp(|\Lambda_n|s+o(|\Lambda_n|))\right)$, where
$s=s(\omega_\beta)$ is the entropy density of the corresponding Gibbs state, and $\Delta_{n,\Lambda}$ is a sequence of positive real numbers
with $\lim_{n\to\infty}\Delta_{n,\Lambda}=0$ for every fixed $\Lambda$. Here, $\beta$ can either be set equal to $\beta(u)$ as defined above,
or equal to the solution of $\tr(H_{\Lambda_n}^p \gamma_{\Lambda_n}^p(\beta))/|\Lambda_n|=u$ (which depends on $n$).
\end{theorem}
As illustrated in Fig.~\ref{fig_setup}, in the limit of large $n$, almost all pure states $|\psi\rangle$ in an energy window subspace will be locally
almost indistinguishable from the
Gibbs state at the corresponding temperature, since the one-norm distance $\|\rho-\sigma\|_1=2\max_{P=P^\dagger=P^2}
|\tr(\rho P)-\tr(\sigma P)|$ being small means that $\rho$ and $\sigma$ give similar expectation value for all possible measurements.
The theorem does not say how quickly $\Delta_{n,\Lambda}$ tends to zero {with increasing $n$}; we will come back to the question of finite-size
estimates later. Earlier work~\cite{Goldstein,Riera} attempted to prove that $\Tr_{\Lambda_n\setminus\Lambda}|\psi\rangle\langle\psi|$ is arbitrarily
close to the local Gibbs state $\gamma_\Lambda(\beta)=\exp(-\beta H_{\Lambda})/Z$. However, this can only be true if the interaction across the boundary of $\Lambda$
is very weak~\cite{Riera}; in particular, the given upper bound on the distance grows with the boundary of $\Lambda$ and is thus interesting
only if $\Lambda$ is small or if the lattice is one-dimensional.
Our theorem shows that in general the local Gibbs state has to be replaced by the reduction of the global Gibbs state to obtain
arbitrary closeness in the thermodynamics limit, unless one considers models that are fine-tuned such that the local Gibbs state agrees
with the reduction of the global Gibbs state.

Before we turn to the proof, we note that the unitarily invariant (Haar) measure in Theorem~\ref{TheCanonTypMainText} can be replaced by a more physically
realistic measure, namely an {$\eta$-approximate $t$-design}~\cite{Brandao,AmbainisEmerson}, for $t=8$ and $\eta=\exp(-|\Lambda_n|s+o(|\Lambda_n|))$. 
Here $o(\cdot)$ is (asymptotic) Landau notation, where $f_n= o(g_n)$ means that $f_n/g_n$ converges to zero in the limit $n\to\infty$. Such
$t$-designs are approximations to the Haar measure that can be efficiently generated in a time which is polynomial in the lattice size $|\Lambda_n|$.
It follows from the results of Low~\cite{Low} that Theorem~\ref{TheCanonTypMainText} remains valid, however with a probability value that is only exponentially
{(not doubly-exponentially) small in the lattice site -- see Theorem~\ref{ThePeriodicBCDerandomized}.}

To prove Theorem~\ref{TheCanonTypMainText}, we invoke the results of~\cite{Popescu}, which tell us that $\Tr_{\Lambda_n\setminus\Lambda}|\psi\rangle\langle\psi|$
is with high probability close to $\Tr_{\Lambda_n\setminus\Lambda}\tau_n$, where $\tau_n$ is the uniformly mixed state on $T_n^p$. We obtain Theorem~\ref{TheCanonTypMainText}
directly, with all constants, if we set $\Delta_{n,\Lambda}$ up to corrections of order $\exp\left(-\frac 1 2 |\Lambda_n| s +o(|\Lambda_n|)\right)$ {(cf. eq.~(\ref{eqdeltaDelta}))} equal to
\begin{equation}
   \delta_{n,\Lambda}:=\left\| \Tr_{\Lambda_n\setminus\Lambda} \tau_n - \Tr_{\Lambda_n\setminus\Lambda} \frac{\exp(-\beta H_{\Lambda_n}^p)}{Z}\right\|_1.
   \label{eqDeltaMNMain}
\end{equation}
It remains to prove that $\delta_{n,\Lambda}\to 0$ as $n\to\infty$. However, $\tau_n$ is nothing but the microcanonical ensemble, and the statement left
to prove is that its predictions on small subsystems $\Lambda$ are equivalent to those of the canonical ensemble in the thermodynamic limit.
Thus, we are naturally led to study the problem of equivalence of ensembles in our setting.

\subsection{Equivalence of ensembles}
To state our result,
note that we can regard $\Lambda_n$ as a torus, by identifying $\mu_i+1$ in the interval $[\lambda_i,\mu_i]$ with $\lambda_i$; this way, we can define
periodic translations of $\Lambda_n$ as those of the resulting torus. A state $\tau_n$ on $\Lambda_n$ will be called $\Lambda_n$-translation-invariant if it is invariant with
respect to all periodic translations of $\Lambda_n$. {Using this notion, our main technical result on equivalence of ensembles reads as follows:}
\begin{theorem}[Summary of Theorem~\ref{TheEquivalence}]
\label{TheMain2MainText}
Suppose that $(\tau_n)_{n\in\N}$ is any sequence of $\Lambda_n$-translation-invariant states on $\Lambda_n$, and $\beta>0$ such that there is a unique
equilibrium state around inverse temperature $\beta$. If
\begin{equation}
   \limsup_{n\to\infty}\frac 1 {|\Lambda_n|}\left(\tr(\tau_n H_{\Lambda_n}^{BC})-S(\tau_n)/\beta\right)\leq f_{\rm th}(\beta)
   \label{eqFreeEnergy}
\end{equation}
for some choice of boundary conditions $BC$, then
\begin{equation}
   \lim_{n\to\infty}\left\| {\rm Tr}_{\Lambda_n\setminus\Lambda} \tau_n - {\rm Tr}_{\Lambda_n\setminus\Lambda}
   \frac{\exp(-\beta_n H_{\Lambda_n}^p)} {Z_n}\right\|_1=0,
   \label{eqEquivEnsMainText}
\end{equation}
where we may set $\beta_n$ either equal to the fixed value $\beta$, or equal to the solution of $\tr(H_{\Lambda_n}^p \gamma_{\Lambda_n}^p(\beta_n))/|\Lambda_n|=u(\beta)$.
\end{theorem}

{Theorem~\ref{TheMain2MainText} implies Theorem~\ref{TheCanonTypMainText}:}
If $\tau_n$ is the microcanonical ensemble, i.e.\ maximal mixture on $T_n^p$, then
$\tr(\tau_n H_{\Lambda_n}^p)/|\Lambda_n|\leq u$ by construction, and $S(\tau_n)=\log\dim(T_n^p)=s|\Lambda_n|+o(|\Lambda_n|)$ according to~\cite[Thm.\ IV.2.14]{Simon}
{(as~\cite{Simon} does not provide a proof, we reproduce the proof in Lemma~\ref{LemReproduced} below).}
Since $u-s/\beta=f_{\rm th}(\beta)$, (\ref{eqFreeEnergy}) holds, which shows equivalence to the canonical ensemble, $\lim_{n\to\infty}\delta_{n,\Lambda}=0$,
and establishes Theorem~\ref{TheCanonTypMainText}.

The crucial property of the microcanonical subspace $T_n^p$ used in this proof is its dimensionality (which is close to maximal given its energy density $u$),
namely $\lim_{n\to\infty}(1/|\Lambda_n|)\log\dim T_n^p=s$.
It follows from Lemma~\ref{LemReproduced} that this property is satisfied if the width $\delta>0$ in the definition of the microcanonical subspace~(\ref{eqMicrocanDelta}) is constant in $n$,
which corresponds to an extensive energy uncertainty. In general,
one can also choose an $n$-dependent width $\delta\equiv \delta_n$; as long as $\delta_n$ tends to zero slowly enough, the necessary limit identity will still hold.
Unfortunately, giving a concrete expression for a possible choice of $\delta_n$ amounts to proving a generalization of Lemma~\ref{LemReproduced} for ``small'' microcanonical subspaces,
and we do not currently have such a generalization.

However, in the special case of the non-interacting Ising model described in Subsection~\ref{SubsecFiniteSize} below,
it is easy to see via standard inequalities (like the ones used in the proof of Theorem~\ref{TheMainFiniteSize}) that one can choose $\delta_n\geq c (\log n)/n$,
with $c>0$ some constant depending on $u$. It is therefore plausible to expect that a comparable scaling of $\delta_n$ might be possible
also in the interacting case.

We now sketch the proof of Theorem~\ref{TheMain2MainText}. We first show that $(\tau_n)_{n\in\N}$ has at least one limit point $\omega$ as a state on
the infinite lattice. Since every $\tau_n$ is $\Lambda_n$-translation-invariant, $\omega$ is translation-invariant, and~(\ref{eqFreeEnergy})
implies that $f(\omega)=f_{\rm th}(\beta)$. Thus, $\omega$ is the unique Gibbs state $\omega_\beta$, and so
\begin{equation}
   \lim_{n\to\infty}\Tr_{\Lambda_n\setminus\Lambda} \tau_n = (\omega_\beta)_{\Lambda}.
   \label{eqLimitAsState}
\end{equation}
Consider the special case where $\tau_n$ equals the local Gibbs state $\gamma_n:=\gamma_{\Lambda_n}^p(\beta)$ which appears in~(\ref{eqEquivEnsMainText}).
Every $\gamma_n$ is $\Lambda_n$-translation-invariant and minimizes the free energy locally, hence
\[
   \tr(\gamma_n H_{\Lambda_n}^p)-S(\gamma_n)/\beta \leq \tr((\omega_\beta)_{\Lambda_n} H_{\Lambda_n}^p)-S((\omega_\beta)_{\Lambda_n})/\beta
   \stackrel{n\to\infty}\longrightarrow f_{\rm th}(\beta),
\]
which shows that~(\ref{eqFreeEnergy}) is satisfied for $\tau_n=\gamma_n$. Consequently $\lim_{n\to\infty}\Tr_{\Lambda_n\setminus\Lambda} \gamma_n = (\omega_\beta)_{\Lambda}$,
and combining this with~(\ref{eqLimitAsState}) proves the theorem.

This proof strategy has been pioneered by Lima~\cite{Lima1,Lima2}; however, our result is more general. In particular,
we allow a more general set of possible interactions, and permit $\beta_n\neq \beta$ to be determined from the finite region $\Lambda_n$.

\subsection{Dynamical thermalization}
It has been shown in~\cite{LindenPopescu,Short,ShortFarrelly} that subsystems of closed quantum systems
equilibrate, subject to some conditions on the initial state and spectrum. In general, the equilibrium state depends on the initial state, and is not thermal unless additional
conditions are met~\cite{Riera}. However, for translation-invariant systems, we can say more. Consider any initial state $\rho_0^{(n)}$ on $\Lambda_n$,
pure or mixed. The index $n$ indicates that the state is chosen to be a function of the lattice size $n$.
We can think of a simple dependence such as $\rho_0^{(n)}=\rho_0^{\otimes \Lambda_n}$ for some fixed (single-site)  state $\rho_0$ on $\C^d$;
however, the only technical condition
we need to assume is that the density of the inner energy $U_n:=\tr(\rho_0^{(n)} H_{\Lambda_n}^p)$ converges to some well-defined
thermal energy density $u:=\lim_{n\to\infty}U_n/|\Lambda_n|$.

The state evolves unitarily
under the Hamiltonian $H_{\Lambda_n}^p$, i.e.\ $\rho^{(n)}(t)=\exp(-itH_{\Lambda_n}^p) \rho_0^{(n)}
\exp(itH_{\Lambda_n}^p)$. We can define the \emph{population entropy} $\bar S(\rho_0^{(n)})$ as follows. From the spectral decomposition
$H_{\Lambda_n}^p=\sum_i E_i\pi_i$, compute the weights $\lambda_i:=\tr(\rho_0^{(n)}\pi_i)$, and set $\bar S(\rho_0^{(n)}):=-\sum_i \lambda_i\log\lambda_i$.
Similarly, there is an inverse temperature $\beta_n$ corresponding to $\rho_0^{(n)}$,
defined by $\tr(H_{\Lambda_n}^p \gamma_{\Lambda_n}^p(\beta_n))=U_n$.
Denote the time average by $\langle\cdot\rangle$, i.e.\ $\rho_{\rm avg}^{(n)}:=\langle \rho^{(n)}(t)\rangle:=
\lim_{T\to\infty}(1/T)\int_0^T \rho^{(n)}(t)dt$. Then {the actual state at time $t$ is close to $\rho_{\rm avg}^{(n)}$ for most times $t$,
and this state is close to thermal:}
\begin{theorem}[Summary of Theorem~\ref{TheThermalizationPeriodic}]
\label{TheMain3}
If there is a unique equilibrium state around inverse temperature $\beta:=\lim_{n\to\infty}\beta_n$,
if the (possibly pure) initial state has close to maximal population entropy, in the sense that
\begin{equation}
   \bar S(\rho_0^{(n)})\geq S(\gamma_{\Lambda_n}^{p}(\beta_n))-o(|\Lambda_n|),
   \label{eqMaxEntMain}
\end{equation}
and if each $H_{\Lambda_n}^p$ is non-degenerate {(i.e.\ all eigenspaces are one-dimensional)}, then unitary time evolution thermalizes the subsystem $\Lambda$ for most times $t$:
\begin{eqnarray}
   \left\langle \left\| \Tr_{\Lambda_n\setminus\Lambda}\rho^{(n)}(t)-\Tr_{\Lambda_n\setminus\Lambda}\rho_{\rm avg}^{(n)}\right\|_1\right\rangle
   &\leq& d^{|\Lambda|}
     \,\sqrt{D_G}\, \exp\left(- \frac{s(\omega_\beta)^2}{4\log d} |\Lambda_n| + o(|\Lambda_n|)\right),\quad\mbox{and}\label{eqDistToAve} \\
   \quad\lim_{n\to\infty}\left\| \Tr_{\Lambda_n\setminus\Lambda}\rho_{\rm avg}^{(n)} -  \Tr_{\Lambda_n\setminus\Lambda}
   \frac{\exp(-\beta_n H_{\Lambda_n}^p)}{Z_n}\right\|_1
   &=& 0,\qquad\strut\label{eqDynMain2}
\end{eqnarray}
where $D_G$ is the gap degeneracy~\cite{ShortFarrelly} of $H_{\Lambda_n}^p$, defined by $D_G=\max_E |\{(i,j)\,\,|\,\, i\neq j, E_i-E_j=E\}|$,
where $E_i$ denotes the eigenvalues of $H_{\Lambda_n}^p$.
\end{theorem}

In Theorem~\ref{TheThermalizationArbitrary}, we generalize this result to the case of arbitrary boundary conditions and degenerate $H_{\Lambda_n}^{BC}$.
Unlike~(\ref{eqDynMain2}), which expresses equivalence of the time-averaged state $\rho_{\rm avg}^{(n)}$ and
the thermal state $\gamma_{\Lambda_n}^p(\beta_n)$ for local
observables $A$ on $\Lambda$, the generalized version shows equivalence of these global states on a different set of observables~\cite{ShortFarrelly},
arising from averaging observables $A$ over translations of $\Lambda$. We also show numerically {in Subsection~\ref{SecNumerical}} that the conditions of non-degeneracy
of $H_{\Lambda_n}^p$ and $D_G=1$ are generically satisfied for randomly chosen translation-invariant nearest-neighbor interactions in one lattice dimension.

Our proof of Theorem~\ref{TheMain3} follows similarly to the proof of the results of~\cite{ShortFarrelly}. First we have to show that the ``effective dimension'' $d_{\rm eff}=e^{S_2(\lambda)}$ is large,
with $S_\alpha(\rho):=(\log\tr(\rho^\alpha))/(1-\alpha)$ the $\alpha$-R\'enyi entropy. We do this via the inequality $S_2\geq 2\varepsilon(S-\varepsilon/(1+\varepsilon)S_0)$
for $0\leq\varepsilon\leq 1$, which we prove from results of~\cite{Zyczkowski}, establishing~(\ref{eqDistToAve}). From
$S(\rho_{\rm avg}^{(n)})\geq \bar S(\rho_0^{(n)})$, we conclude that $\rho_{\rm avg}^{(n)}=:\tau_n$
satisfies~(\ref{eqFreeEnergy}). We then apply Theorem~\ref{TheMain2MainText} to prove~(\ref{eqDynMain2}).

As an example, if $\rho_0^{(n)}$ is a pure state $|\psi_0^{(n)}\rangle\sim \sum_{u-\delta<E_i/|\Lambda_n|<u}|E_i\rangle$ which is a ``flat'' uniform
superposition of eigenstates $|E_i\rangle$ of $H_{\Lambda_n}^p$, Theorem~\ref{TheMain3} applies.
This recovers results of~\cite{Riera}, albeit in a different context.

\subsection{Finite-size estimates}
\label{SubsecFiniteSize}
Estimates on how large $\Lambda_n$ has to be to have good agreement with our asymptotic results, in particular bounds on
$\delta_{n,\Lambda}$ in~(\ref{eqDeltaMNMain}), are expected to depend strongly on the details of the model, such as distance
to phase transitions, correlation lengths etc.~\cite{Deserno}. To get some intuition, we now give analytic bounds for the non-interacting Ising model, {
which already turns out to be a non-trivial problem. For this model, it was already shown in~\cite{Popescu} that local reduced states
are close to thermal in the sense of Theorem~\ref{TheCanonTypMainText}; however, no explicit analytic bounds on the distance have been given in~\cite{Popescu}.
Here we provide tight analytic finite-size bounds.

We set $\Lambda_n=[1,n]\subset\Z^1$, and $H_\Lambda:=\sum_{i\in\Lambda} Z_i$, where $Z_i$ is the Pauli $Z$-matrix at site $i$.
Then the microcanonical state $\tau_n$ is
permutation-invariant, and the canonical state is a product state, {$\gamma_{\Lambda_n}(\beta)=\gamma_\beta^{\otimes n}$,
with $\gamma_\beta:=\gamma_{\{1\}}(\beta)$ the single-site Gibbs state.
We are interested in estimating the distance $\delta_{n,\Lambda}$ in~(\ref{eqDeltaMNMain}).}
In the case where the energy value of the microcanonical subspace~(\ref{eqMicrocanDelta}) is sharp, i.e.\ $\delta=0$,
the state $\tau_n$ is the uniform mixture over a type class, {that is, over the subspace spanned by eigenvectors with a fixed frequency of ``spin-up''.
In this case, it turns out that we can apply the proof of the classical finite de Finetti theorem~\cite{Diaconis} to obtain}
\begin{equation}
   \left\|\Tr_{\Lambda_n\setminus\Lambda}\tau_n-  \gamma_\beta^{\otimes m}\right\|_1\leq\frac{4 m}n,
   \label{eqClassicalDeFinetti}
\end{equation}
where $m:=|\Lambda|$.
Thus, in order to maintain a fixed $1$-norm distance between the states, the total system size $n$ has to be increased linearly with the size of the subsystem $m$.
As mentioned before eq.~(\ref{eqDeltaMNMain}), this also upper-bounds the distance $\Delta_{n,\Lambda}$ in Theorem~\ref{TheCanonTypMainText}
up to corrections exponentially small in the lattice size.

The case of finite energy uncertainty $\delta>0$ is more difficult to treat. If we assume each of the lattice sites holds a qubit ($d=2$) and take an appropriate rescaling of the energy then
\[
   S\left(\strut \gamma_\beta^{\otimes m}\, \left\|\,{\rm Tr}_{\Lambda_n\setminus\Lambda_m}\tau_n \right.\right)\leq \frac {(1-\delta)u}{u-\delta}\cdot \frac m {n-m}+\frac{m u \delta}{u-\delta}\left(
   1+\frac m {n-m}\right)
\]
whenever $m\leq n(u-\delta)$, with $S(\rho\|\sigma):=\tr(\rho\log\rho-\rho\log\sigma)$ the quantum relative entropy.
This claim is formally stated as Lemma~\ref{LemResultDeltaZero}.
For $\delta=0$ (and $m\ll n$), this inequality is similar to~(\ref{eqClassicalDeFinetti}) above, but now with the relative entropy as distance measure.
We expect it to be tight (i.e.\ not to allow for significant improvements) in the case $\delta=0$, since it is well-known that the bound in the classical finite de Finetti theorem,
and thus~(\ref{eqClassicalDeFinetti}), cannot be significantly improved. However, this inequality on the relative entropy has the drawback that it is only interesting
as long as $\delta\lesssim 1/m$. The question arises how $n$ has to be scaled with growing subsystem size $m$
in order to achieve a fixed distance $\delta>0$ (for $\delta=0$, we have seen that $n$ has to be increased linearly with $m$).
In Theorem~\ref{TheMainFiniteSize}, we settle this question up to a correction
term of the order $\log n$: under some conditions on the variables, we show that
\[
   \left\| {\rm Tr}_{\Lambda_n\setminus\Lambda_m}\tau_n - \gamma_\beta^{\otimes m}\right\|_1 \leq
   \frac{2\delta}{n\sqrt{u}}+\sqrt{\frac m {n-m}\left(1+\frac{4\log n}{\log\frac{1-u}u}\right)}.
\]
This inequality is not tight in general (as one sees by comparing with~(\ref{eqClassicalDeFinetti}) for $\delta=0$), but it shows that $n$ has to be increased only slightly
superlinearly with $m$ in order to achieve a fixed $1$-norm distance also in the case $\delta>0$. We leave it as an open question whether the $\log n$ term can be removed.

In order to get some intuition for what happens in the \emph{interacting} case, we numerically study random nearest-neighbor interactions in one lattice
dimension in Subsection~\ref{SecNumerical}. It turns out that the behavior that we have shown analytically for non-interacting models remains approximately
valid also in the interacting case (as far as one can tell for the small lattice sizes $n\leq 11$ that are numerically tractable), see in particular Fig.~\ref{fig:global}.
However, we leave it open whether a similar behavior remains valid in lattice dimensions $\nu\geq 2$, where finite-temperature phase transitions become relevant.

\subsection{Towards eigenstate thermalization}
The question whether some of the results above can be strengthened to hold for \emph{individual energy eigenstates} is known as the
eigenstate thermalization hypothesis (ETH)~\cite{Deutsch,Srednicki}. For example, consider our result on dynamical thermalization, Theorem~\ref{TheMain3}.
For this result to hold, eq.~(\ref{eqMaxEntMain}) must be satisfied, which says that the initial state populates a large number of energy levels.

The question arises whether this assumption can be dropped. In the most extreme case, we could have an energy eigenstate $|E\rangle$ as the initial
state, i.e.\ $\rho_0=|E\rangle\langle E|$. (This notation does not assume non-degeneracy of the spectrum; $|E\rangle$ is an arbitrary pure state in the
eigenspace corresponding to energy $E$.) Energy eigenstates do not evolve, such that $\rho^{(n)}(t)=\rho_0$ is constant in time. Thus $\rho^{(n)}(t)$ is close
to thermal for most times $t$ if and only if
the reduced state ${\rm Tr}_{\Lambda_n\setminus\Lambda} |E\rangle\langle E|$ is close to thermal.

To formulate eigenstate thermalization in more detail, consider
the setup in Fig.~\ref{fig_regions_main}. We have argued above that one should not expect that the local marginals of random global pure states $|\psi\rangle$
are close to a local Gibbs state, due to boundary effects (which led us to consider the reduction of the global Gibbs state instead). More generally, to take boundary
effects into account, we can enlarge the subregion $\Lambda$ by a shell of width $l$; if $l$ is large enough, one would expect that
\begin{equation}
  \label{eqg}
   \Tr_{\Lambda_n\setminus\Lambda} |E\rangle\langle E|\approx \Tr_{\Lambda_n\setminus\Lambda}\gamma_{\Lambda_n}(\beta)\approx   \Tr_{\Lambda_{\rm shell}}\gamma_{\Lambda'}(\beta).
\end{equation}
It is immediately clear that a statement like this cannot literally be true for all eigenstates $|E\rangle$ of all models that we consider: the non-interacting Ising model,
where some eigenstates are product states (and thus marginals are pure and not thermal), is a counterexample.

However, we can prove a weaker version of this statement which is true for all eigenstates of all translation-invariant models with finite range interaction:
there is a state $\omega_E$ on $\Lambda'$ such that ${\rm Tr}_{\Lambda_n\setminus\Lambda}|E\rangle\langle E|\approx {\rm Tr}_{\Lambda_{\rm shell}} \omega_E$,
where $\omega_E$ partially resembles a thermal state. That is, $\omega_E$ does not necessarily have Boltzmann weights on its diagonal (as one would expect from the thermal
state $\gamma_{\Lambda'}(\beta)$), but its off-diagonal elements are close to zero, as they are for the thermal state.

We formulate and prove this result by applying a version of the Lieb-Robinson bound~\cite{LiebRobinson,Bravyi,Masanes}:
for models with finite-range interaction, it states that there are constants $c,C,v>0$ such that for all operators $X$ and $Y$ supported on finite regions
$\mathcal{X},\mathcal{Y}$ of distance $\Delta$, it holds
$\| [X(t),Y] \|_\infty \ \leq\  C\, \|X\|_\infty \|Y\|_\infty \min\{{\cal |X|,|Y|}\}\, e^{-c[\Delta -v |t|]}$, where $X(t)=e^{i H_{\Lambda_n} t} X e^{-i H_{\Lambda_n} t}$.
The constants also appear in the following theorem, where we assume in particular that the Hamiltonian only has interactions between sites of distance $r$ or less.

\begin{figure}[!hbt]
\begin{center}
\includegraphics[angle=0, width=4cm]{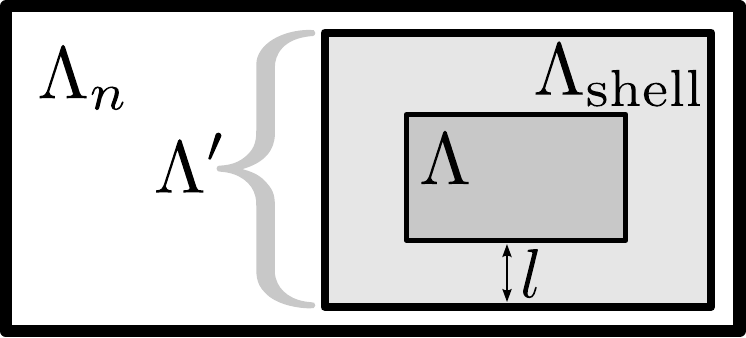}
\caption{Subregions of the whole lattice $\Lambda_n$. We enlarge $\Lambda$ by setting $\Lambda'=\Lambda\cup\Lambda_{\rm shell}$,
where $\Lambda_{\rm shell}$ contains all sites outside of $\Lambda$ which have distance $l$ or less to $\Lambda$. The number of terms of $H_{\Lambda_n}$ that
have support on both $\Lambda'$ and $\Lambda_n\setminus\Lambda'$ is denoted $A$, which quantifies the size of the boundary area of $\Lambda'$.
}
\label{fig_regions_main}
\end{center}
\end{figure}
\begin{theorem}[Summary of Theorem~\ref{TheWeakLocalDiagonality}]
\label{TheETHMain}
There is a state $\omega_E$ on $\Lambda'$ such that
\begin{equation}
   \left\|\Tr_{\Lambda_{\rm shell}}(\omega_E)-\Tr_{\Lambda_n\setminus\Lambda}|E\rangle\langle E|\right\|_1\leq \kappa\cdot e^{-c(l-r)/2},
   \label{eqDistance}
\end{equation}
where $\kappa=2AJ(CA+2)\sqrt{\frac{l-r}{8cv^2}}$ and $J=\max_X\|h_X\|_\infty$, which is close to diagonal in the eigenbasis $\{|e\rangle\}$ of $H_{\Lambda'}$, i.e.
\begin{equation}
   |\langle e_1|\omega_E|e_2\rangle| \leq e^{-(l-r)(e_1-e_2)^2/(8cv^2)}.
   \label{eqCoherences}
\end{equation}
\end{theorem}
This result does not assume translation-invariance; finite range of interaction is sufficient for its validity.
The ETH corresponds to the claim that the theorem holds for the particular choice $\omega_E=\gamma_{\Lambda'}(\beta)$.
As discussed above, the ETH cannot be true in general for all eigenstates of all
models we consider; intuitively, some additional assumptions, possibly along the lines of nonintegrability, are needed.

Even though the mathematical details of the proof are cumbersome, it has a simple physical interpretation. We define $\omega_E$
by evolving $\Tr_{\Lambda_n\setminus\Lambda'}|E\rangle\langle E|$ according to $H_{\Lambda'}$ and averaging the result over small $t$;
concretely, $\omega_E := \int_{-\infty}^\infty\!\! dt\, g(t)\, e^{-i H_{\Lambda'} t}\, {\rm Tr}_{\Lambda_n\setminus \Lambda'} |E\rangle\! \langle E|\,  e^{i H_{\Lambda'} t}$,
with $g(t)$ some Gaussian. The Lieb-Robinson bound guarantees finite speed of information transmission, such that the result will within $\Lambda$
still look very much as if the initial state $|E\rangle\langle E|$ evolved according to the full Hamiltonian $H_{\Lambda_n}$, if the shell is large enough.
Since $|E\rangle\langle E|$ is stationary, this leads to~(\ref{eqDistance}). On the other hand,
interaction across the boundary of $\Lambda'$ will decohere the state $\Tr_{\Lambda_n\setminus\Lambda'}|E\rangle\langle E|$;
in particular, coherences corresponding to energy levels $e_1,e_2$ with large $|e_1-e_2|$ will be suppressed, which yields~(\ref{eqCoherences}).

In Theorems~\ref{TheCanonTypMainText}, \ref{TheMain2MainText} and~\ref{TheMain3}, we quantify thermalization by the distance to the
marginal of the global Gibbs state, that is, $\Tr_{\Lambda_n\setminus\Lambda} \exp(-\beta_n H_{\Lambda_n}^p) /Z_n$. This contrasts with the
above setup~(\ref{eqg}), where one starts with the thermal state in a subregion $\Lambda'$ containing a shell around $\Lambda$.
In fact, it is easy to see that Theorems~\ref{TheCanonTypMainText}, \ref{TheMain2MainText} and~\ref{TheMain3} also hold if
we replace the thermal state on all of $\Lambda_n$ by the thermal state on $\Lambda':=\Lambda_{n'}\supseteq \Lambda$ for $n'\ll n$, as long
as $n'\to\infty$ with $n\to\infty$; that is, if we consider the distance to $\Tr_{\Lambda_{n'}\setminus\Lambda} \exp(-\beta_{n'} H_{\Lambda_{n'}}^p) /Z_{n'}$ instead.
The drawback, however, is that we do not have any non-trivial bounds for fixed finite $n'$ (in contrast to Theorem~\ref{TheETHMain}), which renders
the formulation of the first three theorems in terms of the more general shell setup mathematically equivalent to their current formulation.
Finding non-trivial finite-size bounds, in particular on the size of the shell, remains as an interesting open problem.

\section{Proofs of the main results}
\label{Sec3}
All results are formulated in two versions, namely for periodic and for arbitrary boundary conditions (BC). The main results can be found in the following places:
\begin{itemize}
\item \textbf{Equivalence of ensembles}: Main technical statement is Theorem~\ref{TheEquivalence} (periodic BC), with Example~\ref{ExFlat}
giving the standard formulation in comparing the microcanonical and the canonical ensemble. The corresponding formulations for arbitrary BC
are given in Theorem~\ref{TheEquivalence2} and Example~\ref{ExFlat2}. Note that the norm $\|\cdot\|_{\{m\}}$ appearing there can be replaced
by $\|\cdot\|_{[m]}$, measuring the difference of expectation values on \emph{$m$-block averaged observables} only, due to Lemma~\ref{LemNormEquivalence}.
\item \textbf{Canonical typicality}: For periodic BC, the main result is Theorem~\ref{TheCanonicalTypicalityPeriodic}, with a derandomized version
in terms of $8$-designs given in Theorem~\ref{ThePeriodicBCDerandomized}. The corresponding formulations for arbitrary BC are given in
Theorem~\ref{TheCanonicalTypicalityArbitrary} and Theorem~\ref{TheCanonTypArbitraryDesigns}.
\item \textbf{Dynamical thermalization}: For periodic BC, the main result is Theorem~\ref{TheThermalizationPeriodic}, and for arbitrary BC it is Theorem~\ref{TheThermalizationArbitrary}.
\item \textbf{Finite-size bounds without interaction:} Lemma~\ref{LemNonIntDiaconis} relates the Ising model and equivalence of ensembles
for sharp energy eigenspaces ($\delta=0$) and local Hilbert space dimension $d=2$ to the finite de Finetti theorem. Lemma~\ref{LemResultDeltaZero} is a new derivation
(compared to Diaconis and Freedman~\cite{Diaconis}) in terms of the relative entropy. The main result is Theorem~\ref{TheMainFiniteSize}, proving that the
scaling (bath size increasing linearly with system size to achieve fixed $1$-norm error) remains basically valid also for $\delta>0$.
\item \textbf{Numerical results}: They are given in Subsection~\ref{SecNumerical}, confirming that some assumptions from the main theorems (on degeneracy
of spectra etc.) are generically satisfied. Moreover, they show that the qualitative finite-size scaling that has been proven analytically for non-interacting
systems seems to remain valid for interacting systems with periodic BC, at least for lattice dimension $\nu=1$.
\item \textbf{Eigenstate thermalization}: The main result is Theorem~\ref{TheWeakLocalDiagonality}, showing that energy eigenstates are locally ``weakly diagonal''.
Note that this result does not assume translation-invariance (only finite range of interaction).
\end{itemize}
The notation is specified in Subsection~\ref{SecEquiv} below.
In comparison to Section~\ref{SecSummary}, statements about the minimization of the Helmholtz free energy density $f(\omega):=u(\omega)-s(\omega)/\beta$
are replaced by statements about the maximization of $-\beta\, f(\omega)=s(\omega)-\beta\, u(\omega)$ (following mathematical physics tradition),
which has to be compared with the ``pressure''
\[
   p(\beta,\Phi)=-\beta\, f_{\rm th}(\beta) \qquad (\beta>0).
\]
This has the advantage that $p(\beta,\Phi)$ is also defined for $\beta=0$, i.e.\ infinite temperature.
Moreover, convexity of $\beta\mapsto p(\beta,\Phi)$ will play a crucial role.
Similarly, to conform with mathematical physics literature, we will write $\Phi(X)$ instead of $h_X$, and the map $\Phi$ will be called an ``interaction''.
Furthermore, we will assume that the small subsystem $\Lambda$ equals $\Lambda_m$ for some fixed $m$, which is no loss of generality.

\subsection{Equivalence of ensembles}
\label{SecEquiv}
We start by fixing some notation. We consider a $\nu$-dimensional quantum lattice system, with local Hilbert space
dimension $d$. To every $x\in \Z^\nu$, we associate a local algebra of observables
$\mathcal{A}_x$, which is a copy of $M_d(\C)$, {the algebra of complex $d\times d$ matrices.}
For every finite region $\Lambda\subset\Z^\nu$, we have the local observable
algebra $\mathcal{A}_\Lambda:=\bigotimes_{x\in\Lambda}\mathcal{A}_x$. For every $y\in\Z^\nu$, there is a translation automorphism $\gamma_y$,
mapping observables $A$ in a region $\Lambda$, i.e.\ $A\in \mathcal{A}_\Lambda$, to the corresponding observable $\gamma_y(A)$ in the
translated region $\Lambda+y$, i.e.\ $\gamma_y(A)\in\mathcal{A}_{\Lambda+y}$.

To every finite region $X\subset\Z^\nu$, we associate an \emph{interaction} $\Phi(X)$, which is a self-adjoint operator in $\mathcal{A}_X$,
describing the interaction of the spins in region $\Lambda$. For finite $\Lambda\subset\Z^\nu$, the \emph{local Hamiltonian} $H_\Lambda$ is
\[
   H_\Lambda:=\sum_{X\subset\Lambda} \Phi(X).
\]
We assume that our interaction has finite range, i.e.\ that  $\Phi(X)=0$ whenever the diameter of $X$ is larger than $r$ for some fixed $r\in\N$.
Furthermore, we assume translation-invariance, which can be expressed as
\[
   \Phi(X+y)=\gamma_y(\Phi(X))\qquad\mbox{for all }X\subset\Z^\nu\mbox{ finite},y\in\Z^\nu.
\]
We can also define an observable algebra $\mathcal{A}_\infty$ for the infinite lattice $\Z^\nu$
by a suitable limit procedure, called the \emph{quasi-local algebra}, see~\cite{Simon} for details. The (operator) norm on $\mathcal{A}_\infty$ will be denoted $\|\cdot\|_\infty$.
A \emph{state} $\omega$ on $\mathcal{A}_\infty$ is a positive linear
functional with $\omega(\mathbf{1})=1$. States are automatically weak$^*$-continuous. A state $\omega$ is \emph{translation-invariant} if $\omega(\gamma_y(A))=\omega(A)$
for all $A\in\mathcal{A}_\infty$
and $y\in\Z^\nu$ (it is sufficient to demand this for all $A\in\mathcal{A}_\Lambda$ for all finite regions $\Lambda$).
If $\Lambda\in\Z^\nu$ is finite, there is a density matrix $\omega_\Lambda\in\mathcal{A}_\Lambda$ such that
\[
   \tr(\omega_\Lambda A)=\omega(A)\qquad\mbox{for all }A\in\mathcal{A}_\Lambda.
\]
This yields the following consistency condition: if $\Lambda\subset\Lambda'$ and $\Lambda'$ is finite, then $\omega_\Lambda={\rm Tr}_{\Lambda'\setminus\Lambda}\omega_{\Lambda'}$.
Conversely, every consistent family of density matrices defines a state on $\mathcal{A}_\infty$.

For translation-invariant states, the following definitions are crucial. To state them, we consider sequences of boxes (that is, hyperrectangles) $(\Lambda_n)_{n\in\N}$
with $\Lambda_n\subset \Lambda_{n+1}$
and with the property that for every $x\in\Z^\nu$ there is some $n\in\N$ with $x\in\Lambda_n$. Unless specified otherwise, all sequences of regions $\Lambda_n$ in the following
will be assumed to have these properties.

All logarithms are in base $e$, i.e.\ $\log(\exp(x))=x$.
\begin{definition}
\label{Def1}
Let $\omega$ be a translation-invariant state on $\mathcal{A}_\infty$. Then the following expressions exist:
\begin{itemize}
\item Energy density: $\displaystyle u(\omega):=\lim_{n\to\infty} \frac 1 {|\Lambda_n|} \tr(\omega_{\Lambda_n} H_{\Lambda_n})$,
\item entropy density: $\displaystyle s(\omega):= -\lim_{n\to\infty} \frac 1 {|\Lambda_n|} \tr(\omega_{\Lambda_n} \log \omega_{\Lambda_n})$.
\end{itemize}
Moreover, there is the state-independent quantity \emph{pressure}
\[
   p(\beta,\Phi):=\lim_{n\to\infty}\frac 1 {|\Lambda_n|} \log \tr \exp(-\beta H_{\Lambda_n})
\]
for all $\beta\geq 0$. It satisfies
\begin{equation}
   p(\beta,\Phi)=\sup\{s(\varphi)-\beta\, u(\varphi)\,\,|\,\, \varphi\mbox{ is any translation-invariant state on }\mathcal{A}_\infty\}.
   \label{eqVariationalPrinciple}
\end{equation}
\end{definition}

See~\cite{Simon} for more details. In the following, we consider Gibbs state on the {infinite} lattice. They are defined by any one of the following equivalent conditions.

\begin{definition}
\label{DefGibbs}
Let $\omega$ be a translation-invariant state on the quasi-local algebra $\mathcal{A}_\infty$ over $\Z^\nu$, with translation-invariant finite-range interaction $\Phi$,
and let $\beta>0$. Then the following conditions are equivalent:
\begin{itemize}
\item \textbf{Variational principle:} it holds $p(\beta,\Phi)=s(\omega)-\beta\, u(\omega)$, which is the maximal possible value according to~(\ref{eqVariationalPrinciple}).
\item \textbf{KMS condition} at inverse temperature $\beta$ (see~\cite{Araki74},
\item \textbf{Gibbs condition} at inverse temperature $\beta$ (see also~\cite{Araki74}).
\end{itemize}
If $\omega$ satisfies one of these equivalent conditions, we will call $\omega$ a \emph{Gibbs state} at inverse temperature $\beta$.
We say that Gibbs states are unique around inverse temperature $\beta$ for a given interaction $\Phi$ if there is an open interval
containing $\beta$ such that for every $\beta'$ in this interval, there is a unique (only one) Gibbs state at inverse temperature $\beta'$.
\end{definition}
Since we do not use the KMS and the Gibbs conditions, we do not explain them in detail here. We refer the reader to~\cite{Araki74} and~\cite{Simon}.

For what follows, we need to extend the notion of translation-invariance to finite regions. This is done in the obvious way.
Let $\mathcal{A}_n:=\mathcal{A}_{\Lambda_n}$, with $\Lambda_n$ a sequence
of boxes tending to infinity as $n\to\infty$ in the sense specified above.
Call an observable $A\in\mathcal{A}_n$ \emph{$\Lambda_n$-translation-invariant} if it
is translation-invariant with respect to periodic translations of $\Lambda_n$; that is, translations in which we regard $\Lambda_n$ as a torus.
In more detail, write $\Lambda_n$ as the product of intervals
\[
   \Lambda_n=[\lambda_1,\mu_1]\times [\lambda_2,\mu_2]\times\ldots\times[\lambda_\nu,\mu_\nu],
\]
where $\lambda_i,\mu_i\in\Z$, $\lambda_i\leq\mu_i$. The statement that $\Lambda_n$ tends to infinity means that all $\lambda_i\to - \infty$
and all $\mu_i\to +\infty$ as $n\to\infty$.
Define $\nu$ independent translations $(T_j)_{j=1,\ldots,\nu}$ for $x\in\Lambda_n$ by
\[
   T_j(x)\equiv T_j(x_1,\ldots,x_\nu)=(x_1,\ldots,x_{j-1},x_j\oplus 1, x_{j+1},\ldots,x_\nu),
\]
where
\[
   x_j\oplus 1 =\left\{
      \begin{array}{cl}
         x_j+1 & \mbox{if }x_j+1\leq\mu_j,\\
         \lambda_j & \mbox{otherwise}.
      \end{array}
   \right.
\]
We can interpret $T_j$ as a unitary operator, translating the computational basis vectors, constructed from the translation automorphisms $\gamma_y$.
An observable $A$ will be called $\Lambda_n$-translation-invariant if $T_j A T_j^\dagger = A$ for all $j=1,\ldots,\nu$.

We can also formalize this definition somewhat differently. Denote by $\mathbf{T}(\Lambda_n)$ the set of all periodic translations of $\Lambda_n$
into itself; in other words, regard $\Lambda_n$ as a torus, and $\mathbf{T}(\Lambda_n)$ as the set of translations on the torus.
These are arbitrary compositions of translations $T_j$. If $\alpha\in \Z^\nu$, then the periodic translation by vector $\alpha$ will be denoted
$T_\alpha\in\mathbf{T}(\Lambda_n)$; it equals $T_\alpha=\bigcirc_{j=1}^\nu T_j^{\alpha_j}$, where the circle denotes composition and the $T_j$ are mutually commuting.
Then an observable $A$ is $\Lambda_n$-translation-invariant if and only if $T A T^\dagger=A$ for all $T\in\mathbf{T}(\Lambda_n)$.

So far, we have defined $H_{\Lambda}$ for finite regions $\Lambda$ by summing up all interaction terms that are fully contained in $\Lambda$.
This is usually called the Hamiltonian with \emph{open boundary conditions}. Alternatively, one can consider periodic or other, more general
boundary conditions. We use the following definition.
\begin{definition}[Periodic and arbitrary boundary conditions]
\label{DefBC}
Let $\Phi$ be any finite-range translation-invariant interaction. A region $\Lambda\subset\Z^\nu$ is called \emph{large enough} if for every region $X$ with
$\Phi(X)\neq \emptyset$, there is $y\in\Z^\nu$ such that the translation $X+y$ is contained in $\Lambda$.

A \emph{choice of boundary conditions} is a map that assigns to every large enough, finite set $\Lambda\subset\Z^\nu$ a Hamiltonian $H_\Lambda^{BC}$ such
that
\begin{equation}
   \lim_{n\to\infty} \frac{\left\| H_{\Lambda_n}^{BC}-H_{\Lambda_n}\right\|_\infty}{|\Lambda_n|} = 0
   \label{eqAsymptVanish}
\end{equation}
for every sequence of boxes $(\Lambda_n)_{n\in\N}$ that tends to infinity in the sense specified above.

A particularly important example of choice of boundary conditions is given by \emph{periodic boundary conditions}, with corresponding Hamiltonians denoted $H_{\Lambda}^p$.
Following~\cite{Simon}, we define it as
\[
   H_\Lambda^p:={\sum_{X\cap\Lambda\neq\emptyset}}' T_{-\alpha}\left[\gamma_\alpha(\Phi(X))\right]T_{-\alpha}^\dagger,
\]
where $\alpha\in\Z^\nu$ denotes any vector that translates $X$ into $\Lambda$, i.e.\ $X+\alpha\in\Lambda$, and the prime on the sum indicates that regions
$X,X'$ with $X'_i=X_i+n_i a_i$ for all $i$, where $a_i$ is the $i$-th sidelength of the boxes $\Lambda$, and $n_i\in\Z$, are not included twice in the sum,
but only once (i.e.\ only $X$ or $X'$ will be included).
\end{definition}

The fact that we are demanding that regions $\Lambda$ are large enough implies that our definition of $H_\Lambda^p$ agrees with both of what
Simon~\cite{Simon} calls $H_\Lambda^{p,1}$ and $H_\Lambda^{p,2}$. To see that $H_\Lambda^p$ satisfies~(\ref{eqAsymptVanish}), denote by $\partial\Lambda$
the discrete boundary of $\Lambda$, that is
\begin{equation}
   \partial\Lambda:=\left\{x\in\Lambda\,\,|\,\, \exists y\in\Z^\nu\setminus\Lambda:\enspace {\rm dist}(x,y)\leq 1\right\},
   \label{eqDefBoundary}
\end{equation}
where ${\rm dist}(x,y):=\max_i |y_i-x_i|$. Suppose $x\in\Z^\nu$ is any point.
Since $\Phi$ has finite range and is translation-invariant, there is some finite integer $\kappa\in\N$ equal to the number of finite regions $X$ that contain $x$ and have
$\Phi(X)\neq 0$.
This number is the same for every $x\in\Z^\nu$. Also, $\|\Phi\|:=\max_X \left\|\Phi(X)\right\|_\infty$ is finite. Thus
\[
   \left\|H_{\Lambda_n}^p - H_{\Lambda_n}\right\|_\infty \leq \sum_{X\cap\Lambda_n\neq\emptyset,\enspace X\not\subset \Lambda_n} \|\Phi(X)\|_\infty
   \leq \sum_{x\in\partial\Lambda_n} \kappa\|\Phi\| = \kappa\|\Phi\|\,|\partial\Lambda_n|.
\]
Since $|\partial\Lambda_n|/|\Lambda_n|$ tends to zero for $n\to\infty$, this proves~(\ref{eqAsymptVanish}). We can write $H_{\Lambda_n}^p$ in an alternative
form. Given $\Lambda_n$, denote by $X_1,\ldots, X_N$ subsets of $\Lambda_n$ with the property that no $X_i$ is a periodic
translation of any other $X_j$, and such that all subsets of $\Lambda_n$ can be generated by periodically translating some $X_i$. For example, if
$\Lambda={0,1,2,3}$ on a one-dimensional lattice, then $X_1=\{0\}$, $X_2=\{0,1\}$, $X_3=\{0,2\}$, $X_4=\{0,1,2\}$ and $X_5=\{0,1,2,3\}$ is a possible
choice of those sets. Then we have
\[
   H_{\Lambda_n}^p =\sum_{i=1}^N \sum_{T\in\mathbf{T}(\Lambda_n)} T \Phi(X_i) T^\dagger,
\]
and from the representation it becomes clear that $H_\Lambda^p$ is $\Lambda_n$-translation-invariant.

Note that we do not consider what Simon calls ``external boundary conditions''.

In the following, we will frequently use that the energy density does not depend on the choice of boundary conditions; that is, if $(\tau_n)_{n\in\N}$ is
an arbitrary sequence of states on $\mathcal{A}_n$, then
\[
   \lim_{n\to\infty}\left( \frac{\tr(\tau_n H_{\Lambda_n})}{|\Lambda_n|} - \frac{\tr(\tau_n H_{\Lambda_n}^{BC})}{|\Lambda_n|}\right)=0.
\]
This is because $|\tr(\tau_n H_{\Lambda_n})-\tr(\tau_n H_{\Lambda_n}^{BC})|\leq \|H_{\Lambda_n}-H_{\Lambda_n}^{BC}\|_\infty$.

\begin{lemma}
\label{LemMainLemma}
Let $(\tau_n)_{n\in\N}$ be a sequence of density matrices in $\mathcal{A}_n$ such that every $\tau_n$ is $\Lambda_n$-translation-invariant.
For every $m\in\N$, consider the sequence of states $(\rho_n^{(m)})_{n\in\N}\in\mathcal{A}_m$, defined for $n\geq m$ by
\[
   \rho_n^{(m)}:=\Tr_{\Lambda_n\setminus\Lambda_m} \tau_n.
\]
Define $L^{(m)}$ as the set of all limit points of the sequence $(\rho_n^{(m)})_{n\in\N}$, and $L$ as the set of all possible
sequences $(\sigma_m)_{m\in\N}$ with $\sigma_m\in L^{(m)}$ and $\sigma_{m-1}=\Tr_{\Lambda_m\setminus\Lambda_{m-1}}\sigma_m$.
Then $L$ is not empty, and every element of $L$ defines a translation-invariant state on the quasi-local algebra. Additionally, if
$\beta\geq 0$ is such that
\begin{equation}
   \liminf_{n\to\infty}\frac 1 {|\Lambda_n|}\left(\strut S(\tau_n)-\beta\, \tr(\tau_n H_{\Lambda_n})\right)\geq p(\beta,\Phi),
   \label{eqIneqEq}
\end{equation}
then every state $\omega\in L$ is a Gibbs state at inverse temperature $\beta$, and we have equality in~(\ref{eqIneqEq}). Furthermore, if $L$ contains only a single element $\omega_\beta$,
then $\displaystyle \lim_{n\to\infty} \frac 1 {|\Lambda_n|} \tr(\tau_n H_{\Lambda_n})=u(\omega_\beta)$ and $\lim_{n\to\infty}\frac 1 {|\Lambda_n|} S(\tau_n)=s(\omega_\beta)$.
\end{lemma}
\proof
First, we observe that every element of $L^{(m)}$ generates an element of $L^{(m-1)}$ by taking the partial trace over $\Lambda_m\setminus\Lambda_{m-1}$; that is,
\begin{equation}
   \label{eqProp1}
   \Tr_{\Lambda_m\setminus\Lambda_{m-1}}L^{(m)}\subseteq L^{(m-1)}.
\end{equation}
Similarly, suppose that $\rho\in L^{(m-1)}$. By definition, this means that there is a strictly increasing sequence of natural numbers $(n_k)_{k\in\N}$ such
that $\rho_{n_k}^{(m-1)}\stackrel{k\to\infty}\longrightarrow \rho$. Now consider the sequence $\rho_{n_k}^{(m)}$; since $m$ is fixed, it is a bounded sequence
on a finite-dimensional vector space. By Bolzano-Weierstra\ss, it must have at least one limit point $\bar \rho$. Since $\rho_{n_k}^{(m-1)}=\Tr_{\Lambda_m\setminus
\Lambda_{m-1}}\rho_{n_k}^{(m)}$, we obtain $\rho=\Tr_{\Lambda_m\setminus\Lambda_{m-1}} \bar \rho$. We have thus proven that
\begin{equation}
   \label{eqProp2}
   \mbox{for every }\rho\in L^{(m-1)},\mbox{ there is }\bar\rho\in L^{(m)}\mbox{ such that }\rho=\Tr_{\Lambda_m\setminus\Lambda_{m-1}} \bar \rho.
\end{equation}
Furthermore, by Bolzano-Weierstra\ss, $L^{(1)}$ is non-empty. Combining the properties~(\ref{eqProp1}) and (\ref{eqProp2}), we obtain
$\Tr_{\Lambda_m\setminus\Lambda_{m-1}}L^{(m)}= L^{(m-1)}$ as an equality between non-empty sets. This is 
sketched in Figure~\ref{fig_tree}, where we plot elements of $L^{(m)}$ as dots, with an edge connecting two dots if the left element (in $L^{(m-1)}$) is the partial
trace of the right one (in $L^{(m)}$). Wandering from left to the right, no path will lead to a dead end; furthermore, every
point can be reached this way by starting with some element in $L^{(1)}$.
\begin{figure}[!hbt]
\begin{center}
\includegraphics[angle=0, width=5cm]{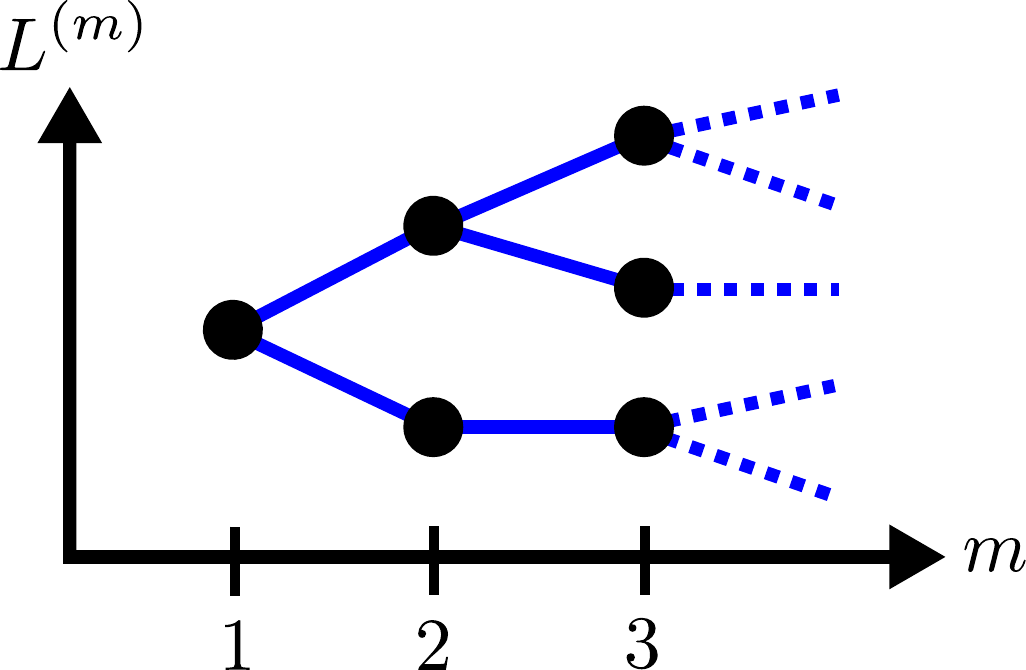}
\caption{Schematics of the sequence $L^{(m)}$ of limit points, as defined in the proof of Lemma~\ref{LemMainLemma}. Note that $|L^{(1)}|>1$ is possible.
}
\label{fig_tree}
\end{center}
\end{figure}
Thus, there is at least one path that starts with some element $\sigma_1\in L^{(1)}$ and extends to infinity -- that is, a sequence $(\sigma_m)_{m\in\N}$
with $\sigma_{m-1}=\Tr_{\Lambda_m\setminus\Lambda_{m-1}}\sigma_m$. $L$ is the set of all these paths and hence not empty.
Every $\omega\in L$ can be interpreted as a state: for any finite region $\Lambda\subset \Z^\nu$, take the smallest $n$ such that $\Lambda\subset \Lambda_n$,
and set $\omega_\Lambda:=\Tr_{\Lambda_n\setminus\Lambda}\omega_m$. This defines a consistent family of density matrices, hence a state on the quasi-local algebra.

Now let $\omega\in L$ be any state. We claim that $\omega$ is translation-invariant. It will be sufficient to show the invariance equation
$\omega(\gamma_y(A))=\omega(A)$ for observables $0\leq A \leq\mathbf{1}$ and translations $\gamma_{\delta_j}$,
where $\delta_j=(0,\ldots,0,\underbrace{1}_j,0,\ldots,0)$. So let $\Lambda\subset\Z^\nu$ be finite, $A\in\mathcal{A}_\Lambda$ an observable with $0\leq A \leq\mathbf{1}$, and
$j\in\{1,\ldots,d\}$; set $\gamma:=\gamma_{\delta_j}$. Choose $m$ large enough such that $\Lambda\subseteq \Lambda_m$ and $\Lambda+\delta_j\subseteq\Lambda_m$.
Let $\varepsilon>0$ be arbitrary. Since $\omega_{\Lambda_m}\in L^{(m)}$, there is some $n\geq m$ such that
\[
   \left\| \Tr_{\Lambda_n\setminus\Lambda_m} \tau_n - \omega_{\Lambda_m}\right\|_1 < \varepsilon.
\]
The effect of the translation of the observable $A$ in the region $\Lambda_n$ can be written
\[
   \gamma(A)\otimes \mathbf{1}_{\Lambda_n\setminus(\Lambda+\delta_j)} = T_j(A\otimes\mathbf{1}_{\Lambda_n\setminus\Lambda})T_j^\dagger,
\]
where $T_j\in\mathcal{A}_{\Lambda_n}$ is the unitary translation operator in $\Lambda_n$ as defined shortly before Definition~\ref{DefBC}.
Using this, we see that there are two real numbers $\Delta,\Delta'\in\R$ with $|\Delta|<2\varepsilon$, $|\Delta'|<2\varepsilon$ such that
\begin{eqnarray*}
   \omega(\gamma(A))&=&\tr\left[\omega_{\Lambda_m}\left(\gamma(A)\otimes \mathbf{1}_{\Lambda_m\setminus (\Lambda+\delta_j)}\right)\right]
   =\tr\left[ \left(\Tr_{\Lambda_n\setminus\Lambda_m}\tau_n\right)\left(\gamma(A)\otimes \mathbf{1}_{\Lambda_m\setminus (\Lambda+\delta_j)}\right)\right]+\Delta \\
   &=&\tr\left[ \tau_n\left(\gamma(A)\otimes\mathbf{1}_{\Lambda_n\setminus(\Lambda+\delta_j)}\right)\right]+\Delta
   =\tr\left[\tau_n T_j(A\otimes\mathbf{1}_{\Lambda_n\setminus\Lambda})T_j^\dagger\right]+\Delta\\
   &=& \tr\left[ T_j^\dagger \tau_n T_j (A\otimes\mathbf{1}_{\Lambda_n\setminus\Lambda})\right]+\Delta
   = \tr\left[ \tau_n (A\otimes\mathbf{1}_{\Lambda_n\setminus\Lambda})\right]+\Delta  \\
   &=& \tr\left[\left(\Tr_{\Lambda_n\setminus\Lambda_m}\tau_n\right)(A\otimes \mathbf{1}_{\Lambda_m\setminus\Lambda})\right]+\Delta
   =\tr\left[\omega_{\Lambda_m} (A\otimes\mathbf{1}_{\Lambda_m\setminus\Lambda})\right]+\Delta+\Delta' \\
   &=& \omega(A)+\Delta+\Delta'.
\end{eqnarray*}
Since $\varepsilon>0$ was arbitrary, this proves translation-invariance of $\omega$.

In particular, every $\omega\in L$ has
a well-defined entropy rate $s(\omega)$; what can we say about it? Fix $m\in\N$, and let
$\omega_m:=\omega_{\Lambda_m}$. Remember that $\rho_n^{(m)}=\Tr_{\Lambda_n\setminus\Lambda_m} \tau_n$.
Since $\omega_m\in L^{(m)}$, there exists a sequence $(n_k)_{k\in\N}$ such that
\[
   \rho_{n_k}^{(m)}\stackrel{k\to\infty}\longrightarrow \omega_m.
\]
Fix $k\in\N$. We now
decompose $\Lambda_{n_k}$ into a disjoint union of boxes,
where each box is a translate of $\Lambda_m$ (where we consider translations as in the notion of $\Lambda_{n_k}$-translation invariance -- that is, we regard
$\Lambda_{n_k}$ as a torus). In general, this cannot be done perfectly, but there will be some remaining part of $\Lambda_{n_k}$ not covered by a translate
of $\Lambda_m$. To spell out the details, let $a_m^{(1)},\ldots,a_m^{(\nu)}$ denote the sidelengths of the box $\Lambda_m$, and $a_{n_k}^{(1)},\ldots,a_{n_k}^{(\nu)}$
the sidelengths of $\Lambda_{n_k}$. Write
\[
   a_{n_k}^{(i)}=\ell_i\cdot a_m^{(i)}+j_i,\qquad \mbox{where }0\leq j_i< a_m^{(i)}.
\]
Clearly, all $\ell_i$ tend to infinity for $k\to\infty$ on fixed $m$.
Let $N_k:=\ell_1\cdot \ell_2\cdot\ldots \cdot\ell_\nu$, then there are $N_k$ translates $\Lambda_m^{(1)},\Lambda_m^{(2)},\ldots,\Lambda_m^{(N_k)}$ of $\Lambda_m$
and the remainder $\Lambda_{\rm rem}\subset\Lambda_{N_k}$, all of them pairwise disjoint, such that
\[
   \Lambda_{n_k} = \bigcup_{i=1}^{N_k} \Lambda_m^{(i)} \cup \Lambda_{\rm rem}.
\]
We have $|\Lambda_{\rm rem}|=|\Lambda_{n_k}|-N_k\cdot a_m^{(1)}\cdot\ldots \cdot a_m^{(\nu)}$, hence
\begin{eqnarray*}
   \frac{|\Lambda_{\rm rem}|}{N_k} &=& \frac{|\Lambda_{n_k}|}{N_k} - a_m^{(1)}\cdot\ldots\cdot a_m^{(\nu)}
   =\frac{a_{n_k}^{(1)}\cdot a_{n_k}^{(2)}\cdot\ldots\cdot a_{n_k}^{(\nu)}}{\ell_1 \ell_2\ldots \ell_\nu} - a_m^{(1)} a_m^{(2)}\ldots a_m^{(\nu)}\\
   &<& \frac{(\ell_1+1)a_m^{(1)} \cdot\ldots\cdot (\ell_\nu+1)a_m^{(\nu)}}{\ell_1\ell_2\ldots\ell_\nu}-a_m^{(1)} a_m^{(2)}\ldots a_m^{(\nu)}
   =|\Lambda_m|\cdot\left[
      \left(1+\frac 1 {\ell_1}\right)\ldots \left(1+\frac 1 {\ell_\nu}\right)-1
   \right]\\
   &\stackrel{k\to\infty}\longrightarrow& 0.
\end{eqnarray*}
As a consequence, we also obtain
\[
   \lim_{k\to\infty}\frac{|\Lambda_{n_k}|}{N_k} = \lim_{k\to\infty}\left( |\Lambda_m| + \frac{|\Lambda_{\rm rem}|}{N_k}\right)=|\Lambda_m|.
\]
Since $\tau_{n_k}$ is $\Lambda_{n_k}$-translation-invariant, its marginals on all the boxes $\Lambda_m^{(i)}$ are equal,
that is, equal to $\rho_{n_k}^{(m)}$. Due to subadditivity of von Neumann entropy $S$, we have
\[
   S(\tau_{n_k})\leq N_k S(\rho_{n_k}^{(m)})+S(\tau_{\Lambda_{\rm rem}})\leq N_k S(\rho_{n_k}^{(m)}) + |\Lambda_{\rm rem}|\cdot\log d,
\]
where $d$ is the single-site Hilbert space dimension. Thus, we obtain
\begin{equation}
   S(\omega_m)= \lim_{k\to\infty} S(\rho_{n_k}^{(m)}) \geq \limsup_{k\to\infty} \frac 1 {N_k} \left[ S(\tau_{n_k})-|\Lambda_{\rm rem}|\cdot\log d\right]
   =\limsup_{k\to\infty} \frac 1 {N_k} S(\tau_{n_k})
   =|\Lambda_m| \limsup_{k\to\infty} \frac 1 {|\Lambda_{n_k}|} S(\tau_{n_k}).
   \label{eqLimSupS}
\end{equation}
Furthermore, we can estimate the energy expectation value of $\omega_m$ as follows. Define a Hamiltonian $H_{\Lambda_{n_k}}^{(m)}$ on $\Lambda_{n_k}$
by ``switching off'' all interaction terms that are not fully contained in one of the $\Lambda_m^{(i)}$, that is,
\[
   H_{\Lambda_{n_k}}^{(m)}:=\sum_{i=1}^{N_k} \sum_{X\subset \Lambda_m^{(i)}} \Phi(X).
\]
We can estimate the norm difference of $H_{\Lambda_{n_k}}^{(m)}$ and $H_{\Lambda_{n_k}}$ as follows. All missing terms are either fully contained in
$\Lambda_{\rm rem}$, or act across the boundary of some $\Lambda_m^{(i)}$. With the boundary $\partial \Lambda_m^{(i)}$ as defined in~(\ref{eqDefBoundary}),
we obtain $\lim_{m\to\infty}|\partial \Lambda_m|/|\Lambda_m|=0$, and due to finite-range of the interaction $\Phi$, there are constants $c_1,c_2>0$ such that
\[
   \left\| H_{\Lambda_{n_k}} - H_{\Lambda_{n_k}}^{(m)}\right\|_\infty \leq c_1 |\Lambda_{\rm rem}| + c_2\sum_{i=1}^{N_k} |\partial \Lambda_m^{(i)}|
   =c_1 |\Lambda_{\rm rem}| + c_2 N_k |\partial \Lambda_m|.
\]
By construction and translation-invariance of $\Phi$, we have
\[
   \tr\left(\tau_{n_k} H_{\Lambda_{n_k}}^{(m)}\right)=N_k\, \tr\left(\rho_{n_k}^{(m)} H_{\Lambda_m}\right).
\]
Combining these identities, we get
\begin{eqnarray}
   \frac 1 {|\Lambda_m|}\tr(\omega_m H_{\Lambda_m}) &=& \frac 1 {|\Lambda_m|}\lim_{k\to\infty} \tr\left(\rho_{n_k}^{(m)} H_{\Lambda_m}\right)
   =\frac 1 {|\Lambda_m|} \lim_{k\to\infty} \frac 1 {N_k} \tr\left(\tau_{n_k} H_{\Lambda_{n_k}}^{(m)}\right)\nonumber \\
   &\leq& \frac 1 {|\Lambda_m|} \limsup_{k\to\infty}\frac 1 {N_k}\left(\strut \tr(\tau_{n_k} H_{\Lambda_{n_k}})+c_1|\Lambda_{\rm rem}|+c_2 N_k |\partial \Lambda_m|\right)\nonumber \\
   &=& \frac{c_2 |\partial \Lambda_m|}{|\Lambda_m|} + \limsup_{k\to\infty}\frac 1 {|\Lambda_{n_k}|} \tr\left(\tau_{n_k} H_{\Lambda_{n_k}}\right).
   \label{eqFirstEnergyDensity}
\end{eqnarray}
Since $\liminf(a_n+b_n)\leq\liminf a_n+\limsup b_n$, we obtain
\begin{eqnarray*}
   \frac 1 {|\Lambda_m|} \left( S(\omega_m)-\beta\, \tr(\omega_m H_{\Lambda_m})\right) &\geq& \limsup_{k\to\infty} \frac 1 {|\Lambda_{n_k}|} S(\tau_{n_K})
   -\beta\limsup_{k\to\infty}\frac 1 {|\Lambda_{n_k}|} \tr\left(\tau_{n_k} H_{\Lambda_{n_k}}\right)-\beta c_2 \frac{|\partial \Lambda_m|}{|\Lambda_m|}\\
   &\geq& \liminf_{k\to\infty} \frac 1 {|\Lambda_{n_k}|} \left( S(\tau_{n_k})-\beta\, \tr(\tau_{n_k} H_{\Lambda_{n_k}})\right)-\beta c_2 \frac{|\partial \Lambda_m|}{|\Lambda_m|}\\
   &\geq& \liminf_{n\to\infty} \frac 1 {|\Lambda_n|} \left( S(\tau_n)-\beta\, \tr(\tau_n H_{\Lambda_n})\right)-\beta c_2 \frac{|\partial \Lambda_m|}{|\Lambda_m|}.
\end{eqnarray*}
Taking the limit $m\to\infty$ finally shows that $s(\omega)-\beta\, u(\omega)\geq p(\beta,\Phi)$.
Since $\omega$ is translation-invariant, we must have equality, and $\omega$ must be a Gibbs state.

Now suppose that $L$ contains only a single element, then so does $L^{(m)}$; hence $\rho_n^{(m)}$ converges for $n\to\infty$, and we can choose
the convergent subsequence to be $n_k=k$. Repeating the calculation of~(\ref{eqFirstEnergyDensity}) with inequality in both directions yields
\begin{eqnarray*}
   \liminf_{k\to\infty} \frac 1 {|\Lambda_k|} \tr(\tau_k H_{\Lambda_k}) &\geq& \frac 1 {|\Lambda_m|} \tr(\omega_m H_{\Lambda_m})-\frac{c_2 |\partial \Lambda_m|}{|\Lambda_m|},\\
   \limsup_{k\to\infty} \frac 1 {|\Lambda_k|} \tr(\tau_k H_{\Lambda_k}) &\leq& \frac 1 {|\Lambda_m|} \tr(\omega_m H_{\Lambda_m})+\frac{c_2 |\partial \Lambda_m|}{|\Lambda_m|},
\end{eqnarray*}
By taking the limit $m\to\infty$ of the right-hand side, we obtain $\lim_{n\to\infty} \frac 1 {|\Lambda_n|} \tr(\tau_n H_{\Lambda_n})=u(\omega_\beta)$.
Then it follows directly from~(\ref{eqIneqEq}) that $\liminf_{n\to\infty} \frac 1 {|\Lambda_n|} S(\tau_n)=s(\omega_\beta)$. Furthermore, (\ref{eqLimSupS}) shows
that $\limsup_{n\to\infty} \frac 1 {|\Lambda_n|} S(\tau_n)\leq s(\omega_\beta)$, hence $\lim_{n\to\infty} \frac 1 {|\Lambda_n|} S(\tau_n)= s(\omega_\beta)$.
\qed

We can always define a \emph{maximally mixed state} $\omega$ on the quasi-local algebra $\mathcal{A}_\infty$, by defining its local density matrix for finite
$\Lambda\subset\Z^\nu$ as $\omega_\Lambda:=\mathbf{1}_\Lambda/ d^{|\Lambda|}$. It is easy to check that this is a consistent family of density matrices,
defining a translation-invariant state on $\mathcal{A}_\infty$. According to Definition~\ref{Def1}, its energy density exists; it is $u(\omega)=\lim_{n\to\infty} \tr(H_{\Lambda_n})/d^{|\Lambda_n|}$.
This fact will be used in the following lemma. To state that lemma, we have to assume that the interaction $\Phi$ does not vanish -- and, in addition, that it
is not \emph{physically equivalent} to zero. An example would be an interaction in one dimension (i.e.\ $\nu=1$) with $\Phi(\{1,2\})=-\Phi(\{1\})\otimes\mathbf{1}_2$,
such that the resulting Hamiltonian is zero up to boundary terms. For a formal definition of physical equivalence and a further example see~\cite{Simon}.
Note also that $\Phi$ is physically equivalent to zero if and only if $p(\beta,\Phi)=\log d$ for all $\beta\geq 0$, which is the same value as for $\Phi=0$.

\begin{lemma}
\label{LemVariational}
Let $\Phi$ be an interaction which is not physically equivalent to zero, with ground state energy density
$u_{\min}(\Phi)=\lim_{n\to\infty} \lambda_{\min}(H_{\Lambda_n})/|\Lambda_n| =-\lim_{\beta\to\infty}\beta^{-1} p(\beta,\Phi)$
and infinite temperature energy density $u_{\max}(\Phi):=\lim_{n\to\infty} \tr(H_{\Lambda_n})/(|\Lambda_n| d^{|\Lambda_n|})$. Then, for every $u\in (u_{\min}(\Phi),u_{\max}(\Phi)]$, there
exists a unique $\beta\equiv \beta(u)\geq 0$ such that there is at least one Gibbs state $\omega$ at inverse temperature $\beta$ with energy density $u(\omega)=u$. Its
entropy density is $s(\omega)=s(u):=p(\beta(u),\Phi)+u\,\beta(u)$, and this is the maximal possible entropy density of any translation-invariant state with energy density $u$.
\end{lemma}
\proof
These statements are proven in~\cite{Simon}; uniqueness of $\beta(u)$ can be seen as follows.
If $\Phi$ is not physically equivalent to zero, then the function $\beta\mapsto p(\beta,\Phi)$ is strictly convex, see~\cite[p.\ 349 and Thm.\ II.1.5]{Simon}.
Consider any translation-invariant state $\omega$; it defines an affine-linear map $\beta\mapsto s(\omega)-\beta\, u(\omega)=:\ell_\omega(\beta)$. According to~(\ref{eqVariationalPrinciple}),
the line $\ell_\omega$ lies completely on or below of the graph of $p$; that is, $\ell_\omega(\beta)\leq p(\beta,\Phi)$ for all $\beta$. According to Definition~\ref{DefGibbs},
it is a Gibbs state if and only if $\ell_\omega$ touches the graph of $p$; that is, if there is some $\beta$ such that $\ell_\omega(\beta)=p(\beta)$. If we are given some value of $u$,
then every translation-invariant state with this energy density has a corresponding line $\ell_\omega$ with slope $(-u)$. Consider all those lines. Then only one of them can
touch the graph of $p$, and it can do so in only one point, due to the strict convexity of $p$. The $\beta$-value of the unique touching point is then $\beta(u)$.
\qed

Now we have all ingredients to prove our main theorem on the equivalence of ensembles.
\begin{theorem}[Equivalence of ensembles]
\label{TheEquivalence}
Let $(\tau_n)_{n\in\N}$ be a sequence of $\Lambda_n$-translation-invariant states on $\mathcal{A}_n$, let $\beta\geq 0$, and let $\Phi$ be a translation-invariant finite-range interaction
which is not physically equivalent to zero, and for which there is a unique Gibbs state $\omega_\beta$ at inverse temperature $\beta$. Suppose that
\[
   \liminf_{n\to\infty}\frac 1 {|\Lambda_n|}\left(\strut S(\tau_n)-\beta\, \tr(\tau_n H_{\Lambda_n})\right)\geq p(\beta,\Phi),
\]
then we have equality in this expression, and
\begin{equation}
   \lim_{n\to\infty}{\rm Tr}_{\Lambda_n\setminus\Lambda_m} \tau_n = (\omega_\beta)_{\Lambda_m}
   \label{eqThermalLimit}
\end{equation}
for every $m\in\N$. Furthermore, we have
\[
   \lim_{n\to\infty}\frac 1 {|\Lambda_n|} S(\tau_n)=s(\omega_\beta),\qquad \lim_{n\to\infty}\frac 1 {|\Lambda_n|} \tr(\tau_n H_{\Lambda_n})=u(\omega_\beta),
\]
and
\begin{equation}
   \lim_{n\to\infty}\left\| {\rm Tr}_{\Lambda_n\setminus\Lambda_m} \tau_n - {\rm Tr}_{\Lambda_n\setminus\Lambda_m}
   \frac{\exp(-\beta H_{\Lambda_n}^p)} {Z_n}\right\|_1=0,
   \label{eqEquiv2}
\end{equation}
where $Z_n=\tr(\exp(-\beta H_{\Lambda_n}^p))$, and $H_{\Lambda_n}^p$ is the Hamiltonian on $\Lambda_n$ with periodic boundary conditions.
If the lattice dimension is $\nu=1$, then $H_{\Lambda_n}^p$ in~(\ref{eqEquiv2}) can be replaced by $H_{\Lambda_n}$, the Hamiltonian with open boundary conditions.
Furthermore, if Gibbs states are unique around inverse temperature $\beta$, define $\beta_n^{BC}$ as the solution of the equation
$\displaystyle \frac 1 {|\Lambda_n|}\tr\left( H_{\Lambda_n}^{BC}\,
\frac{\exp(-\beta_n^{BC} H_{\Lambda_n}^p)}{Z_n}\right)=u_n$, where $BC$ denotes an arbitrary fixed choice of boundary conditions, and $(u_n)_{n\in\N}$ is
an arbitrary sequence with $\lim_{n\to\infty}u_n=u(\omega_\beta)$. Then $\lim_{n\to\infty}\beta_n^{BC}=\beta$, and
\begin{equation}
   \lim_{n\to\infty}\left\| {\rm Tr}_{\Lambda_n\setminus\Lambda_m} \tau_n - {\rm Tr}_{\Lambda_n\setminus\Lambda_m}
   \frac{\exp(-\beta_n^{BC} H_{\Lambda_n}^p)} {Z_n}\right\|_1=0.
   \label{eqEquivFinite}
\end{equation}
\end{theorem}
\proof
Set $\rho_n^{(m)}:={\rm Tr}_{\Lambda_n\setminus\Lambda_m} \tau_n$, and define $L^{(m)}$ and $L$ exactly as in the statement of Lemma~\ref{LemMainLemma}.
Since there is only one Gibbs state $\omega_\beta$ at inverse temperature $\beta$, Lemma~\ref{LemMainLemma} implies that $L=\{\omega_\beta\}$, and
so $L^{(m)}=(\omega_\beta)_{\Lambda_m}$ for all $m\in\N$. In other words, for every $m$, the state $(\omega_\beta)_{\Lambda_m}$ is the unique
limit point of the sequence $(\rho_n^{(m)})_{n\in\N}$, and thus the limit of this sequence. This proves the first identity.
To infer the second identity, eq.~(\ref{eqEquiv2}), either apply~\cite[Thm.\ IV.2.12]{Simon}, or note that $\tau'_n:=\exp(-\beta H_{\Lambda_n}^p)/Z_n$
maximizes the functional $\rho\mapsto S(\rho)-\beta\, \tr(H_{\Lambda_n}^p\rho)$, thus
\begin{eqnarray*}
   \liminf_{n\to\infty}\frac 1 {|\Lambda_n|}\left(S(\tau'_n) - \beta\, \tr(\tau'_n H_{\Lambda_n})\right)&=&
   \liminf_{n\to\infty}\frac 1 {|\Lambda_n|}\left(S(\tau'_n) - \beta\, \tr(\tau'_n H_{\Lambda_n}^p)\right) \\
   &\geq&  \liminf_{n\to\infty}\frac 1 {|\Lambda_n|}\left(S((\omega_\beta)_{\Lambda_n}) - \beta\, \tr((\omega_\beta)_{\Lambda_n} H_{\Lambda_n}^p)\right) \\
   &=& s(\omega_\beta)-\beta\, u(\omega_\beta)=p(\beta,\Phi).
\end{eqnarray*}
Thus $\lim_{n\to\infty}(\Tr_{\Lambda_n\setminus\Lambda_m} \tau_n - \Tr_{\Lambda_n\setminus\Lambda_m} \tau'_n) = (\omega_\beta)_{\Lambda_m}
-(\omega_\beta)_{\Lambda_m}=0$. Note that this also shows that $\lim_{n\to\infty}\frac 1 {|\Lambda_n|} \tr(\tau'_n H_{\Lambda_n})=u(\omega_\beta)$.
In the case of lattice dimension $\nu=1$, apply the fact that
in this case, the local Gibbs state $\exp(-\beta H_{\Lambda_n})/Z_n$ weakly converges to the unique global Gibbs state in the limit $n\to\infty$, as
shown in~\cite{Araki69}.

It remains to prove~(\ref{eqEquivFinite}). To this end, use the notation $Z_n(\beta):=\tr (\exp(-\beta H_{\Lambda_n}^p))$, and $\tau'_n:=\exp(-\beta_n^{BC} H_{\Lambda_n}^p)/
Z_n(\beta_n^{BC})$. First we have to show that $\beta_n^{BC}$ is well-defined for $n$ large enough and that it is a bounded sequence.
Set $\rho(\beta'):=\exp(-\beta' H_{\Lambda_n}^p)/Z_n(\beta')$ for $\beta'\geq 0$. Choose $\beta_0,\beta_1\in\R$ such
that $0<\beta_0<\beta<\beta_1$, and such that the Gibbs states at inverse temperatures $\beta_0$ and $\beta_1$ are unique. Then the previous results show that
$\lim_{n\to\infty} \frac 1 {|\Lambda_n|} \tr\left(H_{\Lambda_n}^{BC}\rho(\beta_i)\right)=u_i$ for $i=0,1$, where $u_i:=u(\omega_{\beta_i})$.
It follows $u_0>u>u_1$, and thus for $n$ large enough, we have $\frac 1 {|\Lambda_n|} \tr(H_{\Lambda_n}^{BC} \rho(\beta_0))>u_n
>\frac 1 {|\Lambda_n|} \tr(H_{\Lambda_n}^{BC} \rho(\beta_1))$, so $\beta_0<\beta_n^{BC}<\beta_1$ for $n$ large enough; in particular,
a solution $\beta_n^{BC}$ can be found in the interval $(\beta_0,\beta_1)$. Moreover, since $\beta_0$ and $\beta_1$ can be chosen
arbitrarily close to $\beta$, this proves that $\lim_{n\to\infty}\beta_n^{BC}=\beta$.
Direct calculation shows that $S(\tau'_n)=\log Z_n(\beta_n^{BC})+\beta_n^{BC} u_n |\Lambda_n|$, thus
\begin{eqnarray*}
   \liminf_{n\to\infty}\frac 1 {|\Lambda_n|} \left(S(\tau'_n)-\beta\, \tr(\tau'_n H_{\Lambda_n})\right)
   &=& \liminf_{n\to\infty} \frac 1 {|\Lambda_n|} \left(\log Z_n(\beta_n^{BC})+\beta_n^{BC}\, u_n |\Lambda_n|-\beta\,
   \tr(\tau'_n H_{\Lambda_n}^{BC})\right)\\
   &=& \liminf_{n\to\infty}\left( (\beta_n^{BC}-\beta)u_n + \frac 1 {|\Lambda_n|} \log Z_n(\beta_n^{BC})\right)\\
   &\geq& \liminf_{n\to\infty}\frac 1 {|\Lambda_n|} \left(\log Z_n(\beta)-|\log Z_n(\beta_n^{BC})-\log Z_n(\beta)|\right)\\
   &\geq& p(\beta,\Phi)-\limsup_{k\to\infty}\frac{|\beta_n^{BC}-\beta|\, \|H_{\Lambda_n}^p\|_{\infty}}{|\Lambda_n|} = p(\beta,\Phi),
\end{eqnarray*}
where we have used that $|\log Z_n(\beta_n^{BC})-\log Z_n(\beta)|\leq |\beta_n^{BC}-\beta|\cdot\|H_{\Lambda_n}^p\|_\infty$, see~\cite[Lemma II.2.2Q]{Simon}.
This shows that $\lim_{n\to\infty}\Tr_{\Lambda_n\setminus\Lambda_m} \tau'_n = (\omega_\beta)_{\Lambda_m}$.
Combining this with~(\ref{eqThermalLimit}) proves~(\ref{eqEquivFinite}).
\qed

In order to obtain some concrete instances of this equivalence of ensembles result, we need a series of lemmas. The first one
is given in~\cite[Thm.\ IV.2.14]{Simon}, though with typos; see also~\cite{Lima1,Lima2}, and for newer results on equivalence of ensembles,
see~\cite{deRoeck}. Since the lemma is crucial for our paper, we give the proof for completeness, {translating the proof of~\cite[Thm.\ III.4.15]{Simon} to the quantum case.}
\begin{lemma}
\label{LemReproduced}
Suppose that $\Phi$ is any finite-range translation-invariant interaction, not physically equivalent to zero. Then, for all $u\in(u_{\min}(\Phi),u_{\max}(\Phi)]$, we have
\[
   \lim_{n\to\infty} \frac 1 {|\Lambda_n|} \log \left| \left\{\text{eigenvalues of }H_{\Lambda_n}\leq u\cdot |\Lambda_n|\right\}\right| = s(u),
\]
where $s(u)$ is defined in Lemma~\ref{LemVariational}.
\end{lemma}
\begin{proof}
We transfer the classical proof of~\cite[Thm.\ III.4.15]{Simon} to the quantum case (with slight modifications and simplifications, using notation established earlier).
Define
\[
   N_{\Lambda_n}(u):=\left| \left\{\text{eigenvalues of }H_{\Lambda_n}\leq u\cdot |\Lambda_n|\right\}\right|\qquad
   \bar s(u):=\limsup_{n\to\infty}\frac 1 {|\Lambda_n|} \log N_{\Lambda_n},\qquad \underline{s}(u):=\liminf_{n\to\infty}\frac 1 {|\Lambda_n|} \log N_{\Lambda_n}.
\]
Denote the eigenvalues of $H_{\Lambda_n}$ by $E_i$, and $Z:=\tr(\exp(-\beta(u) H_{\Lambda_n}))$, then
\[
   1\geq \frac 1 Z \sum_{E_i:\, E_i/|\Lambda_n|\leq u} e^{-\beta(u) E_i}\geq \frac 1 Z N_{\Lambda_n}e^{-\beta(u) u|\Lambda_n|}.
\]
Taking logarithms, we obtain $\frac 1 {|\Lambda_n|} \log N_{\Lambda_n} \leq \frac 1 {|\Lambda_n|} \log Z +\beta(u) u \stackrel{n\to\infty}\longrightarrow s(u)$, hence
\begin{equation}
   \bar s(u)\leq s(u).
   \label{eqB12}
\end{equation}
The converse inequality is more involved. Fix $u_1\leq u_2$, $\delta>0$, and $0<\lambda<1$.
Use the notation of the proof of Lemma~\ref{LemMainLemma}, where we have split $\Lambda_{n_k}$ into disjoint regions
$\Lambda_m^{(i)}$, $i=1,\ldots,N_k$, and $\Lambda_{\rm rem}$. Set $n_k=k$. Denote by $|E_1\rangle,\ldots,|E_M\rangle$ mutually orthonormal eigenvectors of
$H_{\Lambda_m}$ with energy density less than or equal to $u_1$, and $|E'_1\rangle,\ldots,|E'_N\rangle$ mutually orthonormal eigenvectors of $H_{\Lambda_m}$ with
energy density less than or equal to $u_2-\delta$, where $M:=N_{\Lambda_m}(u_1)$ and $N:=N_{\Lambda_m}(u_2-\delta)$. Set $i:=(i_1,\ldots,i_{N_k})$, where
$i_1,\ldots,i_{\lfloor \lambda N_k\rfloor}\in\{1,\ldots,M\}$, and $i_{\lfloor \lambda N_k\rfloor +1},\ldots, i_{N_k}\in\{1,\ldots,N\}$. For every possible choice of $i$, define
\[
   |\psi_i\rangle:=\bigotimes_{l=1}^{\lfloor \lambda N_k\rfloor} |E_{i_l}\rangle_{\Lambda_m^{(l)}}\otimes \bigotimes_{l=\lfloor \lambda N_k\rfloor +1}^{N_k}
   |E'_{i_l}\rangle_{\Lambda_m^{(l)}}\otimes |0\rangle_{\Lambda_{\rm rem}},
\]
where $|0\rangle_{\Lambda_{\rm rem}}$ is an arbitrary pure state on $\Lambda_{\rm rem}$. Then we have
\begin{eqnarray*}
   \frac{\langle \psi_i|H_{\Lambda_k}|\psi_i\rangle}{|\Lambda_k|} &\leq& \frac{\langle \psi_i|H_{\Lambda_k}^{(m)}|\psi_i\rangle + \|H_{\Lambda_k}-H_{\Lambda_k}^{(m)}\|_\infty}
   {|\Lambda_k|}\\
   &\leq& \frac{\lfloor \lambda N_k\rfloor}{|\Lambda_k|} |\Lambda_m| u_1 + \frac{N_k-\lfloor \lambda N_k \rfloor}{|\Lambda_k|}|\Lambda_m|(u_2-\delta)+
   c_1\frac{|\Lambda_{\rm rem}|}{|\Lambda_k|} + c_2 \frac{N_k}{|\Lambda_k|} |\partial \Lambda_m|.
\end{eqnarray*}
If $k$ and $m$ are large enough (while $k\gg m$), the right-hand side is less than $u':=\lambda u_1+(1-\lambda)u_2$. Furthermore, if $i\neq i'$ then $|\psi_i\rangle\perp |\psi_{i'}\rangle$,
thus $N_{\Lambda_k}(u')\geq |\{|\psi_i\rangle\}|=M^{\lfloor \lambda N_k\rfloor} N^{N_k-\lfloor \lambda N_k\rfloor}$. Taking logarithms, we obtain
\[
   \frac 1 {|\Lambda_k|} \log N_{\Lambda_k}\left(u'\right)\geq \frac 1 {|\Lambda_k|}\left(
      \lfloor \lambda N_k \rfloor \log N_{\Lambda_m}(u_1)+(N_k-\lfloor \lambda N_k\rfloor)\log N_{\Lambda_m}(u_2-\delta)
   \right).
\]
Since $\lim_{k\to\infty}N_k/|\Lambda_k|=1/|\Lambda_m|$, this yields
\[
   \underline{s}\left(u'\right)\geq \frac {\lambda} {|\Lambda_m|}\log N_{\Lambda_m}(u_1)+\frac {1-\lambda} {|\Lambda_m|} \log N_{\Lambda_m}(u_2-\delta),
\]
and thus
\begin{equation}
   \underline{s}\left(\lambda u_1+(1-\lambda) u_2\right)\geq \lambda  \underline{s}(u_1)+(1-\lambda) \underline{s}(u_2-\delta).
   \label{eqB13}
\end{equation}
Now consider a fixed value of $u$, and set $\beta:=\beta(u)$.
We use the elementary inequalities for $a\leq b$:
\begin{eqnarray}
    \limsup_{k\to\infty} \frac 1 {|\Lambda_k|} \log \sum_{i:\, E_i/|\Lambda_k|\in [a,b]} e^{-\beta E_i} &\leq& -\beta a +\bar s(b),\label{eqB14} \\
     \liminf_{k\to\infty} \frac 1 {|\Lambda_k|} \log \sum_{i:\, E_i/|\Lambda_k|\in [a,b]} e^{-\beta E_i} &\leq& -\beta a +\underline{s}(b), \label{eqB15}
\end{eqnarray}
where the $E_i$ are now the eigenvalues of $H_{\Lambda_k}$. Suppose that $\beta$ is a point of differentiability of $p(\cdot,\Phi)$ such that (due to strict
convexity) $s(u')-\beta u'<p(\beta,\Phi)$ for all $u'\neq u$.
Let $\alpha>0$ such that $s(u')-\beta u' \leq p(\beta,\Phi)-\alpha$ for all
$u'$ with $|u'-u|>\delta$. Choose $\varepsilon>0$ such that $\varepsilon\beta\leq \alpha/2$. Now we decompose the energy density interval into a disjoint union
\[
   \left( u_{\min}(\Phi),u_{\max}(\Phi)\right] \setminus (u-\delta,u+\delta)=\bigcup_{j=1}^{n-1} I_j,\qquad\mbox{where }I_j=(a_j,b_j]\mbox{ with }
   |b_j-a_j|\leq\varepsilon;\enspace I_n:=(u_{\max}(\Phi),\infty).
\]
Due to~(\ref{eqB14}), we have $\displaystyle \limsup_{k\to\infty}\frac 1 {|\Lambda_k|} \log \sum_{i:\, E_i/|\Lambda_k|\in I_j} e^{-\beta E_i}
\leq\-\beta a_j+\bar s(b_j)$, and due to~(\ref{eqB12}), we obtain for $j\leq n-1$
\[
   \bar s(b_j)-\beta a_j \leq s(b_j)-\beta b_j+\beta(b_j-a_j) \leq s(b_j)-\beta b_j+\frac\alpha 2 \leq p(\beta,\Phi)-\frac\alpha 2,
\]
and for $j=n$, we get
\[
   \bar s(b_j)-\beta a_j = \log d -\beta u_{\max}(\Phi) = s(a_n)-\beta a_n \leq p(\beta,\Phi)-\alpha.
\]
Thus
\begin{eqnarray*}
   \limsup_{k\to\infty} \frac 1 {|\Lambda_k|} \log \sum_{i:\, \left| E_i/|\Lambda_k|-u\right|\geq\delta} e^{-\beta E_i}&=&
   \limsup_{k\to\infty}\frac 1 {|\Lambda_k|} \log \sum_{j=1}^n \sum_{i:\, E_i/|\Lambda_k|\in I_j} e^{-\beta E_i}\\
   &\leq& \limsup_{k\to\infty}\frac 1 {|\Lambda_k|} \left( \log n +\max_j \sum_{i:\, E_i/|\Lambda_k|\in I_j} e^{-\beta E_i}\right)
   \leq p(\beta,\Phi)-\frac 1 2 \alpha.
\end{eqnarray*}
But since $\lim_{k\to\infty}\frac 1 {|\Lambda_k|} \log \sum_i e^{-\beta E_i}=p(\beta,\Phi)$ by definition of the pressure, we obtain
\[
   \liminf_{k\to\infty}\frac 1 {|\Lambda_k|} \log\sum_{i:\, \left| E_i/|\Lambda_k|-u\right|\leq\delta} e^{-\beta E_i}\geq p(\beta,\Phi).
\]
Comparing this to~(\ref{eqB15}) yields $p(\beta,\Phi)\leq -\beta(u-\delta)+\underline{s}(u+\delta)$, hence
\[
   \underline{s}(u+\delta)\geq \lim_{\delta\to 0} \underline{s}(u+\delta)\geq \lim_{\delta\to 0} p(\beta,\Phi)+\beta(u-\delta)=p(\beta,\Phi)+\beta u =s(u).
\]
This finally shows that
\[
   \underline{s}(u+\delta)\geq s(u)\qquad\mbox{for all }\delta>0,\mbox{ if }\beta(u)\mbox{ is a point of differentiability of }\beta\mapsto p(\beta,\Phi).
\]
Since $\beta\mapsto p(\beta,\Phi)$ is strictly convex, the right and left derivatives $D^+ p$ and $D^- p$ exist everywhere,
and the set $B:=\{\beta>0\,\,|\,\, (D^+ p)(\beta)\neq (D^-)(p)(\beta)$ is countable. Furthermore, the set $A:=\{u\in (u_{\min}(\Phi),u_{\max}(\Phi))\,\,|\,\, \beta(u)\in B\}$
is a countable union of closed intervals. If $u$ is any value such that there is a sequence $(u_n)_{n\in\N}$, $u_n\leq u$, with $\lim_{n\to\infty} u_n=u$ and $u_n\not\in A$,
then $\underline{s}(u)\geq s(u_n)$, and due to continuity of $s$, we get $\underline{s}(u)\geq s(u)$. We get this inequality for all $u\not\in A$ and
the left-hand endpoints of intervals in $A$.

Finally, let $[u_0,u_1]\subset A$ be an isolated closed interval and $u\in (u_0,u_1]$. Then for every $\varepsilon\geq 0$ there is $\lambda_{\varepsilon}\in (0,1)$
with $u=\lambda_{\varepsilon}(u_1+\varepsilon)+(1-\lambda_{\varepsilon})u_0$. Then, for every $\varepsilon>0$ small enough such that $u_1+\varepsilon\not\in A$,
and $\delta>0$ small enough such that $u_0-\delta\not\in A$, we get
due to~(\ref{eqB13})
\[
   \underline{s}(u)=\underline{s}(\lambda_{\varepsilon}(u_1+\varepsilon)+(1-\lambda_{\varepsilon})u_0) \geq \lambda_{\varepsilon}\underline{s}(u_1+\varepsilon)
   +(1-\lambda_{\varepsilon})\underline{s}(u_0-\delta) \geq \lambda_{\varepsilon}s(u_1+\varepsilon)+(1-\lambda_{\varepsilon})s(u_0-\delta).
\]
Since $s$ is continuous, we can first take the limit $\delta\to 0$ and then the limit $\varepsilon\to 0$ to obtain
\[
   \underline{s}(u)\geq \lambda_0 s(u_1) + (1-\lambda_0)s(u_0)=s(u),
\]
where we have used the fact that $s$ is linear  on $[u_0,u_1]$. Together with~(\ref{eqB12}), this completes the proof.
\end{proof}

This lemma only refers to the Hamiltonian $H_{\Lambda_n}$ corresponding to open boundary conditions. However, we need this in more generality,
in particular for the case of periodic boundary conditions.

\begin{lemma}
\label{LemHBC}
Let $H_{\Lambda_n}^{BC}$ be the Hamiltonians corresponding to an arbitrary choice of boundary conditions in the sense of Definition~\ref{DefBC}.
Then
\[
   \lim_{n\to\infty} \frac 1 {|\Lambda_n|} \log \left| \left\{\text{eigenvalues of }H_{\Lambda_n}^{BC}\leq u\cdot |\Lambda_n|\right\}\right| = s(u).
\]
\end{lemma}
\proof
Define $H_{\partial \Lambda_n}^{BC}:=H_{\Lambda_n}^{BC}-H_{\Lambda_n}$.
Fix $u$, and let $\tilde u< u$ be arbitrary. If $n$ is large enough, then
\[
   \tilde u |\Lambda_n|+\|H_{\partial\Lambda_n}^{BC}\|_\infty\leq u|\Lambda_n|.
\]
Thus, due to Weyl's Perturbation Theorem~\cite{BhatiaPerturb}, if $\lambda_1,\ldots, \lambda_k$ are the $k$ smallest eigenvalues of $H_{\Lambda_n}$, then
$H_{\Lambda_n}^{BC}$ has eigenvalues $\lambda'_i\leq \lambda_i+\|H_{\partial\Lambda_n}^{BC}\|$. Therefore
\begin{eqnarray*}
   \liminf_{n\to\infty}\frac 1 {|\Lambda_n|} \log \left| \left\{ \text{eigenvalues of }H_{\Lambda_n}^{BC}\leq u |\Lambda_n|\right\}\right| &\geq&
   \liminf_{n\to\infty}\frac 1 {|\Lambda_n|} \log \left| \left\{ \text{eigenvalues of }H_{\Lambda_n}^{BC}\leq \tilde u |\Lambda_n|+\|H_{\partial\Lambda_n}^{BC}\|\right\}\right|  \\
   &\geq&  \liminf_{n\to\infty}\frac 1 {|\Lambda_n|} \log \left| \left\{ \text{eigenvalues of }H_{\Lambda_n}\leq \tilde u |\Lambda_n|\|\right\}\right|=s(\tilde u).
\end{eqnarray*}
By continuity of $s$, since this is true for all $\tilde u<u$, the previous inequality is also true if $s(\tilde u)$ is replaced by $s(u)$. Similarly, if $\tilde u>u$ is
arbitrary, then
\begin{eqnarray*}
   \limsup_{n\to\infty}\frac 1 {|\Lambda_n|} \log \left| \left\{ \text{eigenvalues of }H_{\Lambda_n}^{BC}\leq u |\Lambda_n|\right\}\right| &\leq&
   \limsup_{n\to\infty}\frac 1 {|\Lambda_n|} \log \left| \left\{ \text{eigenvalues of }H_{\Lambda_n}^{BC}\leq \tilde u |\Lambda_n|-\|H_{\partial\Lambda_n}^{BC}\|\right\}\right|  \\
   &\leq&  \limsup_{n\to\infty}\frac 1 {|\Lambda_n|} \log \left| \left\{ \text{eigenvalues of }H_{\Lambda_n}\leq \tilde u |\Lambda_n|\|\right\}\right|=s(\tilde u).
\end{eqnarray*}
This proves the claim.
\qed

As an immediate consequence we obtain the following result.
\begin{example}[Microcanonical versus canonical ensemble]
\label{ExFlat}
The sequence of states $(\tau_n)_{n\in\N}$ which are defined as the maximal mixtures on the microcanonical subspaces
\[
   T_n^p:={\rm span}\left\{ |E\rangle\,\,\left|\,\, H_{\Lambda_n}^p|E\rangle=E|E\rangle,\enspace \frac E {|\Lambda_n|}\in(u-\delta,u)\right.\right\},
\]
where $H_{\Lambda_n}^p$ is the Hamiltonian on $\Lambda_n$ with periodic boundary conditions satisfies the premises
of Theorem~\ref{TheEquivalence}. That is, we obtain equivalence of ensembles in the standard sense:
\[
   \lim_{n\to\infty}\left\| {\rm Tr}_{\Lambda_n\setminus\Lambda_m} \tau_n - {\rm Tr}_{\Lambda_n\setminus\Lambda_m}
   \frac{\exp(-\beta H_{\Lambda_n}^p)} {Z_n}\right\|_1=0,
\]
where one may either set $\beta$ equal to $\beta(u)$, the inverse temperature corresponding to energy density $u$ in the thermodynamic limit,
or equal to the ($n$-dependent) solution of $\displaystyle \frac 1 {|\Lambda_n|}\tr\left( H_{\Lambda_n}^{BC}\, \frac{\exp(-\beta H_{\Lambda_n}^p)}{Z_n}\right)=u$,
where $BC$ denotes an arbitrary fixed choice of boundary conditions.
\end{example}

In this example, as well as in Theorem~\ref{TheEquivalence}, the partial traces cannot be removed: globally, the microcanonical and the canonical ensemble will in general have
large one-norm distance. In the example of a non-interacting system of binary spins, the well-known tightness
of the classical finite de Finetti theorem provides a proof of this, see Lemma~\ref{LemNonIntDiaconis}.

Furthermore, it is crucial to use the reduction of the global Gibbs state,
${\rm Tr}_{\Lambda_n\setminus\Lambda_m}\exp(-\beta H_{\Lambda_n}^p)/Z_n$,
instead of the local Gibbs state, $\exp(-\beta H_{\Lambda_m}^p)$. Replacing the former by the latter renders the statement of the theorem false in general.
This is rather obvious: the local Gibbs state will in general be different from the reduction of the global one, due to interaction terms across the boundary
of $\Lambda_m$.
This phenomenon will also occur in Subsection~\ref{SecLocalDiag}, where we prove a special case of the ``eigenstate thermalization hypothesis'' only
by taking the boundary terms into account.
A concrete counterexample to the naive version of equivalence of ensembles is already given by the classical Ising model, interpreted as a quantum model.

\begin{example}[The Ising model]
\label{ExNoLocal}
Consider the one-dimensional model on $\Lambda_n:=\{-n,\ldots,n\}$
\begin{equation}
   H_{\Lambda_n}^p:=-J\sum_{i=-n}^n Z_i Z_{i+1}-h\sum_{i=-n}^n Z_i,
   \label{eqIsing}
\end{equation}
identifying $n+1\equiv -n$. Here, $Z_i$ denotes the Pauli $Z$-matrix $Z=\left(\begin{array}{cc} 1 & 0 \\ 0 & -1 \end{array}\right)$
on lattice site $i$. This model has a unique Gibbs state $\omega_\beta$ (in the thermodynamic
limit $n\to\infty$) for all $\beta\geq 0$, see~\cite{Huang}.
Fix $m=0$, and consider the reduction of the global microcanonical state
$\tau_n$ to $\Lambda_0=\{0\}$, a single lattice site. Due to Example~\ref{ExFlat} and Theorem~\ref{TheEquivalence}, we have
\[
   \lim_{n\to\infty}{\rm Tr}_{\Lambda_n\setminus\Lambda_0} \tau_n = (\omega_\beta)_{\Lambda_0}.
\]
On the other hand, using the known formula for the magnetization of the Ising model~\cite{Huang}, we have
\begin{equation}
   \tr\left(\strut (\omega_\beta)_{\Lambda_0} Z\right) = \frac 1 {2n+1} \sum_{i=-n}^n \tr\left(\strut(\omega_\beta)_{\Lambda_n} Z_i\right)
   =\frac{\sinh(\beta h)}{\sqrt{\sinh^2(\beta h)+\exp(-4\beta J)}},
   \label{eqNotHolomorphic}
\end{equation}
where the first equality is due to translation-invariance, and the second equality follows from taking the limit $n\to\infty$ and using the well-known result
for the magnetization of this model. We can compare this with the local Gibbs state $\omega_\beta^{\rm loc}$,
which is defined as the normalization of $\exp(-\beta H_{\Lambda_0})$. We run into an immediate conceptual problem: how do we define $H_{\Lambda_0}$?
The most obvious choice is $H_{\Lambda_0}=Z_0$, but we have the freedom to interpret~(\ref{eqIsing}) in different ways, by subtracting local terms from $Z_i Z_{i+1}$
and adding them to the $Z_i$-term. This is exactly the freedom that we encountered before, in the definition of \emph{physical equivalence}
that we discussed before Lemma~\ref{LemVariational}. Whatever we define to be $H_{\Lambda_0}$, it should be \emph{some} fixed Hamiltonian
which can be written in the form $H_{\Lambda_0}=U\left(\begin{array}{cc} E_1 & 0 \\ 0 & E_2\end{array}\right)U^\dagger$, with $U$ unitary
and $E_1,E_2\in\R$ its energy eigenvalues. Our crucial assumption will be that whatever $H_{\Lambda_0}$ is, it should be independent of $\beta$. But then
\[
   \tr(\omega_\beta^{\rm loc} Z)=\frac{\tr\left[ U \left(\begin{array}{cc} \exp(-\beta E_1) & 0 \\ 0 & \exp(-\beta E_2)\end{array}\right)U^\dagger Z\right]}
   {\exp(-\beta E_1)+\exp(-\beta E_2)}.
\]
Regarding this as a function $f(\beta)$ for complex $\beta\in\C$, we obtain a function that is holomorphic except for possibly countably many isolated singularities
on the imaginary axis (if $E_1\neq E_2$). This is not true for~(\ref{eqNotHolomorphic}) which is a function with branch cut singularities due to the presence of the square root.
This shows that $\omega_\beta^{\rm loc}\neq (\omega_\beta)_{\Lambda_0}$ at least for some values of $\beta>0$,
no matter how we define $H_{\Lambda_0}$. Thus ${\rm Tr}_{\Lambda_n\setminus\Lambda_m}\tau_n$
cannot converge to $\omega_\beta^{\rm loc}$ in the thermodynamics limit where $n\to\infty$.
\end{example}

The standard microcanonical ensemble (mentioned in Example~\ref{ExFlat} above) is defined as a flat distribution on the energy windows subspace corresponding
to the interval $(u-\delta,u)$. However, we can apply Theorem~\ref{TheEquivalence} more generally. In order to slightly generalize Example~\ref{ExFlat}, we need
another simple lemma:
\begin{lemma}
\label{LemEntropy}
Let $(p_1,\ldots,p_n)$ be discrete probability distribution, and suppose that there exists $M\geq 1$ such that $\displaystyle \frac {p_i}{p_j}\leq M$ for all $i\neq j$.
Then its Shannon entropy satisfies $H(p)\geq \log n -\log M$.
\end{lemma}
\proof
Let $\ell_i:=\log(1/p_i)$, then $\log(p_i/p_j)=\ell_j-\ell_i$, and the condition above implies $|\ell_i-\ell_j|\leq \log M$ for all $i,j$. Then all $\ell_i$ lie in the
interval $[\ell_{min},\ell_{max}]$, where $\ell_{min}:=\min_i \ell_i$ and $\ell_{max}:=\max_i \ell_i$. This interval has size at most $\ell_{max}-\ell_{min}\leq \log M$.
Since $\min_i p_i\leq 1/n \leq \max_i p_i$, the quantity $\log n$ must be contained in this interval. Thus $|\ell_i-\log n|\leq \log M$ for all $i$. It follows that
\[
   |H(p)-\log n|=\left| \sum_i p_i\log\frac 1 {p_i}-\sum_i p_i\log n\right| =\left| \sum_i p_i\ell_i-\sum_i p_i\log n\right|
   \leq \sum_i p_i |\ell_i -\log n|\leq \log M.
\]
\qed

Now we apply this to prove a generalization of Example~\ref{ExFlat}.

\begin{figure}[!hbt]
\begin{center}
\includegraphics[angle=0, width=5cm]{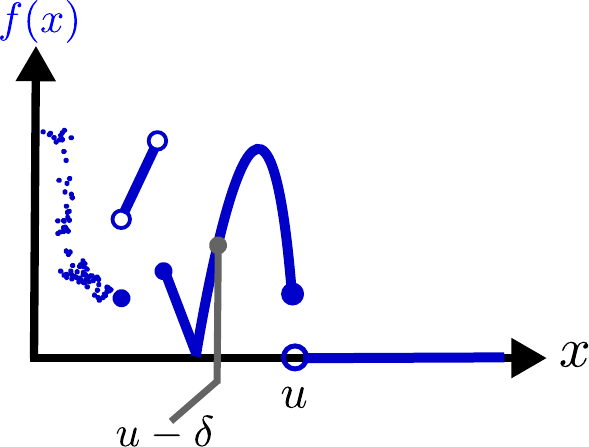}
\caption{Spectral density functions that satisfy the premises of Example~\ref{ExMain}, and yield equivalence of ensembles in the sense
that the corresponding microcanonical state locally resembles the canonical state. The non-negative bounded function $f$ must satisfy $f(x)=0$ for
all $x\geq u$, and there must be $\delta>0$ such that $f$ is continuous and strictly positive on the interval $[u-\delta,u]$. For $x<u-\delta$, $f$ can
have all kinds of discontinuities.
}
\label{fig_density}
\end{center}
\end{figure}

\begin{example}[Microcanonical ensemble with given distribution function]
\label{ExMain}
Let $\Phi$ be an interaction which is not physically equivalent to zero, and let 
$u_{\min}(\Phi)<u\leq u_{\max}(\Phi)$. Let $f:[u_{\min}(\Phi),u_{\max}(\Phi)]\to\R$ be a bounded nonnegative function such that $f(x)= 0$ for all $x>u$ and such that
there exists $\delta>0$ such that $f$ is continuous and strictly positive on $[u-\delta,u]$, cf.\ Figure~\ref{fig_density}. For every $n\in\N$, let $\{|E_i\rangle\}_i$ be an arbitrary energy
eigenbasis of $H_{\Lambda_n}^p$, the Hamiltonian on $\Lambda_n$ with periodic boundary conditions. Then the set of states defined by
\[
   \tau_n:=\frac 1 {\mathcal{N}} \sum_i f\left(\frac{E_i}{|\Lambda_n|}\right) |E_i\rangle\langle E_i|,
\]
where $\mathcal{N}:=\sum_i f(E_i/|\Lambda_n|)$, satisfies the premises of Theorem~\ref{TheEquivalence}. That is, this modified microcanonical ensemble resembles locally
the canonical ensemble.
\end{example}
\proof
Since $\tr(\tau_n H_{\Lambda_n}^p)\leq u|\Lambda_n|$, we have
$\bar u:=\limsup_{n\to\infty}\frac 1 {|\Lambda_n|} \tr(\tau_n H_{\Lambda_n})\leq u$.
Define $I_n:=\{i\,\,|\,\, E_i/|\Lambda_n|\in[u-\delta,u]\}$, then $\tau_n=(1-\lambda_n)\sigma_n+\lambda_n \sigma'_n$, where
\[
   \sigma_n=\sum_{i\in I_n} \frac{f(E_i/|\Lambda_n|)}{\sum_{j\in I_n} f(E_j/|\Lambda_n|)} |E_i\rangle\langle E_i|,\qquad
   \sigma'_n=\sum_{i\not\in I_n}\frac{f(E_i/|\Lambda_n|)}{\sum_{j\not\in I_n} f(E_j/|\Lambda_n|)} |E_i\rangle\langle E_i|,\qquad
   \lambda_n=\frac{\sum_{j\not\in I_n}f(E_j/|\Lambda_n|)}{\sum_j f(E_j/|\Lambda_n|)}.
\]
According to Lemma~\ref{LemHBC}, we have
\[
   \sum_{j\not\in I_n} f(E_j/|\Lambda_n|)\leq \#\{i\,\,|\,\, E_i/|\Lambda_n|<u-\delta\}\cdot \|f\|_\infty = \exp[|\Lambda_n|\, s(u-\delta)+o(|\Lambda_n|)].
\]
On the other hand,
\[
   \sum_j f(E_j/|\Lambda_n|) \geq \sum_{j\in I_n} f(E_j/|\Lambda_n|) \geq \# I_n\cdot \min_{x\in[u-\delta,u]} f(x) = \exp[|\Lambda_n|\, s(u) - o(|\Lambda_n|)].
\]
This shows that $\lim_{n\to\infty}\lambda_n=0$, and concavity of the entropy, i.e.\ $S(\tau_n)\geq (1-\lambda_n)S(\sigma_n)+\lambda_n S(\sigma'_n)$, yields
\[
   \liminf_{n\to\infty}\frac 1 {|\Lambda_n|} S(\tau_n) \geq \liminf_{n\to\infty} \frac 1 {|\Lambda_n|} S(\sigma_n).
\]
But the eigenvalues of $\sigma_n$ are $\displaystyle p_i:=\frac{f(E_i/|\Lambda_n|)}{\sum_{j\in I_n} f(E_j/|\Lambda_n|)}$, such that
$\displaystyle \frac{p_i}{p_j} = \frac{f(E_i/|\Lambda_n|)}{f(E_j/|\Lambda_n|)}\leq \frac b a$, where $a:=\min_{x\in[u-\delta,u]} f(x)$ and
$b:=\max_{x\in[u-\delta,u]} f(x)$. Thus, Lemma~\ref{LemEntropy} shows that
\[
   S(\sigma_n)\geq \log\# I_n-\log\frac b a = |\Lambda_n|\cdot s(u)-o(|\Lambda_n|),
\]
and so $\underline{s}:=\liminf_{n\to\infty}\frac 1 {|\Lambda_n|} S(\tau_n)\geq s(u)$. In summary, we obtain
\[
   \liminf_{n\to\infty} \frac 1 {|\Lambda_n|} \left( S(\tau_n)-\beta\,\tr(\tau_n H_{\Lambda_n})\right)\geq \underline{s} -\beta\, \bar u \geq s(u)-\beta\, u =p(\beta,\Phi).
\]
This proves all the premises of Theorem~\ref{TheEquivalence}.
\qed

\textbf{Remark.}
The condition that $f$ has a discontinuity at $u$ (i.e. $f(u)>0$, but $f(x)=0$ for all $x>u$) can be relaxed: the statement above will remain valid
if $f(u)=0$ as long as $f(x)$ does not tend to zero too quickly as $x\to u$. However, the question what ``too quickly'' means
mathematically seems to depend on the choice of the model, because it depends on subtle properties of the spectrum of $H_\Lambda$, in particular on
the number of eigenvalues in certain intervals with diameters of order $o(|\Lambda|)$. In this paper, we only analyze
what can be said in full generality from translation-invariance alone, without reference to any details of the model.

The main proof idea used in this subsection -- to apply the variational principle~(\ref{eqVariationalPrinciple}) -- has been pioneered by Lima~\cite{Lima1,Lima2}.
Our result however is more general:
\begin{itemize}
\item It involves more general spectral density functions (the function $f$ in Example~\ref{ExMain}) instead of only the flat distribution,
\item {it allows to determine the inverse temperature from the energy density on the finite region $\Lambda_n$,}
\item it allows local lattice site dimensions larger than two, and, most significantly,
\item Lima considers only a restricted set of interactions that commute with a particle number operator, see~\cite[p.\ 183]{Lima2}, and~\cite[p.\ 63]{Lima1}.
There is no such restriction in this work.
\end{itemize}

In the remainder of this subsection, we will consider the case of sequences of states $(\tau_n)_{n\in\N}$ that are not necessarily $\Lambda_n$-translation-invariant.
The simplest example is given by the microcanonical ensembles (in the sense of Example~\ref{ExFlat}) if boundary conditions are not periodic.
The proof of Theorem~\ref{TheEquivalence} does not work any more, because we cannot guarantee that limit points of this sequence,
as states on the quasi-local algebra, are translation-invariant.

However, we can still prove a version of equivalence of ensembles in this case, even though it will be a weaker version. This was already seen by Lima~\cite{Lima2}.
In a nutshell, we will prove an equivalence of ensemble result for a restricted set of observables. The following definition specifies the class of observables
that we will consider.

\begin{definition}[$m$-block periodically averaged observable]
For $m\leq n$, an operator $A\in\mathcal{A}_n$ will be called an \emph{$m$-block periodically averaged observable}
if there exists $A'\in\mathcal{A}_m$ with $A'=(A')^\dagger$ such that
\begin{equation}
   A=\frac 1 {|\mathbf{T}(\Lambda_n)|} \sum_{T\in\mathbf{T}(\Lambda_n)} T (A'\otimes\mathbf{1}) T^\dagger,
   \label{eqPeriodicallyAveraged}
\end{equation}
where $\mathbf{T}(\Lambda_n)$ denotes all periodic translations of the finite region $\Lambda_n$ into itself, and the unit observable is
supported on $\Lambda_n\setminus\Lambda_m$. Moreover, $A$ will be called an \emph{$m$-block periodically averaged effect}
if there exists $A'\in\mathcal{A}_m$ which satisfies the equation above, and additionally satisfies $0\leq A'\leq\mathbf{1}$.
\end{definition}

Note that $m$-block periodically averaged observables $A$ on $\Lambda_n$ are automatically $\Lambda_n$-translation-invariant. The notion ``effect'' refers to the property
that they satisfy $0\leq A \leq\mathbf{1}$ (as inherited from $A'$), and can thus be interpreted as defining a binary measurement with POVM elements $(A,\mathbf{1}-A)$.

The usual $\|\cdot\|_1$-distance on density matrices (which is twice the trace distance) can be interpreted (up to a factor of two)
as the maximal possible difference of probabilities in any binary measurement that is applied to the states:
\[
   \|\rho-\sigma\|_1=2\max_{0\leq P \leq \mathbf{1}} \left| \tr(P\rho)-\tr(P\sigma)\right|.
\]
Similarly, we can define a pseudonorm that quantifies the extent to which two states differ in the expectation value of $m$-block periodically averaged
effects: for $m\in\N$ and $M=M^\dagger\in\mathcal{A}_n$ with $n\geq m$, set
\[
   \|M\|_{\{m\}}:=2\max\left\{ \left|\tr(PM)\right|\,\,|\,\, P\mbox{ is an $m$-block periodically averaged effect on $\Lambda_n$}\right\}.
\]
As a consequence, $\|\rho-\sigma\|_{\{m\}}$ denotes the maximal difference in probabilities of any measurements described by $m$-block averaged
effects that are performed on $\rho$ resp.\ $\sigma$. It is clear that $0\leq\|A\|_{\{m\}}\leq \|A\|_1$, and the norm properties
$\|\lambda A\|_{\{m\}}=|\lambda|\, \|A\|_{\{m\}}$ for $\lambda\in\R$ as well as $\|A+B\|_{\{m\}}\leq \|A\|_{\{m\}}+\|B\|_{\{m\}}$ are satisfied.
However, $\|A\|_{\{m\}}$ can be zero without $A$ being zero, which shows that $\|\cdot\|_{\{m\}}$ is not a norm.

In the case where we have an $m$-block periodically averaged observable which does not come from an effect, we have the following inequality:
\begin{lemma}
\label{LemEffectToObs}
Let $A$ be an $m$-block periodically averaged observable on $\Lambda_n$, coming from an observable $A'\in\mathcal{A}_m$ according to~(\ref{eqPeriodicallyAveraged}).
Then for all quantum states $\rho,\sigma$ on $\Lambda_n$, we have
\[
   \left| \tr(\rho A)-\tr(\sigma A)\right| \leq  \|A'\|_\infty \|\rho-\sigma\|_{\{m\}}.
\]
\end{lemma}
\begin{proof}
Denote by $\lambda_{\min}$ resp.\ $\lambda_{\max}$ the smallest resp.\ largest eigenvalue of $A'$.
If $\lambda_{\max}=\lambda_{\min}$ then there is nothing to prove. Otherwise, set $B':=(\lambda_{\max}-\lambda_{\min})^{-1}(A'-\lambda_{\min}\mathbf{1})$,
then $0\leq B'\leq  \mathbf{1}$. Define $B:=\frac 1 {|\mathbf{T}(\Lambda_n)|} \sum_{T\in\mathbf{T}(\Lambda_n)} T (B'\otimes \mathbf{1})T^\dagger$,
then $B$ is an $m$-block periodically averaged effect, and hence
\[
   \left| \tr(\rho B)-\tr(\sigma B)\right|\leq \frac 1 2 \|\rho-\sigma\|_{\{m\}}.
\]
On the other hand, we have $B=(\lambda_{\max}-\lambda_{\min})^{-1}(A-\lambda_{\min}\mathbf{1})$. Substituting this into the previous inequality,
and using that $\lambda_{\max}-\lambda_{\min}\leq 2\|A'\|_\infty$, we obtain the claimed inequality.
\end{proof}

As a preparation, we need a lemma which says that \emph{periodically averaged} local Gibbs states for arbitrary boundary conditions converge to the global Gibbs
state if it is unique.
\begin{lemma}
\label{LemLocalGibbs}
Fix any $\beta\geq 0$, and let $H_{\Lambda_n}^{BC}$ be a sequence of Hamiltonians with arbitrary boundary conditions,
corresponding to an interaction $\Phi$ which is not physically equivalent to zero and which has a unique Gibbs state $\omega_\beta$ at inverse temperature $\beta$.
Then, for every $m\in\N$,
\begin{equation}
   \lim_{n\to\infty} \Tr_{\Lambda_n\setminus\Lambda_m} \left(
      \frac 1 {|\mathbf{T}(\Lambda_n)|} \sum_{T\in\mathbf{T}(\Lambda_n)} T\, \frac{\exp(-\beta H_{\Lambda_n}^{BC})}{Z_n} \, T^\dagger
   \right) = (\omega_\beta)_{\Lambda_m},
   \label{eqLemma16Part1}
\end{equation}
where $Z_n=\tr\left(\exp(-\beta H_{\Lambda_n}^{BC})\right)$. Furthermore, if Gibbs states are unique around inverse temperature $\beta>0$, and if we define
$\beta_n^{BC}$ as the solution of the equation
$\frac 1 {|\Lambda_n|} \tr\left(H_{\Lambda_n}^{BC}\,\frac{\exp(-\beta_n^{BC} H_{\Lambda_n}^{BC})}{Z_n}\right)=u_n$, with $(u_n)_{n\in\N}$
an arbitrary sequence with $\lim_{n\to\infty} u_n=u(\omega_\beta)$, then $\lim_{n\to\infty}\beta_n^{BC}=\beta$, and
\[
   \lim_{n\to\infty} \Tr_{\Lambda_n\setminus\Lambda_m} \left(
      \frac 1 {|\mathbf{T}(\Lambda_n)|} \sum_{T\in\mathbf{T}(\Lambda_n)} T\, \frac{\exp(-\beta_n^{BC} H_{\Lambda_n}^{BC})}{Z'_n}\, T^\dagger
   \right) = (\omega_\beta)_{\Lambda_m},
\]
where $Z'_n=\tr\left(\exp(-\beta_n^{BC} H_{\Lambda_n}^{BC})\right)$.
\end{lemma}
\proof
Set $\rho_n(\beta'):=\exp(-\beta' H_{\Lambda_n}^{BC})/Z_n(\beta')$, and $\rho_n:=\rho_n(\beta)$.
By construction, $\rho_n$ maximizes the functional $\rho\mapsto S(\rho)-\beta\, \tr(H_{\Lambda_n}^{BC} \rho)$. Thus
\[
   S(\rho_n)-\beta\, \tr(H_{\Lambda_n}^{BC} \rho_n)\geq S\left( (\omega_{\beta})_{\Lambda_n}\right)-\beta\, \tr\left( H_{\Lambda_n}^{BC} (\omega_{\beta})_{\Lambda_n}\right).
\]
Set $\rho'_n:=1/(|\mathbf{T}(\Lambda_n)|) \sum_{T\in\mathbf{T}(\Lambda_n)} T \rho_n T^\dagger$, then concavity of the entropy implies
$S(\rho'_n)\geq S(\rho_n)$. Since $T^\dagger H_{\Lambda_n}^p T = H_{\Lambda_n}^p$ for all $T\in\mathbf{T}(\Lambda_n)$, we have
$\tr(\rho'_n H_{\Lambda_n}^p)=\tr(\rho_n H_{\Lambda_n}^p)$. We obtain
\begin{eqnarray}
   \liminf_{n\to\infty} \frac 1 {|\Lambda_n|} \left( S(\rho'_n)-\beta\, \tr(\rho'_n H_{\Lambda_n})\right) &\geq&
   \liminf_{n\to\infty} \frac 1 {|\Lambda_n|} \left( S(\rho_n)-\beta\, \tr(\rho'_n H_{\Lambda_n}^p)\right)
   =\liminf_{n\to\infty} \frac 1 {|\Lambda_n|} \left( S(\rho_n)-\beta\, \tr(\rho_n H_{\Lambda_n}^p)\right)\nonumber \\
    &=&\liminf_{n\to\infty} \frac 1 {|\Lambda_n|} \left( S(\rho_n)-\beta\, \tr(\rho_n H_{\Lambda_n}^{BC})\right)\nonumber \\
   &\geq& \liminf_{n\to\infty}\frac 1 {|\Lambda_n|}\left( S\left( (\omega_{\beta})_{\Lambda_n}\right)-\beta\, \tr\left( H_{\Lambda_n}^{BC} (\omega_{\beta})_{\Lambda_n}\right) \right)\nonumber \\
   &=& s(\omega_{\beta})-\beta\, u(\omega_{\beta})=p(\beta,\Phi). \label{eqLemma16Calculation}
\end{eqnarray}
Since every $\rho'_n$ is $\Lambda_n$-translation-invariant, Theorem~\ref{TheEquivalence} proves~(\ref{eqLemma16Part1}) and also
$\lim_{n\to\infty}\frac 1 {|\Lambda_n|} \tr(\rho'_n H_{\Lambda_n})=u(\omega_\beta)$. Thus
\[
   u(\omega_\beta)=\lim_{n\to\infty}\frac 1 {|\Lambda_n|} \tr(\rho'_n H_{\Lambda_n}^p)= \lim_{n\to\infty}\frac 1 {|\Lambda_n|} \tr(\rho_n H_{\Lambda_n}^p)
   =\lim_{n\to\infty}\frac 1 {|\Lambda_n|} \tr(\rho_n H_{\Lambda_n}^{BC}).
\]
Choose $\beta_0,\beta_1\in\R$ such
that $0<\beta_0<\beta<\beta_1$, and such that the Gibbs states at inverse temperatures $\beta_0$ and $\beta_1$ are unique. Then the previous results show that
$\lim_{n\to\infty} \frac 1 {|\Lambda_n|} \tr\left(H_{\Lambda_n}^{BC}\rho_n(\beta_i)\right)=u_i$ for $i=0,1$, where $u_i:=u(\omega_{\beta_i})$.
It follows $u_0>u>u_1$, and thus for $n$ large enough, we have $\frac 1 {|\Lambda_n|} \tr(H_{\Lambda_n}^{BC} \rho_n(\beta_0))>u_n
>\frac 1 {|\Lambda_n|} \tr(H_{\Lambda_n}^{BC} \rho_n(\beta_1))$, so $\beta_0<\beta_n^{BC}<\beta_1$ for $n$ large enough; in particular,
a solution $\beta_n^{BC}$ can be found in the interval $(\beta_0,\beta_1)$. Moreover, since $\beta_0$ and $\beta_1$ can be chosen
arbitrarily close to $\beta$, this proves that $\lim_{n\to\infty}\beta_n^{BC}=\beta$. We can then repeat the calculation~(\ref{eqLemma16Calculation}),
with $\beta$ after the minus sign replaced by $\beta_n^{BC}$ where necessary, $\omega_\beta$ left unchanged, $\rho_n$ replaced by $\rho_n(\beta_n^{BC})$,
and $\rho'_n$ replaced by $\rho'_n(\beta_n^{BC}):=
1/(|\mathbf{T}(\Lambda_n)|) \sum_{T\in\mathbf{T}(\Lambda_n)} T \rho_n(\beta_n^{BC}) T^\dagger$, proving the final claim of the lemma.
\qed

Now we have all the ingredients to prove our main theorem on equivalence of ensembles.

\begin{theorem}[Equivalence of ensembles, non-translation-invariant states]
\label{TheEquivalence2}
Let $(\tau_n)_{n\in\N}$ be a sequence of states on $\mathcal{A}_n$, let $\beta\geq 0$, and let $\Phi$ be a translation-invariant finite-range interaction
which is not physically equivalent to zero, and for which there is a unique Gibbs state $\omega_\beta$ at inverse temperature $\beta$. Suppose that
\[
   \liminf_{n\to\infty}\frac 1 {|\Lambda_n|}\left(\strut S(\tau_n)-\beta\, \tr(\tau_n H_{\Lambda_n})\right)\geq p(\beta,\Phi),
\]
then we have equality in this expression, and
\[
   \lim_{n\to\infty} \left\| \tau_n - \frac{\exp(-\beta H_{\Lambda_n}^{BC})}{Z_n}\right\|_{\{m\}}=0,\qquad\mbox{as well as}\qquad
   \lim_{n\to\infty}\frac 1 {|\Lambda_n|} \tr(\tau_n H_{\Lambda_n})=u(\omega_\beta),
\]
where $Z_n=\tr(\exp(-\beta H_{\Lambda_n}^{BC}))$, and $H_{\Lambda_n}^{BC}$ is the Hamiltonian on $\Lambda_n$ corresponding to $\Phi$ with arbitrary boundary conditions.
Furthermore, if Gibbs states are unique around inverse temperature $\beta>0$, we have
\[
   \lim_{n\to\infty} \left\| \tau_n - \frac{\exp(-\beta_n^{BC} H_{\Lambda_n}^{BC})}{Z_n}\right\|_{\{m\}}=0,
\]
where $\beta_n^{BC}$ is defined as the solution of the equation
$\frac 1 {|\Lambda_n|} \tr\left(H_{\Lambda_n}^{BC}\,\frac{\exp(-\beta_n^{BC} H_{\Lambda_n}^{BC})}{Z_n}\right)=u_n$, where $(u_n)_{n\in\N}$ is any
sequence with $\lim_{n\to\infty}u_n=u(\omega_\beta)$.
\end{theorem}
\begin{proof}
We prove both claims at once, by defining two sequences $(\beta_n)_{n\in\N}$ and $(\beta'_n)_{n\in\N}$, either setting $\beta_n:=\beta$ and $\beta'_n:=\beta$,
\emph{or} setting $\beta_n:=\beta_n^{BC}$ and $\beta'_n:=\beta_n^p$.
Define $\displaystyle\Omega(\sigma):=\frac 1 {|\mathbf{T}(\Lambda_n)|} \sum_{T\in\mathbf{T}(\Lambda_n)} T\sigma T^\dagger$,
then it is easy to check that $\Omega$ is Hilbert-Schmidt self-adjoint, i.e.\ $\tr(A \Omega(B))=\tr(\Omega(A) B)$ for $A=A^\dagger$, $B=B^\dagger$.
Furthermore, define
$\tau'_n:=\Omega(\tau_n)$, then concavity of the entropy
implies that $S(\tau'_n)\geq S(\tau_n)$.
Since the Hamiltonian with periodic boundary conditions satisfies $T H_{\Lambda_n}^p T^\dagger = H_{\Lambda_n}^p$, we obtain
$\tr(\tau'_n H_{\Lambda_n}^p)=\tr(\tau_n H_{\Lambda_n}^p)$, and thus
\begin{eqnarray*}
   \liminf_{n\to\infty}\frac 1 {|\Lambda_n|} \left( S(\tau'_n)-\beta\, \tr(\tau'_n H_{\Lambda_n})\right) &\geq&
   \liminf_{n\to\infty}\frac 1 {|\Lambda_n|} \left( S(\tau_n)-\beta\, \tr(\tau'_n H_{\Lambda_n}^p)\right)
   =\liminf_{n\to\infty}\frac 1 {|\Lambda_n|} \left( S(\tau_n)-\beta\, \tr(\tau_n H_{\Lambda_n}^p)\right)\\
   &=& \liminf_{n\to\infty}\frac 1 {|\Lambda_n|} \left( S(\tau_n)-\beta\, \tr(\tau_n H_{\Lambda_n})\right) \geq p(\beta,\Phi).
\end{eqnarray*}
Thus, $\tau'_n$ satisfies the premises of Theorem~\ref{TheEquivalence}, and~(\ref{eqEquiv2}) and~(\ref{eqEquivFinite}) tell us that
\[
   \lim_{n\to\infty}\left\| {\rm Tr}_{\Lambda_n\setminus\Lambda_m} \tau'_n - {\rm Tr}_{\Lambda_n\setminus\Lambda_m}
   \frac{\exp(-\beta'_n H_{\Lambda_n}^p)} {Z'_n}\right\|_1=0,
\]
where $Z'_n=\tr\left(\exp(-\beta'_n H_{\Lambda_n}^p)\right)$.
Now let $A$ be any $m$-block periodically averaged effect on $\Lambda_n$, then it is of the form~(\ref{eqPeriodicallyAveraged})
with $A'\in\mathcal{A}_m$, $0\leq A'\leq \mathbf{1}$.
A simple calculation shows that $\tr(\tau_n A)=\tr(\tau'_n(A'\otimes\mathbf{1}))$, and $\Lambda_n$-translation-invariance of
$H_{\Lambda_n}^p$ implies that $\displaystyle \tr\left(\frac{\exp(-\beta'_n H_{\Lambda_n}^p)}{Z'_n} A\right)= \tr\left(\frac{\exp(-\beta'_n H_{\Lambda_n}^p)}{Z'_n} (A'\otimes\mathbf{1})\right)$.
Thus
\begin{eqnarray}
   \nonumber
   \left\| \tau_n - \frac{\exp(-\beta'_n H_{\Lambda_n}^p)}{Z'_n}\right\|_{\{m\}}&=&
   2\max_A \left| \tr(\tau_n A)- \tr\left(\frac{\exp(-\beta'_n H_{\Lambda_n}^p)}{Z'_n} A\right)\right| \\ \nonumber
   &=&2\max_{A'}\left| \tr\left( A'\, {\rm Tr}_{\Lambda_n\setminus\Lambda_m} \tau'_n\right) - \tr\left( A'\, {\rm Tr}_{\Lambda_n\setminus\Lambda_m}
   \frac{\exp(-\beta'_n H_{\Lambda_n}^p)}{Z'_n}\right)\right| \\
   &\leq&
   \left\| {\rm Tr}_{\Lambda_n\setminus\Lambda_m} \tau'_n - {\rm Tr}_{\Lambda_n\setminus\Lambda_m}
   \frac{\exp(-\beta'_n H_{\Lambda_n}^p)} {Z'_n}\right\|_1 \stackrel{n\to\infty}\longrightarrow 0.
   \label{eqSpecialCase}
\end{eqnarray}
Now we extend this to arbitrary boundary conditions. Let $A$ be any $m$-block periodically averaged effect, then there exists $A'\in\mathcal{A}_m$ such that
$A=\Omega(A'\otimes\mathbf{1})$. Setting $Z_n=\tr(\exp(-\beta_n H_{\Lambda_n}^{BC}))$, we obtain
\begin{eqnarray*}
   \left| \tr\left(A \frac{\exp(-\beta_n H_{\Lambda_n}^{BC})}{Z_n}\right)\right. &-& \left. \tr\left(A \frac{\exp(-\beta'_n H_{\Lambda_n}^p)}{Z'_n}\right)\right|
   = \left| \tr\left(\Omega(A'\otimes\mathbf{1}) \frac{\exp(-\beta_n H_{\Lambda_n}^{BC})}{Z_n}\right)- \tr\left(\Omega(A'\otimes\mathbf{1})
   \frac{\exp(-\beta'_n H_{\Lambda_n}^p)}{Z'_n}\right)\right| \\
   &=& \left| \tr\left((A'\otimes\mathbf{1}) \Omega\left(\frac{\exp(-\beta_n H_{\Lambda_n}^{BC})}{Z_n}\right)\right)- \tr\left((A'\otimes\mathbf{1})
   \Omega\left(\frac{\exp(-\beta'_n H_{\Lambda_n}^p)}{Z'_n}\right)\right)\right| \\
   &=& \left|\tr\left(A' \Tr_{\Lambda_n\setminus\Lambda_m}  \Omega\left(\frac{\exp(-\beta_n H_{\Lambda_n}^{BC})}{Z_n}\right)\right)
   -\tr\left(A' \Tr_{\Lambda_n\setminus\Lambda_m}  \Omega\left(\frac{\exp(-\beta'_n H_{\Lambda_n}^p)}{Z'_n}\right)\right)\right| \\
   &\leq& \frac 1 2 \left\| \Tr_{\Lambda_n\setminus\Lambda_m}  \Omega\left(\frac{\exp(-\beta_n H_{\Lambda_n}^{BC})}{Z_n}\right)
   -\Tr_{\Lambda_n\setminus\Lambda_m}  \Omega\left(\frac{\exp(-\beta'_n H_{\Lambda_n}^p)}{Z'_n}\right) \right\|_1 \stackrel{n\to\infty}\longrightarrow 0
\end{eqnarray*}
for all $m\in\N$ according to Lemma~\ref{LemLocalGibbs}. Taking the supremum over all $A$ shows that
\[
   \lim_{n\to\infty}\left\|\frac{\exp(-\beta_n H_{\Lambda_n}^{BC})}{Z_n} - \frac{\exp(-\beta'_n H_{\Lambda_n}^p)}{Z'_n}\right\|_{\{m\}}=0\qquad\mbox{for all }m\in\N.
\]
Combining this with~(\ref{eqSpecialCase}) proves the second claim. Furthermore, Theorem~\ref{TheEquivalence} implies that
\[
   u(\omega_\beta)=\lim_{n\to\infty}\frac 1 {|\Lambda_n|} \tr(\tau'_n H_{\Lambda_n}^p) = \lim_{n\to\infty}\frac 1 {|\Lambda_n|} \tr(\tau_n H_{\Lambda_n}^p)
   =\lim_{n\to\infty}\frac 1 {|\Lambda_n|} \tr(\tau_n H_{\Lambda_n}).
\]
This completes the proof of the theorem.
\end{proof}
The simplest example application is as follows.
\begin{example}[Microcanonical versus canonical ensemble, arbitrary boundary conditions]
\label{ExFlat2}
The sequence of states $(\tau_n)_{n\in\N}$ which are defined as the maximal mixtures on the microcanonical subspaces
\[
   T_n^{BC}:={\rm span}\left\{ |E\rangle\,\,\left|\,\, H_{\Lambda_n}^{BC}|E\rangle=E|E\rangle,\enspace \frac E {|\Lambda_n|}\in(u-\delta,u)\right.\right\},
\]
where $H_{\Lambda_n}^{BC}$ is a Hamiltonian on $\Lambda_n$ with arbitrary boundary conditions, satisfies the premises
of Theorem~\ref{TheEquivalence2}. That is, if Gibbs states are unique around inverse temperature $\beta:=\beta(u)$,
we obtain equivalence of ensembles on $m$-block periodically averaged observables:
\[
   \lim_{n\to\infty}\left\| \tau_n -    \frac{\exp(-\beta H_{\Lambda_n}^{BC})} {Z_n}\right\|_{\{m\}}=0\qquad\mbox{for all }m\in\N.
\]
Furthermore, the same result is true if $\beta$ is defined as the ($n$-dependent) solution of $\frac 1 {|\Lambda_n|} \tr\left(H_{\Lambda_n}^{BC}\,
\frac{\exp(-\beta H_{\Lambda_n}^{BC})}{Z_n}\right)=u$.
\end{example}
\textbf{Remark.} The choice of boundary conditions in the definition of $T_n^{BC}$ and in the statement of the example need not be identical.
\begin{proof}
Apply Lemma~\ref{LemHBC} and $\frac 1 {|\Lambda_n|}\tr(\tau_n H_{\Lambda_n}^{BC})\leq u$ to show that $\liminf_{n\to\infty}\frac 1 {|\Lambda_n|}\left(\strut S(\tau_n)-\beta\, \tr(\tau_n H_{\Lambda_n})\right)\geq p(\beta,\Phi)$.
\end{proof}

For non-periodic boundary conditions, it is somewhat unnatural to consider periodically averaged observables. Instead, we may consider
$m$-block averaged observables, where the region $\Lambda_m$ is translated only inside the boundaries of $\Lambda_n$, without considering the periodic
extension of the latter.

\begin{definition}[$m$-block averaged observable]
For $m\leq n$, define $\mathbf{T}(\Lambda_m,\Lambda_n):=\{y\in\Z^\nu\,\,|\,\, \Lambda_m+y\subset \Lambda_n\}$.
An operator $A\in\mathcal{A}_n$ will be called an \emph{$m$-block averaged observable}
if there exists $A'\in\mathcal{A}_m$ with $A'=(A')^\dagger$ (resp.\ \emph{$m$-block averaged effect} if $0\leq A'\leq \mathbf{1}$) such that
\begin{equation}
   A=\frac 1 {|\mathbf{T}(\Lambda_m,\Lambda_n)|} \sum_{y\in\mathbf{T}(\Lambda_m,\Lambda_n)} \gamma_y(A')\otimes\mathbf{1},
   \label{eqAveraged}
\end{equation}
where the unit observable is supported on $\Lambda_n\setminus(\Lambda_m+y)$.
Moreover, we define the pseudonorm $\|\cdot\|_{[m]}$ on self-adjoint
operators $M\in\mathcal{A}_n$ by
\[
   \|M\|_{[m]}:=2\max\left\{ \left|\tr(PM)\right|\,\,|\,\, P\mbox{ is an $m$-block averaged effect on $\Lambda_n$}\right\}.
\]
\end{definition}

The following lemma translates Lemma~\ref{LemEffectToObs} to the pseudonorm $\|\cdot\|_{[m]}$ and also generalizes it.
\begin{lemma}
Let $A$ be an $m$-block averaged observable on $\Lambda_n$, coming from an observable $A'\in\mathcal{A}_m$ according to~(\ref{eqAveraged}).
Then for all quantum states $\rho,\sigma$ on $\Lambda_n$, we have
\[
   \left| \tr(\rho A)-\tr(\sigma A)\right| \leq  \|A'\|_\infty \|\rho-\sigma\|_{[m]}.
\]
Furthermore, we have $\|A\|_\infty\leq\|A'\|_\infty$; if in addition $A'\geq 0$, then we also have $\|A\|_\infty\geq \frac 1 {|\Lambda_m|}\|A'\|_\infty$. In the special case
where $|\Lambda_m|=1$, we have $\|A\|_\infty=\|A'\|_\infty$ whether or not $A'$ is positive.
\end{lemma}
\begin{proof}
The proof of the first statement is identical to that of Lemma~\ref{LemEffectToObs} and thus omitted. Clearly, $\|A\|_\infty\leq \|A'\|_\infty$ follows
directly from the definition~(\ref{eqAveraged}) and $\|\gamma_y(A')\otimes \mathbf{1}\|_\infty = \|\gamma_y(A')\|_\infty = \|A'\|_\infty$.
Since $\Lambda_m$ is a box, it can be written $\Lambda_m=[\lambda_1,\mu_1]\times\ldots\times[\lambda_\nu,\mu_\nu]$. Consider two boxes $\Lambda\subset\Lambda_n$
and $\Lambda'\subset\Lambda_n$ which are congruent to $\Lambda_m$, i.e.\ are translations of $\Lambda_m$. We call $\Lambda$ and $\Lambda'$ \emph{equivalent}
if there is a translation $y$ such that $\Lambda'=\Lambda+y$, which has components $y_i=k_i(\mu_i-\lambda_i)$ with $k_i\in\Z$. In other words, equivalent boxes
(which are shaped like $\Lambda_m$) do not overlap, and they can tesselate $\Lambda_n$ (up to sites close to the boundary).

Every equivalence class is uniquely determined by an element $x\in\Lambda_m$, which specifies a box $\Lambda$ in that equivalence class which is $\Lambda_m+y$,
where $x=(\lambda_1,\ldots,\lambda_\nu)+y$. Thus, the number of equivalence classes is upper-bounded by $|\Lambda_m|$.
Now call two translations $y,z\in\mathbf{T}(\Lambda_m,\Lambda_n)$ equivalent if $\Lambda_m+y$ is equivalent to $\Lambda_m+z$ in the sense just specified.
There will be $N$ equivalence classes $\mathbf{T}_1,\ldots,\mathbf{T}_N$, where $N\leq |\Lambda_m|$, and $\mathbf{T}(\Lambda_m,\Lambda_n)=\bigcup_{i=1}^N\mathbf{T}_i$,
which is a disjoint union. Consequently, at least one of them -- say, $\mathbf{T}_j$ -- must have $|\mathbf{T}_j|\geq |\mathbf{T}(\Lambda_m,\Lambda_n)|/N$.
For the moment, suppose that $A'$ is a positive-semidefinite matrix. Then there is a state $|\psi\rangle$ on $\Lambda_m$ such that
$\|A'\|_\infty=\langle\psi|A'|\psi\rangle$. We can write $\Lambda_n=\bigcup_{y\in \mathbf{T}_j} (\Lambda_m+y)\cup \Lambda_{rest}$, where unions are
disjoint. Now we define a state $|\Psi\rangle$ on $\Lambda_n$, by taking the tensor product of copies of $|\psi\rangle$ in the regions $\Lambda_m+y$, and
an arbitrary pure reference state $|0\rangle$ on $\Lambda_{rest}$. We get
\begin{eqnarray*}
   \|A||_\infty&\geq& \langle\Psi|A|\Psi\rangle = \frac 1 {|\mathbf{T}(\Lambda_m,\Lambda_n)|} \sum_{i=1}^N \sum_{y\in\mathbf{T}_i} \langle\Psi|\gamma_y(A')\otimes\mathbf{1}|\Psi\rangle
   \geq \frac 1 {|\mathbf{T}(\Lambda_m,\Lambda_n)|} \sum_{y\in\mathbf{T}_j} \langle\Psi|\gamma_y(A')\otimes\mathbf{1}|\Psi\rangle\\
   &=& \frac {|\mathbf{T}_j|}{|\mathbf{T}(\Lambda_m,\Lambda_n)|} \langle \psi|A'|\psi\rangle \geq \frac 1 N \|A'\|_\infty\geq \frac 1 {|\Lambda_m|}\|A'\|_\infty.
\end{eqnarray*}
If $|\Lambda_m|=1$, choose the single-site state $|\psi\rangle$ such that $\|A'\|_\infty=|\langle\psi|A'|\psi\rangle|$. Let $|\psi^{\otimes\Lambda_n}\rangle$ be the state $|\psi\rangle$,
copied onto every lattice site of $\Lambda_n$. Then
\[
   \|A\|_\infty\geq \left|\langle \psi^{\otimes\Lambda_n}|A|\psi^{\otimes\Lambda_n}\rangle\right|=\left|\frac 1 {|\mathbf{T}(\Lambda_m,\Lambda_n)|}
   \sum_{y\in\mathbf{T}(\Lambda_m,\Lambda_n)} \langle \psi^{\otimes\Lambda_n} | \gamma_y(A')\otimes\mathbf{1}|\psi^{\otimes\Lambda_n}\rangle\right|
   =|\langle\psi|A'|\psi\rangle| = \|A'\|_\infty.
\]
The claim follows.
\end{proof}

Asymptotically, that is for large $n$, the pseudonorms $\|\cdot\|_{\{m\}}$ and $\|\cdot\|_{[m]}$ are equivalent. This is the statement of the following lemma.
Thus, our equivalence of ensemble results in Theorem~\ref{TheEquivalence2} and Example~\ref{ExFlat2} remain valid of the former pseudonorm is replaced
by the latter. This yields a more natural physical interpretation of our results.

\begin{lemma}[Equivalence of both averaging methods]
\label{LemNormEquivalence}
For every $m\leq n$ and all states $\rho,\sigma$ on $\Lambda_n$, we have
\[
   \left|\, \|\rho-\sigma\|_{\{m\}} - \|\rho-\sigma\|_{[m]}\right| \leq 8 |\Lambda_m|\cdot\frac{|\partial \Lambda_n|}{|\Lambda_n|}
\]
which tends to zero for fixed $m$ as $n\to\infty$.
\end{lemma}
\begin{proof}
Define the completely positive map $\Phi:\mathcal{A}_m\to\mathcal{A}_n$ by setting $\Phi(A')$ as the right-hand side of~(\ref{eqPeriodicallyAveraged}).
Similarly, define the completely positive map $\Phi':\mathcal{A}_m\to\mathcal{A}_n$ by setting $\Phi'(A')$ as the right-hand side of~(\ref{eqAveraged}).
Note that $\Phi(\mathbf{1})=\mathbf{1}=\Phi'(\mathbf{1})$.
Then $\|M\|_{\{m\}} = 2\max_{0\leq A'\leq\mathbf{1}} |\tr(M\Phi(A'))|$ and $\|M\|_{[m]} = 2\max_{0\leq A'\leq\mathbf{1}} |\tr(M\Phi'(A'))|$. If $M$ is
traceless (as is the case for $M=\rho-\sigma$), then $\tr(M\Phi(\mathbf{1}-A'))=-\tr(M\Phi(A'))$ and similarly for $\Phi'$, and the absolute values under the maxima
can be removed. Thus
\begin{eqnarray*}
   \left| \, \|M\|_{\{m\}} - \|M\|_{[m]}\right| &=& 2 \left| \max_{0\leq A'\leq \mathbf{1}} \tr(M\Phi(A')) - \max_{0\leq A'\leq \mathbf{1}} \tr(M\Phi'(A'))\right|
   \leq 2  \max_{0\leq A'\leq \mathbf{1}} \left| \tr(M\Phi(A')) - \tr(M\Phi'(A'))\right| \\
   &\leq& 2 \|M\|_1 \max_{0\leq A'\leq \mathbf{1}} \left\| \Phi(A')-\Phi'(A')\right\|_\infty.
\end{eqnarray*}
To compare $\Phi$ and $\Phi'$, we note that we can interpret every translation $y\in\mathbf{T}(\Lambda_m,\Lambda_n)$ as a periodic translation
$T\in\mathbf{T}(\Lambda_n)$ such that $\gamma_y(A')\otimes \mathbf{1}=T(A'\otimes\mathbf{1})T^\dagger$ for every $A'\in\mathcal{A}_m$; this is
an equality of observables on $\Lambda_n$. In this sense, we can write $\mathbf{T}(\Lambda_m,\Lambda_n)\subset \mathbf{T}(\Lambda_n)$.
A simple application of the triangle inequality and $\|A'\|_\infty\leq 1$ gives
\begin{eqnarray*}
   \left\|\Phi(A')-\Phi'(A')\right\|_\infty &\leq& \left\| \frac 1 {|\mathbf{T}(\Lambda_n)|} \sum_{T\in\mathbf{T}(\Lambda_n)} T(A'\otimes\mathbf{1})T^\dagger
   -\frac 1 {|\mathbf{T}(\Lambda_n)|} \sum_{y\in \mathbf{T}(\Lambda_m,\Lambda_n)} \gamma_y(A')\otimes \mathbf{1}\right\|_\infty \\
   && + \left\| \frac 1 {|\mathbf{T}(\Lambda_n)|} \sum_{y\in \mathbf{T}(\Lambda_m,\Lambda_n)} \gamma_y(A')\otimes \mathbf{1}
   -\frac 1 {|\mathbf{T}(\Lambda_m,\Lambda_n)|} \sum_{y\in \mathbf{T}(\Lambda_m,\Lambda_n)} \gamma_y(A')\otimes \mathbf{1}
   \right\|_\infty \\
   &=& \frac 1 {|\mathbf{T}(\Lambda_n)|} \left\|\sum_{T\in\mathbf{T}(\Lambda_n)\setminus \mathbf{T}(\Lambda_m,\Lambda_n)} T(A'\otimes\mathbf{1})T^\dagger\right\|_\infty
   +\left(\frac 1 {|\mathbf{T}(\Lambda_m,\Lambda_n)|} - \frac 1 {|\mathbf{T}(\Lambda_n)|}\right)\left\|\sum_{y\in\mathbf{T}(\Lambda_m,\Lambda_n)}
   \gamma_y(A')\otimes\mathbf{1}\right\|_\infty \\
   &\leq& 2 \frac{|\mathbf{T}(\Lambda_n)|-|\mathbf{T}(\Lambda_m,\Lambda_n)|}{|\mathbf{T}(\Lambda_n)|}.
\end{eqnarray*}
Estimating this expression is a matter of simple lattice geometry. First, it is easy to see that $|\mathbf{T}(\Lambda_n)|=|\Lambda_n|$, the number of sites in the region.
Consider any translation $T\in\mathbf{T}(\Lambda_n)\setminus \mathbf{T}(\Lambda_m,\Lambda_n)$. It translates $\Lambda_m$ periodically inside $\Lambda_n$,
but \emph{not} in a way such that the same is achieved by a non-periodic translation $\gamma_y$ with $y\in\Z^\nu$. Instead, the corresponding $y$-translation will
map $\Lambda_m$ partially inside and partially outside of $\Lambda_n$. That is, there must be some intersection of $y+\Lambda_m$ with
the boundary of $\Lambda_n$ defined in~(\ref{eqDefBoundary}). However, for every given boundary point $x\in\partial \Lambda_n$, there are only $|\Lambda_m|$ many
translations $y$ such that $x\in \Lambda_m+y$. Hence
\[
   |\mathbf{T}(\Lambda_n)|-|\mathbf{T}(\Lambda_m,\Lambda_n)| = |\mathbf{T}(\Lambda_n)\setminus \mathbf{T}(\Lambda_m,\Lambda_n)|
   \leq |\partial \Lambda_n|\cdot |\Lambda_m|.
\]
Combining the previous inequalities, and using that $\|\rho-\sigma\|_1\leq 2$, completes the proof.
\end{proof}

\subsection{Canonical typicality}
With the results of the previous subsection, in particular Examples~\ref{ExFlat} and~\ref{ExFlat2}, it is easy to prove a general result on
canonical typicality for translation-invariant quantum systems.
\begin{theorem}[Canonical typicality, periodic boundary conditions]
\label{TheCanonicalTypicalityPeriodic}
Let $\Phi$ be any translation-invariant finite-range interaction, not physically equivalent to zero, with corresponding periodic boundary condition Hamiltonians
$H_{\Lambda_n}^p$, let $u_{\min}(\Phi)<u\leq u_{\max}(\Phi)$
and $\delta>0$.
Suppose that there is a unique infinite-volume Gibbs state $\omega_\beta$ at inverse temperature $\beta\equiv\beta(u)$. Consider the microcanonical subspace
\[
   T_n^p:={\rm span}\left\{ |E\rangle\,\,\left|\,\, H_{\Lambda_n}^p|E\rangle=E|E\rangle,\enspace \frac E {|\Lambda_n|}\in(u-\delta,u)\right.\right\}.
\]
If $|\psi\rangle\in T_n^p$ is a random pure state, then for every $m\in\N$ there is a sequence of positive real numbers $(\Delta_{m,n})_{n\in\N}$ with $\lim_{n\to\infty}
\Delta_{m,n}=0$, such that
\[
   {\rm Prob}\left\{\left\| {\rm Tr}_{\Lambda_n\setminus\Lambda_m} |\psi\rangle\langle\psi| - {\rm Tr}_{\Lambda_n\setminus\Lambda_m}
   \frac{\exp(-\beta H_{\Lambda_n}^p)} {Z_n}\right\|_1 \geq \Delta_{m,n} + \varepsilon
   \right\}\leq \exp\left(\strut -\varepsilon^2 \exp(|\Lambda_n|\, s(u)+o(|\Lambda_n|))\right)
\]
for every $\varepsilon\geq 0$. Furthermore, if Gibbs states are unique around inverse temperature $\beta>0$, then the same result is true if $\beta$ is chosen
as the ($n$-dependent) solution of $\frac 1 {|\Lambda_n|} \tr\left(H_{\Lambda_n}^{BC}\,\frac{\exp(-\beta H_{\Lambda_n}^p)}{Z_n}\right)=u$, where BC denotes
an arbitrary fixed choice of boundary conditions.
\end{theorem}
\begin{proof}
It follows from~\cite[Theorem 1]{Popescu} that
\[
   {\rm Prob}\left\{\left\| {\rm Tr}_{\Lambda_n\setminus \Lambda_m}|\psi\rangle\langle\psi| - \Omega_{m,n}\right\|_1 \geq \varepsilon+\frac{d^{|\Lambda_m|}}{\sqrt{|T_n^p|}}\right\}
   \leq 2 \exp\left( - \frac{|T_n^p|\varepsilon^2}{18\pi^3}\right)
\]
for all $\varepsilon\geq 0$, where $|T_n^p|$ denotes the dimension of the subspace $T_n^p$, and
$\Omega_{m,n}:={\rm Tr}_{\Lambda_n\setminus\Lambda_m}\tau_n$, with $\tau_n$ the maximally mixed state on $T_n^p$. Set
\begin{equation}
   \delta_{m,n}:=\left\| \Omega_{m,n} - {\rm Tr}_{\Lambda_n\setminus\Lambda_m}
   \frac{\exp(-\beta H_{\Lambda_n}^p)} {Z_n}\right\|_1,
   \label{eqdeltamn}
\end{equation}
then Example~\ref{ExFlat} resp.\ Theorem~\ref{TheEquivalence} imply that $\lim_{n\to\infty}\delta_{m,n}=0$. Thus, the previous statements imply
\[
   {\rm Prob}\left\{\left| {\rm Tr}_{\Lambda_n\setminus\Lambda_m} |\psi\rangle\langle\psi| - {\rm Tr}_{\Lambda_n\setminus\Lambda_m} \frac{\exp(-\beta H_{\Lambda_n}^p)}{Z_n}\right\|_1
   \geq \varepsilon+\frac{d^{|\Lambda_m|}}{\sqrt{|T_n^p|}}+\delta_{m,n}\right\} \leq 2 \exp\left( - \frac{|T_n^p|\varepsilon^2}{18\pi^3}\right).
\]
Furthermore, according to Lemma~\ref{LemHBC}, we have $|T_n^p|=\exp[|\Lambda_n|\, s(u)+o(|\Lambda_n|)]$.
Setting
\begin{equation}
   \Delta_{m,n}:=\delta_{m,n}+d^{|\Lambda_m|}/\sqrt{|T_n^p|}
   \label{eqdeltaDelta}
\end{equation}
completes the proof of the theorem.
\end{proof}
Example~\ref{ExNoLocal} shows again that we cannot in general replace the restriction of the global Gibbs state,
${\rm Tr}_{\Lambda_n\setminus\Lambda_m}\exp(-\beta H_{\Lambda_n}^p)/Z_n$, with the local Gibbs state, $\exp(-\beta H_{\Lambda_m}^{BC})/Z_m$,
no matter what boundary conditions we choose for $H_{\Lambda_m}^{BC}$.

Similarly as for our equivalence of ensembles result, we can prove an analogue of this theorem in the case of arbitrary boundary conditions
by replacing $\|\cdot\|_1$ by $\|\cdot\|_{\{m\}}$.

\begin{theorem}[Canonical typicality, arbitrary boundary conditions]
\label{TheCanonicalTypicalityArbitrary}
Let $\Phi$ be any translation-invariant finite-range interaction, not physically equivalent to zero,
with corresponding arbitrary boundary condition Hamiltonians $H_{\Lambda_n}^{BC}$, let $u_{\min}(\Phi)<u\leq u_{\max}(\Phi)$
and $\delta>0$. Suppose that there is a unique infinite-volume Gibbs state $\omega_\beta$ at inverse temperature $\beta\equiv\beta(u)$.
Consider the microcanonical subspace
\[
   T_n^{BC}:={\rm span}\left\{ |E\rangle\,\,\left|\,\, H_{\Lambda_n}^{BC}|E\rangle=E|E\rangle,\enspace \frac E {|\Lambda_n|}\in(u-\delta,u)\right.\right\}.
\]
If $|\psi\rangle\in T_n^{BC}$ is a random pure state, then for every $m\in\N$ there is a sequence of positive real numbers $(\Delta_{m,n})_{n\in\N}$ with $\lim_{n\to\infty}
\Delta_{m,n}=0$, such that
\[
   {\rm Prob}\left\{\left\| |\psi\rangle\langle\psi| - 
   \frac{\exp(-\beta H_{\Lambda_n}^{BC})} {Z_n}\right\|_{[m]} \geq \Delta_{m,n} + \varepsilon
   \right\}\leq \exp\left(\strut -\varepsilon^2 \exp(|\Lambda_n|\, s(u)+o(|\Lambda_n|))\right)
\]
for every $\varepsilon\geq 0$. Furthermore, if Gibbs states are unique around inverse temperature $\beta>0$, then the same result is true if $\beta$ is chosen
as the ($n$-dependent) solution of $\frac 1 {|\Lambda_n|} \tr\left(H_{\Lambda_n}^{BC}\,\frac{\exp(-\beta H_{\Lambda_n}^{BC})}{Z_n}\right)=u$.
\end{theorem}
\begin{proof}
Denote by $\tau_n^{BC}$ the maximally mixed state on $T_n^{BC}$.
Suppose that $\eta\geq 0$ is any real number such that
\begin{equation}
   \left\| \,|\psi\rangle\langle\psi|-\tau_n^{BC}\right\|_{[m]} \geq \eta.
   \label{eqEvent}
\end{equation}
By definition, this means that there exists some observable $A'\in\mathcal{A}_m$ such that
\[
   2\left| \frac 1 {|\mathbf{T}(\Lambda_m,\Lambda_n)|} \sum_{y\in\mathbf{T}(\Lambda_m,\Lambda_n)} \left(
      \langle \psi|\gamma_y(A')\otimes\mathbf{1}|\psi\rangle - \tr(\tau_n^{BC} \gamma_y(A')\otimes\mathbf{1}
   \right)\right|\geq\eta,
\]
and thus, there must be some $y\in\mathbf{T}(\Lambda_m,\Lambda_n)$ such that
\[
   2\left| 
      \langle \psi|\gamma_y(A')\otimes\mathbf{1}|\psi\rangle - \tr(\tau_n^{BC} \gamma_y(A')\otimes\mathbf{1})
   \right|\geq \eta.
\]
Let $\Lambda:=\Lambda_m+y$, then $|\Lambda|=|\Lambda_m|$, $\Lambda\subset\Lambda_n$, and
\[
   \left\| \Tr_{\Lambda_n\setminus\Lambda} |\psi\rangle\langle\psi| - \Tr_{\Lambda_n\setminus\Lambda} \tau_n^{BC}\right\|_1 \geq \eta.
\]
Now consider the case $\eta=\varepsilon+d^{|\Lambda_m|}/\sqrt{|T_n^{BC}|}$. According to~\cite[Theorem 1]{Popescu}, the probability that the previous
inequality holds on Haar-random choice of $|\psi\rangle$ is upper-bounded by $2 \exp\left( - |T_n^{BC}|\varepsilon^2/(18\pi^3)\right)$. Thus
\[
   {\rm Prob}\left\{\left\|\, |\psi\rangle\langle\psi| - \tau_n^{BC}\right\|_{[m]} \geq \varepsilon+\frac{d^{|\Lambda_m|}}{\sqrt{|T_n^p|}}\right\}
   \leq 2 \exp\left( - \frac{|T_n^{BC}|\varepsilon^2}{18\pi^3}\right).
\]
Now set
\[
   \delta_{m,n}:=\left\|\tau_n^{BC}-\frac{\exp(-\beta H_{\Lambda_n}^{BC})}{Z_n}\right\|_{[m]},
\]
and set $\Delta_{m,n}:=\delta_{m,n}+d^{|\Lambda_m|}/\sqrt{|T_n^{BC}|}$. Example~\ref{ExFlat2} and Lemma~\ref{LemNormEquivalence} show
that $\lim_{n\to\infty}\delta_{m,n}=0=\lim_{n\to\infty}\Delta_{m,n}$,
and arguing as in the proof of Theorem~\ref{TheCanonicalTypicalityPeriodic} completes the proof.
\end{proof}

Drawing a pure state $|\psi\rangle$ according to the Haar measure is a process that cannot be achieved efficiently in practice, as parameter
counting shows. Thus, it is also to be expected that no process in nature really produces a Haar-random state. However, what \emph{can} be achieved efficiently -- for
example, by application of random local unitaries~\cite{Brandao} -- are approximations to the Haar measure known as (approximate) \emph{unitary $t$-designs}.
As shown in~\cite{Low}, they give a way to ``derandomize'' results like the canonical typicality theorems above.

There are different definitions of what is called an $\varepsilon$-approximate $k$-design $\nu$; they all have in common that
the computational effort of sampling from them scales polynomially in
$\log\varepsilon$ and $\log d$, where $d$ is the underlying Hilbert space dimension.

Here, we use the definition from~\cite{Low}. It utilizes the notion of a \emph{balanced monomial of degree $k$} of a matrix $U$, which
is a monomial in the components of $U$ and $U^\dagger$ which contains the same number ($k$) of conjugated as unconjugated elements.
For example, $U_{ij} U^\ast_{pq}$ is a balanced monomial of degree $1$.

\begin{definition}[Approximate design]
A measure $\nu$ on the unitary group $U(d)$ is called an $\varepsilon$-approximate (unitary) $k$-design, if for all balanced monomials $M$ of degree
less than or equal to $k$, we have
\[
   |\mathbb{E}_{U\sim \nu} M(U) - \mathbb{E}_{U\sim \mu_H} M(U)| \leq\frac\varepsilon {d^k},
\]
where $\mathbb{E}_{U\sim\mu}$ denotes the expectation with respect to a measure $\mu$, and $\mu_H$ is the Haar measure.
\end{definition}

We now use Theorem 1.4 in~\cite{Low} to prove a derandomized version of canonical typicality. Note that the theorem in~\cite{Low} uses as an implicit
additional assumption that $k$ is an integer-multiple of $8$.

\begin{theorem}[Canonical typicality, periodic boundary conditions, derandomized version]
\label{ThePeriodicBCDerandomized}
Let $\Phi$ be any translation-invariant finite-range interaction, not physically equivalent to zero,
with corresponding periodic boundary condition Hamiltonians $H_{\Lambda_n}^p$, let $u_{\min}(\Phi)<u\leq u_{\max}(\Phi)$
and $\Delta>0$. Suppose that there is a unique infinite-volume Gibbs state $\omega_\beta$ at inverse temperature $\beta\equiv\beta(u)$.
Consider the microcanonical subspace
\[
   T_n^p:={\rm span}\left\{ |E\rangle\,\,\left|\,\, H_{\Lambda_n}^p|E\rangle=E|E\rangle,\enspace \frac E {|\Lambda_n|}\in(u-\delta,u)\right.\right\}.
\]
Choose a state $|\psi\rangle$ at random from $T_n^p$ by choosing a unitary from an $\varepsilon$-approximate $8$-design and applying it to
a fixed initial pure state, where $\varepsilon=\exp(-|\Lambda_n| s(u)+o(|\Lambda_n|))$.
Then for every $m\in\N$ large enough such that $d^{|\Lambda_m|}\geq 14$,
there is a sequence of positive real numbers $(\delta_{m,n})_{n\in\N}$ with $\lim_{n\to\infty}\delta_{m,n}=0$, such that
\[
   {\rm Prob}_\nu\left\{\left\| {\rm Tr}_{\Lambda_n\setminus\Lambda_m} |\psi\rangle\langle\psi| - {\rm Tr}_{\Lambda_n\setminus\Lambda_m}
   \frac{\exp(-\beta H_{\Lambda_n}^p)} {Z_n}\right\|_1 \geq \delta_{m,n} + \kappa  \right\}\leq \frac{ d^{3|\Lambda_m|}}{\kappa^2}
   \exp\left(\strut -|\Lambda_n| s(u)+o(|\Lambda_n|)\right)
\]
for all $\kappa>0$. Furthermore, if Gibbs states are unique around inverse temperature $\beta>0$, then the same result is true if $\beta$ is chosen
as the ($n$-dependent) solution of $\frac 1 {|\Lambda_n|} \tr\left(H_{\Lambda_n}^{BC}\,\frac{\exp(-\beta H_{\Lambda_n}^p)}{Z_n}\right)=u$, where BC denotes
an arbitrary fixed choice of boundary conditions.
\end{theorem}
\begin{proof}
Let $\tau_n$ be the maximally mixed state on $T_n^p$, and set $\varepsilon:=6 d^{3|\Lambda_m|}/|T_n^p|$. Due to~\cite[Theorem 1.4]{Low}, we have
\begin{equation}
   {\rm Prob}_\nu\left\{\left\|\Tr_{\Lambda_n\setminus\Lambda_m} |\psi\rangle\langle\psi| - \Tr_{\Lambda_n\setminus\Lambda_m} \tau_n\right\|_1\geq\kappa\right\}
   \leq\frac{24 d^{3|\Lambda_m|}}{|T_n^p|\kappa^2}
   \label{eqDerandom}
\end{equation}
for all $\kappa>0$. Define $\delta_{m,n}$ as in~(\ref{eqdeltamn}), use Example~\ref{ExFlat} and absorb the factor $24$ into the $\exp(o|\Lambda_n|)$-term.
\end{proof}

One still has concentration on the thermal state; however, in contrast to the Haar measure result in Theorem~\ref{TheCanonicalTypicalityPeriodic},
the concentration is now exponential in the number of sites $|\Lambda_n|$, not doubly-exponential. This behavior is more in line with standard
expectations on physical systems in statistical mechanics.

It is now clear how Theorem~\ref{TheCanonicalTypicalityArbitrary} can be derandomized, by imitating the proof of Theorem~\ref{TheCanonicalTypicalityArbitrary}
in conjunction with the $T_n^{BC}$-analogue of~(\ref{eqDerandom}) and the inequality $\|\cdot\|_{[m]}\leq \|\cdot\|_1$. We omit the details.

\begin{theorem}[Canonical typicality, arbitrary boundary conditions, derandomized version]
\label{TheCanonTypArbitraryDesigns}
Let $\Phi$ be any translation-invariant finite-range interaction with corresponding arbitrary boundary condition Hamiltonians $H_{\Lambda_n}^{BC}$, let $u_{\min}(\Phi)<u\leq u_{\max}(\Phi)$
and $\delta>0$. Suppose that there is a unique infinite-volume Gibbs state $\omega_\beta$ at inverse temperature $\beta\equiv\beta(u)$.
Consider the microcanonical subspace
\[
   T_n^{BC}:={\rm span}\left\{ |E\rangle\,\,\left|\,\, H_{\Lambda_n}^{BC}|E\rangle=E|E\rangle,\enspace \frac E {|\Lambda_n|}\in(u-\delta,u)\right.\right\}.
\]
Choose a state $|\psi\rangle$ at random from $T_n^{BC}$ by choosing a unitary from an $\varepsilon$-approximate $8$-design and applying it to
a fixed initial pure state, where $\varepsilon=\exp(-|\Lambda_n| s(u)+o(|\Lambda_n|))$.
Then for every $m\in\N$ large enough such that $d^{|\Lambda_m|}\geq 14$,
there is a sequence of positive real numbers $(\delta_{m,n})_{n\in\N}$ with $\lim_{n\to\infty}\delta_{m,n}=0$, such that
\[
   {\rm Prob}_\nu\left\{\left\|  |\psi\rangle\langle\psi| - 
   \frac{\exp(-\beta H_{\Lambda_n}^{BC})} {Z_n}\right\|_{[m]} \geq \delta_{m,n} + \kappa  \right\}\leq \frac{ d^{3|\Lambda_m|}}{\kappa^2}
   \exp\left(\strut -|\Lambda_n| s(u)+o(|\Lambda_n|)\right)
\]
for all $\kappa>0$. Furthermore, if Gibbs states are unique around inverse temperature $\beta>0$, then the same result is true if $\beta$ is chosen
as the ($n$-dependent) solution of $\frac 1 {|\Lambda_n|} \tr\left(H_{\Lambda_n}^{BC}\,\frac{\exp(-\beta H_{\Lambda_n}^{BC})}{Z_n}\right)=u$.
\end{theorem}

Since the effort of sampling from an $\varepsilon$-approximate $8$-design $\nu$ scales polynomially in $\log\varepsilon$ and the logarithm of the Hilbert space
dimension, we obtain that sampling from $\nu$ in the theorems above amounts to an effort that grows only polynomially in $|\Lambda_n|$, i.e.\ the particle
number.

\subsection{Dynamical thermalization}
We can apply the previous results to obtain statements about dynamical thermalization, using the results of~\cite{ShortFarrelly} which are elaborations
of earlier results in~\cite{LindenPopescu} and~\cite{Short}. However, for the technicalities,
we need to relate the von Neumann entropy with the R\'enyi entropy of order two.
For $\alpha>0$ with $\alpha\neq 1$ and density matrices $\rho$, we define~\cite{Petz}
\[
   S_\alpha(\rho):=\frac 1 {1-\alpha}\log \tr(\rho^\alpha),
\]
and the limit $\alpha\to 1$ recovers von Neumann entropy, $S_1(\rho):=S(\rho)=-\tr(\rho\log\rho)$, and the limit $\alpha\to 0$ yields
$S_0(\rho):=\log{\rm rank}(\rho)$.
If $\alpha\leq\alpha'$ then $S_\alpha\geq S_{\alpha'}$. In fact, we will use R\'enyi entropy only for classical probability vectors $\lambda=(\lambda_1,\ldots,\lambda_N)$,
and write sloppily $S_\alpha(\lambda)$ for classical R\'enyi entropy, which is the same as the quantum R\'enyi entropy of the diagonal matrix with entries $\lambda_i$.
We use some inequalities and insights from~\cite{Zyczkowski} to show the following:
\begin{lemma}
\label{LemRenyi}
For every $0\leq\varepsilon\leq 1$, we have $\displaystyle S_2(\rho)\geq 2\varepsilon\left(S(\rho)-\frac{\varepsilon}{1+\varepsilon} S_0(\rho)\right)\geq
2\varepsilon(S(\rho)-\varepsilon S_0(\rho))$.
\end{lemma}
\begin{proof}
As shown in~\cite{Zyczkowski}, we have $\frac{\partial}{\partial q} \frac{q-1} q S_q \geq 0$,
hence $\frac{q-1} q S_q \leq \frac 1 2 S_2$ for all $q\in[1,2]$. Since the function $q\mapsto S_q$ is convex, the value of $S_q$ lies on or above the line
$g(x):=S_0-(S_0-S_1)x$ that connects $S_0$ and $S_1$, i.e.\ $S_q\geq g(q)=S_0-(S_0-S_1)q$. We get
\[
   S_2\geq \frac{2(q-1)} q S_q \geq \frac{2(q-1)} q \left(\strut S_0-(S_0-S_1)q\right) = 2(q-1) S_1 - \frac{2(q-1)^2} q S_0.
\]
Setting $q=:1+\varepsilon$ proves the claim.
\end{proof}

Following~\cite{ShortFarrelly}, for any Hamiltonian $H$, we define its \emph{gap degeneracy} by
\[
   D_G(H):=\max_E \left|\strut \{ (i,j)\,\,|\,\, i\neq j,\enspace E_i - E_j=E\}\right|,
\]
where the $E_i$ denotes the (energy) eigenvalues of $H$. Using Theorem 3 of~\cite{ShortFarrelly}, we can easily show the following.

\begin{theorem}[Thermalization, periodic boundary conditions]
\label{TheThermalizationPeriodic}
Let $\Phi$ be a translation-invariant finite-range interaction
which is not physically equivalent to zero, and $(\rho_0^{(n)})_{n\in\N}$ any sequence of initial states on $\Lambda_n$ which have energy expectation value
of $U_n:=\tr(\rho_0^{(n)}H_{\Lambda_n}^p)$ with density $U_n/|\Lambda_n|$ converging to some value $u\in(u_{\min}(\Phi),u_{\max}(\Phi))$ as $n\to\infty$.

Suppose that the initial states have close to maximal ``population entropy'' in the following sense. Define
$\bar S(\rho_0^{(n)}):=S(\lambda_1,\ldots,\lambda_N)$, where $S$ is Shannon entropy, and $\lambda_i:=\tr(\rho_0^{(n)} \pi_i)$ is the probability that the
$i$-th energy level is populated, where $H_{\Lambda_n}^p=:\sum_{i=1}^N E_i \pi_i$
is the spectral decomposition. Furthermore, suppose that either $H_{\Lambda_n}^p$ is non-degenerate, or that every $\pi_i \rho_0^{(n)} \pi_i$ is $\Lambda_n$-translation-invariant.
Then, determine the corresponding inverse temperature $\beta_n$ for which
\[
   \tr(H_{\Lambda_n}^p \gamma_{\Lambda_n}^p(\beta_n))=U_n,\qquad\mbox{where }\gamma_{\Lambda_n}^p(\beta_n):=\frac{\exp(-\beta_n H_{\Lambda_n}^p)}{Z_n}.
\]
If the initial states have close to maximal population entropy in the sense that
\begin{equation}
   \bar S(\rho_0^{(n)})\geq S(\gamma_{\Lambda_n}^p(\beta_n))-o(|\Lambda_n|),
   \label{eqMaxEnt}
\end{equation}
then unitary time evolution $\rho^{(n)}(t):=\exp(-itH_{\Lambda_n}^p)\rho_0^{(n)}\exp(itH_{\Lambda_n}^p)$ thermalizes the subsystem $\Lambda_m$ for most times $t$:
\begin{eqnarray*}
   \left\langle \left\| \Tr_{\Lambda_n\setminus\Lambda_m}\rho^{(n)}(t)-\left\langle\Tr_{\Lambda_n\setminus\Lambda_m}  \rho^{(n)}(t)\right\rangle\right\|_1\right\rangle
   &\leq& d^{|\Lambda_m|} \sqrt{D_G(H_{\Lambda_n}^p)} \exp\left( - \frac{s(\omega_\beta)^2}{4\log d} |\Lambda_n| + o(|\Lambda_n|)\right),\qquad\mbox{ and}\\
   \lim_{n\to\infty}\left\|\left\langle\Tr_{\Lambda_n\setminus\Lambda_m} \rho^{(n)}(t)\right\rangle - \Tr_{\Lambda_n\setminus\Lambda_m} \frac{\exp(-\beta_n H_{\Lambda_n}^p)}{Z_n}\right\|_1
   &=& 0,
\end{eqnarray*}
where $Z_n=\tr(\exp(-\beta_n H_{\Lambda_n}^p))$, and $\langle\cdot\rangle$ denotes the average over all times $t\geq 0$. Furthermore, in this statement, $\beta_n$ can
be replaced by $\beta:=\beta(u)$.
\end{theorem}
\textbf{Remark.} If $H_{\Lambda_n}^p$ is non-degenerate, we have $\bar S(\rho_0^{(n)})=S(\bar \rho_0^{(n)})$, where $\bar\rho_0^{(n)}:=\sum_i \pi_i\rho_0^{(n)}\pi_i$ is
the dephased initial state. Furthermore, we can summarize the result by saying that
\[
   \left\langle \left\|\Tr_{\Lambda_n\setminus\Lambda_m}\rho^{(n)}(t) - \Tr_{\Lambda_n\setminus\Lambda_m}\frac{\exp(-\beta_n H_{\Lambda_n}^p)}{Z_n}\right\|_1
   \right\rangle\stackrel{n\to\infty}\longrightarrow 0
\]
as long as the gap degeneracy $D_G$ grows at most subexponentially with $|\Lambda_n|$. However, the more detailed formulation above contains more information:
while the difference to the Gibbs state may tend to zero polynomially in $|\Lambda_n|$, the result shows strong equilibration of time evolution indicated by a trace distance
which goes to zero exponentially in $|\Lambda_n|$.
\begin{proof}
According to~\cite[Theorem 3 resp.~(25)]{ShortFarrelly}, we have
\begin{equation}
   \left\langle \left\|\Tr_{\Lambda_n\setminus\Lambda_m}\rho^{(n)}(t)-(\bar\rho_0^{(n)})_{\Lambda_m}\right\|_1\right\rangle \leq d^{|\Lambda_m|}
   \sqrt{\frac{D_G(H_{\Lambda_n}^p)}{d_{\rm eff}}},
   \label{eqTonyResult}
\end{equation}
where $d_{\rm eff}^{-1}=\sum_i \lambda_i^2$, thus $d_{\rm eff}=\exp(S_2(\lambda_1,\ldots,\lambda_N))$, and $\bar\rho_0^{(n)}=\langle\rho^{(n)}(t)\rangle=\sum_{i=1}^N\pi_i\rho_0^{(n)}\pi_i$.
If $H_{\Lambda_n}^p$ is non-degenerate, then every $\pi_i\rho_0^{(n)}\pi_i$ is a real multiple of $\pi_i$ and thus $\Lambda_n$-translation-invariant. Thus, the conditions of
the lemma ensure that $\bar\rho_0^{(n)}$ is $\Lambda_n$-translation-invariant. Since the $\pi_i\rho_0^{(n)}\pi_i/\lambda_i$ for $\lambda_i\neq 0$ are density matrices with mutually
orthogonal supports, we have
\[
   S(\bar\rho_0^{(n)})=S\left(\sum_{i:\lambda_i\neq 0}\lambda_i\,\frac{\pi_i\rho_0^{(n)}\pi_i}{\lambda_i}\right) = S(\lambda_1,\ldots,\lambda_N)+\sum_{i:\lambda_i\neq 0}
   \lambda_i S\left(\frac{\pi_i \rho_0^{(n)} \pi_i}{\lambda_i}\right)\geq S(\lambda_1,\ldots,\lambda_N)=\bar S(\rho_0^{(n)}).
\]
Note that $\tr(\rho_0^{(n)}H_{\Lambda_n}^p)=\tr(\bar \rho_0^{(n)}H_{\Lambda_n}^p)$. Furthermore, Theorem~\ref{TheEquivalence} shows that $\lim_{n\to\infty}\beta_n=\beta
:=\beta(u)$, and $\rho(\beta_n)$ maximizes the functional $\rho\mapsto S(\rho)-\beta_n \tr(\rho H_{\Lambda_n}^p)$. Thus
\begin{eqnarray}
   \liminf_{n\to\infty}\frac 1 {|\Lambda_n|} \left( S(\bar\rho_0^{(n)})-\beta\, \tr(\bar\rho_0^{(n)}H_{\Lambda_n})\right)&\geq&
   \liminf_{n\to\infty}\frac 1 {|\Lambda_n|} \left( \bar S(\rho_0^{(n)})-\beta\, \tr(\bar\rho_0^{(n)}H_{\Lambda_n}^p)\right)\nonumber \\
   &\geq&  \liminf_{n\to\infty}\frac 1 {|\Lambda_n|} \left( S(\gamma_{\Lambda_n}^p(\beta_n))-\beta\, \tr(\rho_0^{(n)}H_{\Lambda_n}^p)\right)\nonumber \\
   &=& \liminf_{n\to\infty}\frac 1 {|\Lambda_n|} \left( S(\gamma_{\Lambda_n}^p(\beta_n))-\beta\, \tr(\gamma_{\Lambda_n}^p(\beta_n) H_{\Lambda_n}^p)\right) \nonumber \\
   &=& \liminf_{n\to\infty}\frac 1 {|\Lambda_n|} \left( S(\gamma_{\Lambda_n}^p(\beta_n))-\beta_n \tr(\gamma_{\Lambda_n}^p(\beta_n) H_{\Lambda_n}^p)\right)\nonumber \\
   &\geq& \liminf_{n\to\infty}\frac 1 {|\Lambda_n|} \left( S((\omega_\beta)_{\Lambda_n})-\beta_n \tr((\omega_\beta)_{\Lambda_n} H_{\Lambda_n}^p)\right)\nonumber \\
   &=& s(\omega_\beta) -\beta\, u(\omega_\beta) = p(\beta,\Phi), \label{eqLongCalculation}
\end{eqnarray}
and Theorem~\ref{TheEquivalence} proves that
\[
   \lim_{n\to\infty}\left\| {\rm Tr}_{\Lambda_n\setminus\Lambda_m} \bar\rho_0^{(n)} - {\rm Tr}_{\Lambda_n\setminus\Lambda_m}
   \frac{\exp(-\beta_n H_{\Lambda_n}^p)} {Z_n}\right\|_1=0
\]
and $\lim_{n\to\infty} \frac 1 {|\Lambda_n|} S(\bar\rho_0^{(n)})=s(\omega_\beta)$ as well as $\lim_{n\to\infty}\frac 1 {|\Lambda_n|}\tr(\bar\rho_0^{(n)}H_{\Lambda_n})
=\lim_{n\to\infty}\frac 1 {|\Lambda_n|}\tr(\rho_0^{(n)} H_{\Lambda_n}^p)=u(\omega_\beta)$. Together with~(\ref{eqMaxEnt}), this implies that
$\bar S(\rho_0^{(n)})=s(\omega_\beta) |\Lambda_n|+o(|\Lambda_n|)$.
It remains to estimate $d_{\rm eff}$. This will be done via Lemma~\ref{LemRenyi}. Writing $\lambda=(\lambda_1,\ldots,\lambda_N)$ and using that
$S_0(\lambda)\leq \log N \leq |\Lambda_n|\log d$ and $S(\lambda)=\bar S(\rho_0^{(n)})$, we obtain $S_2(\lambda)\geq 2 \varepsilon(S(\lambda)-\varepsilon |\Lambda_n|\log d)
=2\varepsilon\left(\strut (s(\omega_\beta)-\varepsilon\log d)|\Lambda_n| + o(|\Lambda_n|)\right)$
for all $0\leq \varepsilon\leq 1$. The special case $\varepsilon=s(\omega_\beta)/(2\log d)$ yields
$\displaystyle d_{\rm eff}=\exp(S_2(\lambda)) \geq \exp\left(\frac{s(\omega_\beta)^2}{2\log d} |\Lambda_n| + o(|\Lambda_n|)\right)$.
\end{proof}

Here is an example of a suitable sequence of initial states that appeared in work by Riera et al.~\cite{Riera}:
\begin{example}[``Flat'' pure initial state]
Consider pure initial states $\rho_0^{(n)}=|\psi_0^{(n)}\rangle\langle\psi_0^{(n)}|$ which have a flat energy distribution in an energy window, as
discussed in~\cite{Riera}. Concretely, denote the energy eigenstates of $H_{\Lambda_n}^p$ by $|E_i\rangle$, fix $\delta>0$, and set (up to normalization)
\[
   |\psi_0^{(n)}\rangle \sim \sum_{u-\delta<E_i/|\Lambda_n|<u} |E_i\rangle.
\]
If $H_{\Lambda_n}^p$ is non-degenerate, then $\bar S(\rho_0^{(n)})$ is the logarithm of the number of energy levels between densities $u-\delta$ and $u$, which is
$s(u)|\Lambda_n|+o(|\Lambda_n|)=S(\gamma^p_{\Lambda_n}(\beta))+o(|\Lambda_n|)$ according to Lemma~\ref{LemHBC}. Thus, Theorem~\ref{TheThermalizationPeriodic} proves
thermalization of small subsystems. The same conclusion holds if $|\psi_0^{(n)}\rangle$ is not exactly flat, but populates the energy levels as given
in Example~\ref{ExMain} and Figure~\ref{fig_density}.
\end{example}

This example and Theorem~\ref{TheThermalizationPeriodic} (in one formulation) assume that $H_{\Lambda_n}^p$ is non-degenerate. In fact, we show numerically
in Subsection~\ref{SecNumerical} that generic models of the kind we consider are non-degenerate, despite translation-invariance.
Alternatively, we can lift the condition of non-degeneracy or periodic boundary conditions
by proving a weaker statement about $m$-block-averaged observables.

\begin{theorem}[Thermalization, arbitrary boundary conditions]
\label{TheThermalizationArbitrary}
Let $\Phi$ be a translation-invariant finite-range interaction
which is not physically equivalent to zero, and $(\rho_0^{(n)})_{n\in\N}$ any sequence of initial states on $\Lambda_n$ which have energy expectation value
of $U_n:=\tr(\rho_0^{(n)}H_{\Lambda_n}^{BC})$ with density $U_n/|\Lambda_n|$ converging to some value $u\in(u_{\min}(\Phi),u_{\max}(\Phi))$ as $n\to\infty$,
where $BC$ denotes an arbitrary fixed choice of boundary conditions.

Suppose that the initial states have close to maximal ``population entropy'' in the following sense. Define
$\bar S(\rho_0^{(n)}):=S(\lambda_1,\ldots,\lambda_N)$, where $S$ is Shannon entropy, and $\lambda_i:=\tr(\rho_0^{(n)} \pi_i)$, where $H_{\Lambda_n}^{BC}=:\sum_{i=1}^N E_i \pi_i$
is the spectral decomposition. Then, determine the corresponding inverse temperature $\beta_n$ for which
\[
   \tr(H_{\Lambda_n}^{BC} \gamma_{\Lambda_n}^{BC}(\beta_n))=U_n,\qquad\mbox{where }\gamma_{\Lambda_n}^{BC}(\beta_n):=\frac{\exp(-\beta_n H_{\Lambda_n}^{BC})}{Z_n}.
\]
If the initial states have close to maximal population entropy in the sense that
\[
   \bar S(\rho_0^{(n)})\geq S(\gamma_{\Lambda_n}^{BC}(\beta_n))-o(|\Lambda_n|),
\]
then unitary time evolution $\rho^{(n)}(t):=\exp(-itH_{\Lambda_n}^{BC})\rho_0^{(n)}\exp(itH_{\Lambda_n}^{BC})$ thermalizes 
all $m$-block averaged observables for most times $t$:
\begin{eqnarray*}
   \left\langle \left\| \rho^{(n)}(t)-\left\langle  \rho^{(n)}(t)\right\rangle\right\|_{[m]}\right\rangle
   &\leq& d^{|\Lambda_m|} \sqrt{D_G(H_{\Lambda_n}^{BC})} \exp\left( - \frac{s(\omega_\beta)^2}{4\log d} |\Lambda_n| + o(|\Lambda_n|)\right),\qquad\mbox{ and}\\
   \lim_{n\to\infty}\left\|\left\langle \rho^{(n)}(t)\right\rangle -  \frac{\exp(-\beta_n H_{\Lambda_n}^{BC})}{Z_n}\right\|_{[m]}
   &=& 0,
\end{eqnarray*}
where $Z_n=\tr(\exp(-\beta_n H_{\Lambda_n}^{BC}))$, and $\langle\cdot\rangle$ denotes the average over all times $t\geq 0$. Furthermore, in this statement, $\beta_n$ can
be replaced by $\beta:=\beta(u)$.
\end{theorem}
\textbf{Remark.} As in the previous theorem, we can summarize the result (at the expense of losing some information) as
\[
   \lim_{n\to\infty}\left\langle \left\|\rho^{(n)}(t)-\frac{\exp(-\beta_n H_{\Lambda_n}^{BC})}{Z_n}\right\|_{[m]}\right\rangle = 0
\]
whenever the gap degeneracy $D_G$ does not grow too quickly with $|\Lambda_n|$. In fact, we can always force $D_G$ to be equal to one -- that is,
remove degeneracies -- by adding appropriate boundary conditions in the sense of Definition~\ref{DefBC}.
\begin{proof}
For any $X=X^\dagger\in\mathcal{A}_n$, we can estimate the $\|\cdot\|_{[m]}$-norm via
\begin{eqnarray*}
   \|X\|_{[m]}&=& 2 \max\left\{ \left.\left| \frac 1 {|\mathbf{T}(\Lambda_m,\Lambda_n)|}\sum_{y\in\mathbf{T}(\Lambda_m,\Lambda_n)} \tr\left[
   X(\gamma_y(A')\otimes\mathbf{1})\right] \right|\,\,\right|\,\, A'\in\mathcal{A}_m,\, 0\leq A'\leq\mathbf{1}\right\}\\
   &\leq& 2\max\left\{ \left. \frac 1 {|\mathbf{T}(\Lambda_m,\Lambda_n)|}\sum_{y\in\mathbf{T}(\Lambda_m,\Lambda_n)} \left|\tr\left[
   X(\gamma_y(A')\otimes\mathbf{1})\right]\right| \,\,\right|\,\, A'\in\mathcal{A}_m,\, 0\leq A'\leq\mathbf{1}\right\} \\
   &\leq& \frac 2 {|\mathbf{T}(\Lambda_m,\Lambda_n)|} \sum_{y\in\mathbf{T}(\Lambda_m,\Lambda_n)} \max\left\{\left.
      \left| \tr\left[ X(\gamma_y(A')\otimes\mathbf{1})\right]\right|\,\,\right|\,\, A'\in\mathcal{A}_m,\, 0\leq A'\leq\mathbf{1}
   \right\}\\
   &=& \frac 1 {|\mathbf{T}(\Lambda_m,\Lambda_n)|} \sum_{y\in\mathbf{T}(\Lambda_m,\Lambda_n)} \left\|\Tr_{\Lambda_n\setminus(\Lambda_m+y)} X\right\|_1.
\end{eqnarray*}
Using again the results of~\cite{ShortFarrelly} in the form~(\ref{eqTonyResult}), setting again $\bar\rho_0^{(n)}:=\langle \rho^{(n)}(t)\rangle=\sum_i \pi_i \rho_0^{(n)} \pi_i$, we obtain
\[
   \left\langle \left\|\rho^{(n)}(t)-\bar\rho_0^{(n)}\right\|_{[m]}\right\rangle \leq 
   \frac 1 {|\mathbf{T}(\Lambda_m,\Lambda_n)|} \sum_{y\in\mathbf{T}(\Lambda_m,\Lambda_n)} \left\langle \left\|\Tr_{\Lambda_n\setminus(\Lambda_m+y)} \rho^{(n)}(t)
   -\Tr_{\Lambda_n\setminus(\Lambda_m+y)} \bar\rho_0^{(n)}\right\|_1\right\rangle
   \leq d^{|\Lambda_m|}\sqrt{\frac{D_G(H_{\Lambda_n}^{BC})}{d_{\rm eff}}},
\]
where $d_{\rm eff}=\exp(S_2(\lambda))$. As in the proof of Theorem~\ref{TheThermalizationPeriodic}, we have $S(\bar\rho_0^{(n)})\geq\bar S(\rho_0^{(n)})=S(\lambda)$,
and also $\tr(\rho_0^{(n)}H_{\Lambda_n}^{BC})=\tr(\bar\rho_0^{(n)}H_{\Lambda_n}^{BC})$. Furthermore, Lemma~\ref{LemLocalGibbs} implies that
$\lim_{n\to\infty}\beta_n=\beta:=\beta(u)$. Thus, we can repeat the calculation~(\ref{eqLongCalculation}) in the proof of Theorem~\ref{TheThermalizationPeriodic},
and obtain that $\liminf_{n\to\infty}\frac 1 {|\Lambda_n|} \left(S(\bar\rho_0^{(n)})-\beta\,\tr(\bar\rho_0^{(n)}H_{\Lambda_n})\right)\geq p(\beta,\Phi)$.
Consequently, Theorem~\ref{TheEquivalence2} and Lemma~\ref{LemNormEquivalence} imply that
\[
   \lim_{n\to\infty}\left\|\bar\rho_0^{(n)} - \frac{\exp(-\beta_n H_{\Lambda_n}^{BC})}{Z_n}\right\|_{[m]} = 0 \qquad\mbox{and}\qquad
   \lim_{n\to\infty}\frac 1 {|\Lambda_n|} \tr(\bar\rho_0^{(n)} H_{\Lambda_n}^{BC})=u(\omega_\beta).
\]
As in the proof of Theorem~\ref{TheThermalizationPeriodic}, it also follows that $\bar S(\rho_0^{(n)})\geq s(\omega_\beta) |\Lambda_n| + o(|\Lambda_n|)=S(\lambda)$.
Repeating the final steps of the proof of Theorem~\ref{TheThermalizationPeriodic}
yields the claimed estimate for $d_{\rm eff}$.
\end{proof}

\subsection{Finite-size estimates for systems without interaction}
\label{SecIsing}
As the most simple special case, consider the non-interacting Hamiltonian
\[
   H_\Lambda:=\sum_{x\in \Lambda} h_x,
\]
where $h_x=\gamma_x(h)$ denotes a fixed self-adjoint matrix $h$ sitting on site $x\in\Lambda$. This corresponds to an interaction $\Phi$ of the form
\[
   \Phi(X)=\left\{
      \begin{array}{cl}
         \gamma_X(h) & \mbox{if }\# X=1 \\
         0 & \mbox{otherwise}.
      \end{array}
   \right.
\]
Since there is no interaction, the dimension $\nu$ of the lattice $\mathbb{Z}^\nu$ does not play any role; without loss of generality, we may assume that $\nu=1$.
Similarly, we set $\Lambda_m=[1,m]\subset \Z$. Without interaction, the (restriction of the global) Gibbs state becomes the product state
\[
   {\rm Tr}_{\Lambda_n\setminus\Lambda_m}\frac{\exp(-\beta H_{\Lambda_n})} {Z_n} = \gamma_\beta^{\otimes m},
\]
where $\gamma_\beta = \exp(-\beta h)/Z_1$ is the single-site Gibbs state with single-site partition function $Z_1$. We will now look at equivalence of ensembles -- and its
finite-size behavior -- in this special case.
That is, we consider the maximally mixed state $\tau_n$ on
\[
   T_n:={\rm span}\left\{ |E\rangle\,\, \left| \,\ \frac E n \in [u-\delta,u]\right. \right\},
\]
where $\delta>0$ and $u$ will be considered fixed in what follows. On every site, we can choose the local basis such that $h$ is diagonal, denoting the corresponding
single-site eigenstates of $h$ by $\{|0\rangle,\ldots, |d-1\rangle\}$. (Recall that $d$ denotes the single-site Hilbert space dimension.)
For $0\leq j \leq d-1$, the eigenvalue corresponding to $|j\rangle$ will be denoted $E_j$; that is,
\[
   h |j\rangle = E_j |j\rangle.
\]
We may always choose a basis and shift the energy such that $0=E_0\leq E_1\leq\ldots\leq E_{d-1}$, i.e.
\[
   h=\left(
      \begin{array}{cccc}
         0 & & & \\
         & E_1 & & \\
         & & \ddots & \\
         & & & E_{d-1}
      \end{array}
   \right).
\]
Every string $s=s_1 s_2\ldots s_n$ of length $n$ over the alphabet $\{0,\ldots,d-1\}$ describes an eigenvector $|s\rangle:=|s_1\rangle\otimes \ldots \otimes |s_n\rangle$
of $H$ on $n$ sites, where $H|s\rangle=\sum_i h|s_i\rangle = \sum_i E_{s_i}$. Thus, the microcanonical subspace can also be written
\[
   T_n={\rm span}\left\{ |s\rangle\,\,\left|\,\, s\in \{0,1,\ldots,d-1\}^n,\enspace \frac 1 n \sum_{i=1}^n E_{s_i} \in [u-\delta,u]\right. \right\}.
\]
Our goal is to estimate the difference
\begin{equation}
   \left\| {\rm Tr}_{\Lambda_n\setminus\Lambda_m} \tau_n - \gamma_\beta^{\otimes m}\right\|_1.
   \label{eqDifference}
\end{equation}
Since all relevant operators compute, we can restrict to the probability distributions on the diagonal; we have a purely classical problem.
Our first observation is that a tight upper bound on this expression is known in the special case $\delta=0$ and $d=2$; it has been obtained in proofs
of the finite classical de Finetti Theorem~\cite{Diaconis}.

\begin{lemma}
\label{LemNonIntDiaconis}
In the case of a perfectly sharp microcanonical subspace, i.e.\ $\delta=0$, and of qubit systems, i.e.\ $d=2$, we have
\[
   \left\| {\rm Tr}_{\Lambda_n\setminus\Lambda_m} \tau_n - \gamma_\beta^{\otimes m}\right\|_1\leq \frac{4m} n,
\]
assuming that the energy density $u$ is chosen such that the corresponding microcanonical subspace $T_n$ is not empty.
\end{lemma}
\proof
We have $\displaystyle h=\left(\begin{array}{cc} 0 & \\ & E_1 \end{array}\right)$, and so $\displaystyle \rho_\beta =
\left(\begin{array}{cc} 1 & \\ & e^{-\beta E_1} \end{array}\right)\cdot \frac 1 {1+e^{-\beta E_1}}$.
The inverse temperature $\beta$ is determined by $\tr(\gamma_\beta h)=u$. In this case, $u=E_1\cdot p_1$, where $p_1$ is the
relative frequencies of $1$'s in the strings $s$ with $|s\rangle\in T_n$. This equation implies
$\displaystyle \gamma_\beta =\left(\begin{array}{cc} 1-p_1 & \\ & p_1\end{array}\right)$, with classical probability distribution $P_\beta:=(1-p_1,p_1)$
on the diagonal. If we denote by $Q$ the classical probability distribution on $\{0,1\}^m$ determined by the diagonal elements of ${\rm Tr}_{\Lambda_n\setminus\Lambda_m}\tau_n$,
we have
\[
   \left\| {\rm Tr}_{\Lambda_n\setminus\Lambda_m} \tau_n - \gamma_\beta^{\otimes m}\right\|_1=\left\| Q-P_\beta^{\otimes m}\right\|_1,
\]
where $\|\cdot\|_1$ on the right-hand side denotes the variation distance of two probability distributions:
\[
   \|P-Q\|_1=\sum_{i=1}^{2^m} |P_i-Q_i|=2\max_{A\subseteq \{1,\ldots,2^m\}} |P(A)-Q(A)|.
\]
Consider an urn $U$ with $n$ balls, where $p_1\cdot n$ of them are marked by a ``$1$'' and all others marked by a ``$0$''. Then $P_\beta^{\otimes m}$
describes the distribution obtained by $m$ draws from $U$ with replacement, whereas $Q$ described the distribution obtained by $m$ draws from $U$
\emph{without} replacement, where in both cases the order of the results is taken into account. These distributions are considered in~\cite{Diaconis} in
the proof a finite version of the classical de Finetti theorem. The main result then follows from Theorem~(4) in~\cite{Diaconis}.
\qed

For $d\geq 3$, even if $\delta=0$, the results of~\cite{Diaconis} do not directly yield an upper bound on expression~(\ref{eqDifference}). This is for two reasons.
First, the typical subspace $T_n$ will in general not be spanned by a single type class, but by several ones. For example, consider the case $d=3$ with
energies $E_0=0$, $E_1=1$, and $E_2=2$. Fixing the energy density to $u=2/3$ yields the microcanonical subspace
\[
   T_3= {\rm span}\left\{ |011\rangle, |101\rangle,|110\rangle,|002\rangle,|020\rangle,|200\rangle\right\}.
\]
This is a disjoint union of two type classes. While $T_4=\{0\}$ and $T_5=\{0\}$, we have
\[
   T_6={\rm span}\left\{ |000022\rangle,\ldots, |000112\rangle,\ldots,|001111\rangle,\ldots\right\},
\]
where the dots denotes all permutations. This is a union of three type classes. Then the results in~\cite{Diaconis} do not prove directly that
${\rm Tr}_{\Lambda_n\setminus\Lambda_m}\tau_n$ is close to a product state, but that it is close to a \emph{convex combination} of product states, resembling the de Finetti
theorem.

In this particular example, it can be checked numerically that the qualitative behavior of Lemma~\ref{LemNonIntDiaconis} remains true: $n$ needs
to be increased linearly with $m$ in order to achieve a fixed one-norm distance error. The inverse temperature turns out to be $\beta=\log[(1+\sqrt{33})/4]$,
and $\gamma_\beta={\rm diag}\left( (15-\sqrt{33})/18, (\sqrt{33}-3)/9,(9-\sqrt{33})/18\right)$. The subspace $T_n$ is non-trivial whenever $n$ is a multiple of $3$.
Define the function $f:\N\to\N$ by
\[
   f(m):=\mbox{ smallest possible }n\in\N\mbox{ such that }\left\| {\rm Tr}_{\Lambda_n\setminus\Lambda_m} \tau_n - \gamma_\beta^{\otimes m}\right\|_1\leq \frac 1 {100}.
\]
This function is evaluated numerically in Figure~\ref{fig_min_n}. It can be seen that $n=f(m)$ increases linearly with $m$.
\begin{figure}[!hbt]
\begin{center}
\includegraphics[angle=0, width=7cm]{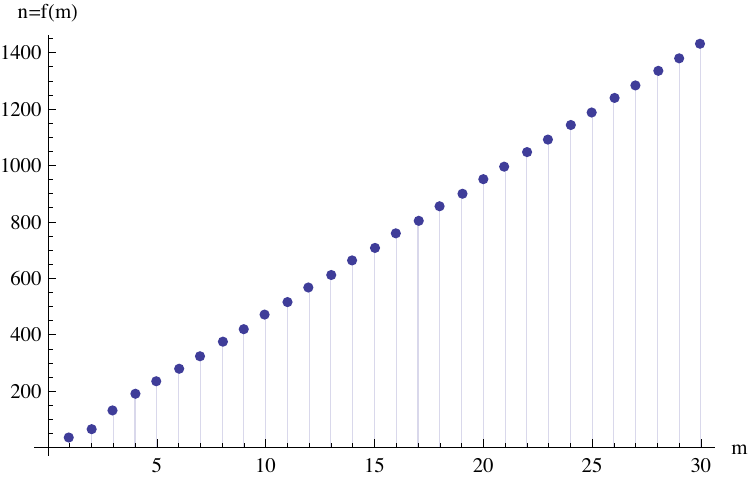}
\caption{Minimal number of sites $n$ to guarantee that a subsystem of given size $m$ is $\varepsilon$-close to the local Gibbs state, where $\varepsilon=1/100$,
energy density $u=2/3$, local Hilbert space dimension $d=3$, and energy levels $E_0=0$, $E_1=1$ and $E_2=2$. In this case, the microcanonical subspace
of width $\delta=0$ is spanned by more than one type class. It can be seen that the size of the ``bath'' has to be increased linearly with the size of the subsystem.
}
\label{fig_min_n}
\end{center}
\end{figure}

It turns out that for $\delta=0$, the previous example is atypical in the sense that \emph{generic} energy windows usually lead to microcanonical
subspaces $T_n$ that contain only a single type class. This can be characterized as in the following lemma. We use the standard terminology to
call a set of real numbers $E_1,\ldots,E_{d-1}$ \emph{rationally dependent} if there are rational numbers $\lambda_1,\ldots,\lambda_{d-1}\in\mathbb{Q}$,
not all of them zero, such that $\sum_{i=1}^{d-1}\lambda_i E_i=0$, and otherwise \emph{rationally independent}.
\begin{lemma}
\label{LemRationalIndependent}
Suppose that $\delta=0$. Then, all non-trivial microcanonical subspaces $T_n\neq\{0\}$, for all $n$ and $u$, are spanned by a single type class if and only if
the energies $E_1,\ldots, E_{d-1}$ are rationally independent.
\end{lemma}
\proof
We denote type classes as follows:
\[
   T(k_0,\ldots,k_{d-1}):=\left\{s\in\{0,\ldots,d-1\}^n\,\,|\,\, \#\{i:s_i=j\}=k_j\mbox{ for all }j\right\},
\]
that is, the set of all strings that have $k_0$ zeroes, $k_1$ ones, and so on. All strings $s$ in the same type class have the same energy
$\langle s|H|s\rangle=:E_s=\sum_i E_{s_i}=\sum_{j=0}^{d-1} k_j E_j$. Thus, the microcanonical subspace $T_n$ must be a disjoint union of (spans of) type classes.

Suppose the energies are rationally independent,
and suppose that ${\rm span}\, T(k_0,\ldots,k_{d-1})\subset T_n$ and at the same time ${\rm span}\,T(k'_0,\ldots,k'_{d-1})\subset T_n$. Then
\[
   u\cdot n = k_1 E_1+\ldots + k_{d-1} E_{d-1} = k'_1 E_1+\ldots + k'_{d-1} E_{d-1}.
\]
Thus
\[
   \underbrace{(k_1-k'_1)}_{\in\mathbb{Z}}E_1+\ldots + \underbrace{(k_{d-1}-k'_{d-1})}_{\in \mathbb{Z}}E_{d-1}=0,
\]
and rational independence implies that $k_j=k'_j$ for all $j$, so $T_n$ is the span of a single type class.

Conversely, suppose that $E_1,\ldots,E_{d-1}$ are rationally dependent. That is, there are $k_1,\ldots,k_{d-1}\in\Z$, not all $k_i=0$, such that
\[
   k_1E_1+\ldots+k_{d-1}E_{d-1}=0.
\]
There is at least one energy $E_i$ with $E_i>0$, so this equation can only be satisfied if $\min_i k_i=:k_j<0$ and $\max_i k_i>0$. Set
\[
   E:=-k_j(E_1+\ldots+E_{d-1})=(k_1-k_j)E_1+\ldots+(k_{d-1}-k_j)E_{d-1}
\]
Then all $k'_i:=k_i-k_j\geq 0$ are integers, and they cannot all be zero. Choose any $n\in\N$ with
\[
   n\geq\max\{(d-1)|k_j|,(k_1-k_j)+(k_2-k_j)+\ldots+(k_{d-1}-k_j)\}.
\]
Set $k_0:=n-(d-1)|k_j|\geq 0$ and $k'_0:=n-[(k_1-k_j)+\ldots+(k_{d-1}-k_j)]\geq 0$, and set the energy density to $u:=E/n$.
Then we have
\begin{eqnarray*}
   {\rm span}\, T(k_0,|k_j|,|k_j|,\ldots,|k_j|)&\subseteq& T_n,\\
   {\rm span}\, T(k'_0,k_1-k_j,k_2-k_j,\ldots,k_{d-1}-k_j)&\subseteq& T_n.
\end{eqnarray*}
Thus, $T_n$ is spanned by at least two different type classes.
\qed

There is a second reason why the results in~\cite{Diaconis} cannot directly be used if $d\geq 3$, even in the case where $\delta=0$ \emph{and}
assuming the rational independence of the energies. It follows from~\cite{Diaconis} that in this case
\begin{equation}
   \left\|{\rm Tr}_{\Lambda_n\setminus\Lambda_m} \tau_n -\gamma^{\otimes m}\right\|_1\leq\frac{2dm} n;
   \label{eqDiaconis}
\end{equation}
however, the state $\gamma$ is in general \emph{not} equal to $\gamma_\beta$ for any $\beta$. Instead, $\gamma$ is the single-site density matrix with
the symbols' relative frequencies in the type class as eigenvalues, and this is in general not a thermal state.

As a simple example, consider the case $d=3$, with single-site Hamiltonian $h={\rm diag}\left(0,1,\sqrt{2}\right)$ and energy density
$u=(2+\sqrt{2})/6$. If $n$ is a multiple of $6$, then $T_n$ contains all basis vectors $|s\rangle$ with strings $s\in\{0,1\}^n$ that have
$n/2$ zeroes, $n/3$ ones and $n/6$ twos. Then the $\gamma$ appearing in~(\ref{eqDiaconis}) is $\gamma={\rm diag}(1/2,1/3,1/6)$, and there
does not exist any $\beta$ such that $\gamma=\gamma_\beta$.

In the following, we will generalize the result of Lemma~\ref{LemNonIntDiaconis}, by showing that also in the case of a microcanonical subspace of width $\delta>0$,
the qualitative behavior of Figure~\ref{fig_min_n} remains true, at least in the case $d=2$, i.e. in the qubit case. First, we prove a lemma which shows this for
$\delta=0$ or $\delta$ depending on $n$ and approaching zero fast enough. Later, we will extend the result to arbitrary fixed $\delta>0$ by some large deviations argument.

\begin{lemma}
\label{LemResultDeltaZero}
Consider the case of qubits, i.e.\ $d=2$, and shift the energies such that $E_0=0$ and $E_1=1$. If $\tau_n$ is the maximal mixture on the non-trivial microcanonical
subspace corresponding to the energy interval $n\cdot [u-\delta,u]$,
with $0\leq\delta<u\leq \frac 1 2$, and $\gamma_\beta$ is the single-site Gibbs state with corresponding inverse temperature $\beta$, then we have for subsystems
of size $m\leq n(u-\delta)$,
\[
   S\left(\strut \gamma_\beta^{\otimes m}\, \left\|\,{\rm Tr}_{\Lambda_n\setminus\Lambda_m}\tau_n \right.\right)\leq \frac {(1-\delta)u}{u-\delta}\cdot \frac m {n-m}+\frac{m u \delta}{u-\delta}\left(
   1+\frac m {n-m}\right),
\]
where $S$ denotes the quantum relative entropy (with logarithm in base $e$). In particular, if $\delta=0$, the relative entropy is upper-bounded by $m/(n-m)$, and the Pinsker inequality
yields
\[
   \left\| {\rm Tr}_{\Lambda_n\setminus\Lambda_m}\tau_n - \gamma_\beta^{\otimes m}\right\|_1 \leq \sqrt{\frac 1 2 \cdot \frac m {n-m}}  \qquad\qquad (\mbox{special case }\delta=0).
\]
\end{lemma}
\proof
As explained above, the calculation is classical: we can regard $\gamma_\beta^{\otimes m}$ as a classical probability distribution on the binary strings
of length $m$, given by
\[
   P_\beta^{\otimes m}(x)=u^k(1-u)^{m-k},\qquad\mbox{where }k\mbox{ is the number of ones in }x.
\]
From elementary combinatorics, the marginal distribution $Q$ is given by
\[
   Q(x)=\left(\dim T_n\right)^{-1}\cdot \sum_{\ell\in[\lceil n(u-\delta)\rceil, \lfloor n u\rfloor ],\enspace \ell\geq k} {{n-m} \choose {\ell-k}},
\]
and the numerator counts all possible ways to complete $x$ to a string of length $n$ which has $\ell$ ones such that the energy is in the suitable interval.
Since $k\leq m\leq n(u-\delta)\leq\lceil n(u-\delta)\rceil$, the condition $\ell\geq k$ is automatically satisfied for all $\ell$ in the summation interval; hence this
condition can be removed from the specification of the sum. The dimension of the microcanonical subspace is given by
\begin{equation}
   \dim T_n = \sum_{\ell=\lceil n(u-\delta)\rceil}^{\lfloor n u \rfloor} {n\choose\ell}.
   \label{eqDim}
\end{equation}
Thus, the quantum relative entropy $S$ can be written in terms of the classical relative entropy $H$,
\begin{eqnarray}
   S\left(\strut \gamma_\beta^{\otimes m}\, \left\|\,{\rm Tr}_{\Lambda_n\setminus\Lambda_m}\tau_n \right.\right) &=& H\left(\left.P_\beta^{\otimes m}\,\right\|\, Q\right)
   =\sum_{x\in\{0,1\}^m} P_\beta^{\otimes m}(x)\left( \log P_\beta^{\otimes m}(x)-\log Q(x)\right)\nonumber\\
   &=&\sum_{k=0}^m {m\choose k} u^k (1-u)^{m-k}\left(\log P_k-\log Q_k\right), \label{eqRelEntropy}
\end{eqnarray}
where
\begin{eqnarray}
   P_k&=&u^k(1-u)^{m-k},\label{eqPk}\\
   Q_k&=& \left(\dim T_n\right)^{-1}\cdot \sum_{\ell=\lceil n(u-\delta)\rceil}^{\lfloor n u \rfloor} {{n-m} \choose {\ell-k}}.\nonumber
\end{eqnarray}
Using that $(n-m)!=n!/[(n-m+1)(n-m+2)\ldots n]$ and similar identities for $(\ell -k)!$ and $[n-\ell-(m-k)]!$, we obtain
\[
   Q_k= \left(\dim T_n\right)^{-1}\cdot \sum_{\ell=\lceil n(u-\delta)\rceil}^{\lfloor n u \rfloor}\frac{n!\prod_{j=0}^{k-1}(\ell-j)
   \prod_{j=0}^{m-k-1}(n-\ell-j)}{\prod_{j=0}^{m-1}(n-j)\, \ell! (n-\ell)!}
\]
In order to eliminate all $\ell$-variables from all products, we substitute the inequalities
\begin{eqnarray*}
   \ell-j &\geq& \lceil n(u-\delta)\rceil - j,\\
   n-\ell-j &\geq& n - \lfloor n u \rfloor - j
\end{eqnarray*}
and obtain
\[
   Q_k\geq \left(\dim T_n\right)^{-1} \frac{\prod_{j=0}^{k-1}\left(\strut \lceil n(u-\delta)\rceil-j\right)\prod_{j=0}^{m-k-1}\left(\strut n-\lfloor nu\rfloor - j\right)}
   {\prod_{j=0}^{m-1}(n-j)}
   \sum_{\ell=\lceil n(u-\delta)\rceil}^{\lfloor n u \rfloor} {n\choose \ell}.
\]
Thus, the sum on the right-hand side exactly cancels the factor $\left(\dim T_n\right)^{-1}$ according to~(\ref{eqDim}), and we obtain
\begin{eqnarray*}
   \log Q_k&\geq& \sum_{j=0}^{k-1}\log\left(\strut \lceil n(u-\delta)\rceil-j\right)+\sum_{j=0}^{m-k-1}\log\left(\strut n - \lfloor nu\rfloor - j\right)
   -\sum_{j=0}^{m-1}\log(n-j) \\
   &=& \sum_{j=0}^{k-1} \log\frac{\lceil n(u-\delta)\rceil - j}{n-j-m+k} + \sum_{j=0}^{m-k-1}\log\frac{n-\lfloor n u \rfloor - j}{n-j}.
\end{eqnarray*}
It is easy to check that the addends in both sums are (negative and) decreasing functions in $j$; thus, we can lower-bound the sums by integrals:
\begin{eqnarray*}
   \log Q_k&\geq& \int_0^k \log\frac{\lceil n(u-\delta)\rceil - j}{n-j-m+k}\, dj + \int_0^{m-k}\log\frac{n-\lfloor n u \rfloor - j}{n-j}\, dj \\
   &=& \lceil n(u-\delta)\rceil \log\lceil n(u-\delta)\rceil - \lceil n(u-\delta)\rceil \log\left(\lceil n(u-\delta)\rceil - k\right) + k\log\left(\lceil n(u-\delta)\rceil - k\right)\\
   &&+(n-m)\log(n-m)+\left(n-\lfloor nu \rfloor\right)\log\left( n - \lfloor nu \rfloor\right)-\left(n-\lfloor n u \rfloor\right)\log\left(n-\lfloor n u \rfloor - m + k\right)\\
   &&+(m-k)\log\left(n-\lfloor n u \rfloor - m + k\right)-n\log n.
\end{eqnarray*}
The right-hand side contains the expressions $f\left(\strut \lceil n(u-\delta)\rceil\right)$ and $g\left(\strut n - \lfloor n u \rfloor\right)$, where
$f(x):=x\log x-x\log(x-k)+k\log(x-k)$ and $g(x):=x\log x - x \log(x-m+k)+(m-k)\log(x-m+k)$. It is easy to check that $f$ and $g$ are both increasing in the relevant
intervals, thus we have $f\left(\strut \lceil n(u-\delta)\rceil\right)\geq f\left(\strut n(u-\delta)\right)$ and $g\left(\strut n - \lfloor n u \rfloor\right)
\geq g(n-nu)$, and all the floors and ceilings in the inequality above can be dropped.

Due to~(\ref{eqPk}), we have $\log P_k=k\log u +(m-k)\log(1-u)$, thus
\begin{eqnarray*}
   \log P_k - \log Q_k&\leq& k\log u + (m-k)\log(1-u)-n(u-\delta)\log[n(u-\delta)]+n(u-\delta)\log[n(u-\delta)-k]\\
   && -k\log[n(u-\delta)-k]-(n-m)\log(n-m)-n(1-u)\log[n(1-u)]\\
   && + n(1-u)\log[n(1-u)-m+k]-(m-k)\log[n(1-u)-m+k]+n\log n.
\end{eqnarray*}
The largest contribution to the sum in~(\ref{eqRelEntropy}) will be those $k$ where $k\approx mu$. This motivates the definition
$\varepsilon_k:=k-mu$ (despite the name, this can be a negative number). Replacing all $k$ by $mu+\varepsilon_k$ yields
\begin{eqnarray*}
   \log P_k - \log Q_k&\leq& [n-m-nu+mu+\varepsilon_k]\log\left( 1+\frac{\varepsilon_k}{(1-u)(n-m)}\right)-(mu+\varepsilon_k)\log\left(1-\frac\delta u\right)\\
   &&+(nu-mu-\varepsilon_k-n\delta)\log\left( 1 - \frac{\varepsilon_k+m\delta}{(u-\delta)(n-m)}\right) - n \delta\log \left(1-\frac m n\right).
\end{eqnarray*}
All real numbers $x>-1$ satisfy $x/(1+x)\leq \log(1+x)\leq x$. Thus
\begin{eqnarray}
   \log P_k - \log Q_k&\leq& [(n-m)(1-u)+\varepsilon_k]\cdot \frac{\varepsilon_k}{(1-u)(n-m)} + (mu+\varepsilon_k)\frac{\delta/u}{1-\delta/u}\nonumber\\
   && + [n(u-\delta)-m u-\varepsilon_k]\left( - \frac{\varepsilon_k+m\delta}{(u-\delta)(n-m)}\right)+n\delta \frac{m/n}{1-m/n}.\label{eqSubstThisS}
\end{eqnarray}
We have the following three equations for the Binomial distribution:
\begin{eqnarray}
   \sum_{k=0}^m {m \choose k}u^k (1-u)^{m-k} &=& 1,\label{eqBinomial1}\\
   \sum_{k=0}^m {m \choose k}u^k (1-u)^{m-k}\varepsilon_k &=& 0, \label{eqBinomial2}\\
   \sum_{k=0}^m {m \choose k}u^k (1-u)^{m-k}\varepsilon_k^2 &=& mu(1-u),\label{eqBinomial3}
\end{eqnarray}
where~(\ref{eqBinomial1}) is simply the normalization of the Binomial distribution, (\ref{eqBinomial3}) is its variance, and~(\ref{eqBinomial2})
follows from its expectation value. Thus, when substituting~(\ref{eqSubstThisS}) into the expression~(\ref{eqRelEntropy}) for the relative entropy,
we can drop all terms linear in $\varepsilon_k$. We obtain
\begin{eqnarray*}
    S\left(\strut \gamma_\beta^{\otimes m}\, \left\|\,{\rm Tr}_{\Lambda_n\setminus\Lambda_m}\tau_n \right.\right) &\leq& \sum_{k=0}^m {m\choose k} u^k (1-u)^{m-k}\left[
       \frac{\varepsilon_k^2}{(1-u)(n-m)} + \frac{m\delta}{1-\delta/u}+n\delta \frac m {n-m}\right. \\
       \ &&\left. \qquad\qquad\qquad\qquad\qquad\qquad +\frac{\varepsilon_k^2}{(u-\delta)(n-m)}-\frac{m\delta}{(u-\delta)(n-m)}
       \left(n(u-\delta)-m u\right)
    \right]\\
    &=& \frac {(1-\delta)u}{u-\delta}\cdot \frac m {n-m}+\frac{m u \delta}{u-\delta}\left(1+\frac m {n-m}\right).
\end{eqnarray*}
This proves the claim.
\qed

\begin{theorem}
\label{TheMainFiniteSize}
Consider the case of qubits, i.e.\ $d=2$, and shift the energies such that $E_0=0$ and $E_1=1$. Suppose that $\tau_n$ is the maximal mixture on the non-trivial microcanonical
subspace corresponding to the energy interval $n\cdot [u-\delta,u]$,
with $0\leq\delta<u< \frac 1 2$, and $\gamma_\beta$ is the single-site Gibbs state with corresponding inverse temperature $\beta$.
If the size of the subsystem $m$ is large enough such that $\frac{20}m\log\frac m u \leq\log\frac{1-u}u$, and at the same time
$5\leq m\leq n(u-\delta)$, then we have
\[
   \left\| {\rm Tr}_{\Lambda_n\setminus\Lambda_m}\tau_n - \gamma_\beta^{\otimes m}\right\|_1 \leq
   \frac{2\delta}{n\sqrt{u}}+\sqrt{\frac m {n-m}\left(1+\frac{4\log n}{\log\frac{1-u}u}\right)}.
\]
\end{theorem}
\proof
We start by introducing some notation. For arbitrary subsets $S\subseteq [0,u]$ define $\tau_n^S$ to be the maximally mixed state on the subspace
\[
   T_n^S:={\rm span}\left\{|s\rangle\,\,\left|\,\ s\in\{0,1\}^n,\enspace \frac 1 n \sum_{i=1}^n E_{s_i}\in S\right.\right\}.
\]
As before, we set $T_n:=T_n^{[u-\delta,u]}$ and $\tau_n:=\tau_n^{[u-\delta,u]}$. Moreover, define
\[
   \mu_n^{S}:=\frac{\dim T_n^S}{\dim T_n};
\]
then, if we write $[u-\delta,u]$ as any disjoint union of two sets $S$ and $T$, the microcanonical state can be written as a convex combination,
$\tau_n=\mu_n^S\tau_n^S+\mu_n^T \tau_n^T$.
In the following, $(\alpha_n)_{n\in\N}$ will be any sequence of positive real numbers tending to zero, satisfying $1/n<\alpha_n<u-1/n$, to be specified later. We start with the identity
\[
   \tau_n=\mu_n^{\left[ u-\delta,u-\alpha_n\right)} \tau_n^{\left[ u-\delta,u-\alpha_n\right)}
   +\mu_n^{\left[ u-\alpha_n,u\right]} \tau_n^{\left[ u-\alpha_n ,u\right]}.
\]
Due to convexity, the Pinsker inequality, and Lemma~\ref{LemResultDeltaZero}, we have
\begin{eqnarray}
   \left\| {\rm Tr}_{\Lambda_n\setminus\Lambda_m}\tau_n - \gamma_\beta^{\otimes m}\right\|_1 &\leq&  \mu_n^{[u-\delta,u-\alpha_n)} \left\| {\rm Tr}_{\Lambda_n\setminus\Lambda_m}\tau_n^{[u-\delta,u-\alpha_n)} - \gamma_\beta^{\otimes m}\right\|_1
   +\mu_n^{[u-\alpha_n,u]} \left\| {\rm Tr}_{\Lambda_n\setminus\Lambda_m}\tau_n^{[u-\alpha_n,u]} - \gamma_\beta^{\otimes m}\right\|_1 \nonumber\\
   &\leq& 2 \mu_n^{[u-\delta,u-\alpha_n)}+ \sqrt{ \frac 1 2 S\left(\strut \gamma_\beta^{\otimes m}\, \left\|\,{\rm Tr}_{\Lambda_n\setminus\Lambda_m}\tau_n^{[u-\alpha_n,u]} \right.\right)} \nonumber\\
   &\leq& 2 \mu_n^{[u-\delta,u-\alpha_n)}+ \sqrt{\frac 1 2} \cdot\sqrt{\frac{(1-\alpha_n)u}{u-\alpha_n}\cdot\frac m {n-m} + \frac{m u \alpha_n}{u-\alpha_n}
   \left(1+\frac m {n-m}\right)}.\label{eqMainIneq}
\end{eqnarray}
Let $\tilde u$ be the largest $p\in[u-\delta,u]$ with the property that $T_n^{\{p\}}\neq\{0\}$; it is given by the equation $\lfloor u\cdot n\rfloor = \tilde u\cdot n$.
Then we can upper-bound the measure $\mu_n^{[u-\delta,u-\alpha_n)}$ in the following way:
\begin{eqnarray}
   \mu_n^{[u-\delta,u-\alpha_n)}&=& \frac{\sum_{p\in[u-\delta,u-\alpha_n)} \dim T_n^{\{p\}}}{\sum_{p\in[u-\delta,u]} \dim T_n^{\{p\}}}
   \leq \frac{\sum_{j\in\N: j/n \in [u-\delta,u-\alpha_n)} {n\choose j}}{\dim T_n^{\{\tilde u\}}}
   \leq \frac{\# \{j\in\N:j/n \in [u-\delta,u-\alpha_n)\} \cdot {n\choose{\lfloor n(u-\alpha_n)\rfloor}}}{{n\choose{n\tilde u}}}\nonumber\\
   &\leq& \frac{n\delta{n\choose{\lfloor n(u-\alpha_n)\rfloor}}}{{n\choose{\lfloor nu\rfloor}}}.\label{eqUpperBoundMu}
\end{eqnarray}
The Binomial coefficients can be estimated by using Lemma 17.5.1 in~\cite{CoverThomas}:
For $0<p<1$ such that $np$ is an integer, we have
\[
   \frac 1 {\sqrt{8np(1-p)}} \leq {n\choose{np}} e^{-nH(p)} \leq \frac 1 {\sqrt{\pi n p (1-p)}},
\]
where $H(p)=-p\log p -(1-p)\log (1-p)$ is the binary entropy function. Substituting this into~(\ref{eqUpperBoundMu}), defining $p$ by $np=\lfloor n(u-\alpha_n)\rfloor$,
and using that $p\leq u-\alpha_n$ as
well as $u\geq \tilde u\geq u-1/n$, we obtain
\begin{eqnarray}
   \mu_n^{[u-\delta,u-\alpha_n)}&\leq&n\delta{n\choose{np}} {n\choose {n\tilde u}}^{-1}\leq n\delta\sqrt{\frac{\tilde u(1-\tilde u)}{p(1-p)}} e^{n[H(u-\alpha_n)-H(u-1/n)]}
   \label{eqMuAgain}
\end{eqnarray}
(note that $u>1/n$ due to $n\geq m/(u-\delta)>m/u\geq 1/u$). Since the binary entropy function $H$ is concave in the interval $[0,1/2]$, we have
\[
   H(u-\alpha_n)\leq H\left(u-\frac 1 n\right)-H'\left(u-\frac 1 n\right)\cdot \left(\alpha_n-\frac 1 n\right)\quad\Rightarrow
   \quad H(u-\alpha_n)-H\left(u-\frac 1 n\right)\leq -\left(\alpha_n-\frac 1 n\right)\log\frac{1-(u-1/n)}{u-1/n}.
\]
Substituting this and $\tilde u(1-\tilde u)\leq 1/4$ as well as $1/\sqrt{p(1-p)}\leq\sqrt{2/p}$ and $p\geq u-\alpha_n-1/n$ into~(\ref{eqMuAgain}), we get
\[
   \mu_n^{[u-\delta,u-\alpha_n)} \leq \frac 1 2 n\delta \sqrt{\frac 2 {u-\alpha_n-1/n}}\,
   (c_n)^{-n\left(\alpha_n-\frac 1 n\right)},\qquad\mbox{where }
   c_n=\left(
      \frac{1-\left(u-\frac 1 n\right)}{u-\frac 1 n}
   \right).
\]
Now set
\begin{equation}
   \alpha_n:=\frac 1 n +\frac{2\log n}{n \log c_n}=\mathcal{O}\left(\frac{\log n} n\right).
   \label{DefAlphan}
\end{equation}
Since $n>m/u$, this is less than $u-1/n$ as necessary if $m$ is large enough; it turns out that $m\geq 5$ and $(20/m) \log(m/u)\leq\log((1-u)/u)$
gives in fact $\alpha_n<u/2-1/n$. This yields $(c_n)^{-n\left(\alpha_n-\frac 1 n\right)}=n^{-2}$, and so
\[
   \mu_n^{[u-\delta,u-\alpha_n)} \leq \frac\delta{n\sqrt{u}}.
\]
Substituting this and $u-\alpha_n>u/2$ as well as $m/(n-m)\leq 1$ into~(\ref{eqMainIneq}) yields
\begin{eqnarray*}
\left\| {\rm Tr}_{\Lambda_n\setminus\Lambda_m}\tau_n - \gamma_\beta^{\otimes m}\right\|_1 &\leq& \frac{2\delta}{n\sqrt{u}}+\frac 1 {\sqrt{2}}\sqrt{2\frac m {n-m}+4m\alpha_n}.
\end{eqnarray*}
Then the claim follows by substituting~(\ref{DefAlphan}) and  $\log c_n=\log(1-u+1/n)-\log(u-1/n)>\log(1-u)-\log u$.
\qed

This mainly recovers the result depicted in Figure~\ref{fig_min_n}, where the size of the ``bath'', $n-m$, has to be increased linearly
with the size of the subsystem, $m$, to achieve a fixed error. In this theorem, for $\delta>0$, the $(\log n)$-term contributes a small correction to this
behavior, and $n$ has to be increased slightly super-linearly with $m$.

\subsection{Numerical results on finite-size behavior in one dimension}
\label{SecNumerical}
Here we provide numerical examples that not only show that random local Hamiltonians satisfy our requirements for canonical typicality and dynamical thermalization, but also that the replacing the global Gibbs state with the local Gibbs state does not give the correct statistics.  This emphasizes that entanglement is key to understanding why closed quantum systems can conform to thermodynamic predictions.
The class of Hamiltonians that we consider are random $2$--local Hamiltonians acting on $n$ qubits on a line with periodic boundary conditions:
\begin{equation}
H_{\Lambda_n}^p= \sum_{i=1}^n \left(\strut H_0^{(i)} + H_{\rm{int}}^{(i,i+1{~\rm mod~} n)}\right),
\end{equation}
where the onsite term is of the form for constants $a_1,a_2$ and $a_3$,
\begin{align}
H_0^{(i)} = a_{1} \sigma_x^{(i)} + a_{2}\sigma_y^{(i)}+a_{3}\sigma_z^{(i)},
\end{align}
and the interaction term takes the form, for constants $b_{1,1}, b_{1,2},\ldots,b_{3,3}$,
\begin{align}
H_{\rm int}^{(i,j)}= b_{1,1}\sigma_x^{(i)}\sigma_x^{(j)}+b_{1,2}\sigma_x^{(i)}\sigma_y^{(j)}+\cdots+b_{3,3}\sigma_z^{(i)}\sigma_z^{(j)}.
\end{align}
The constants $a_{i}$ and $b_{i,j}$ are chosen randomly according to a Gaussian distribution with zero mean and unit variance.  For ease of comparison, each random translationally invariant Hamiltonian is re-normalized to have unit norm. Note that one-dimensional translation-invariant systems with finite-range interaction do not exhibit finite temperature phase transitions.

The numerical experiments begin by drawing a random Hamiltonian $H_{\Lambda_n}^p$ for a fixed value of $\beta$ and energy window $\delta$.  The first step is to compute the energy density $u$ using $u=\frac{1}{|\Lambda_n|}{\rm Tr}\left(\gamma_{\Lambda_n}^p H_{\Lambda_n}^p \right)$ where $\gamma_{\Lambda_n}^p=\exp(-\beta H_{\Lambda_n}^p)/Z$ is the thermal state that results from the choice of $\beta$.  The Hamiltonian is then diagonalized and all energy eigenvectors within the window $(u-\delta, u)$ are found.  A random state $|\psi\rangle$ is then constructed out of the span of these vectors, and then we compute 
$\left\| {\rm Tr}_{\Lambda_n\setminus\Lambda_m} |\psi\rangle\!\langle\psi| - {\rm Tr}_{\Lambda_n\setminus\Lambda_m} \frac{\exp(-\beta H_{\Lambda_n}^p)}{Z}\right\|_1$, as per Theorem~\ref{TheCanonicalTypicalityPeriodic}.  We take the subsystem to consist of a single qubit, i.e.\ $m=1$, and the bath contain $n-1$ qubits  in all these examples.
This process is repeated for many such random Hamiltonians and we compute the mean and the standard deviation of these distances, which allows us to see whether the correspondence predicted by Theorem~\ref{TheCanonicalTypicalityPeriodic} is typical for this ensemble of random local Hamiltonians.

\begin{figure}[t!]
\begin{minipage}[t]{0.45\textwidth}
\includegraphics[width=\textwidth]{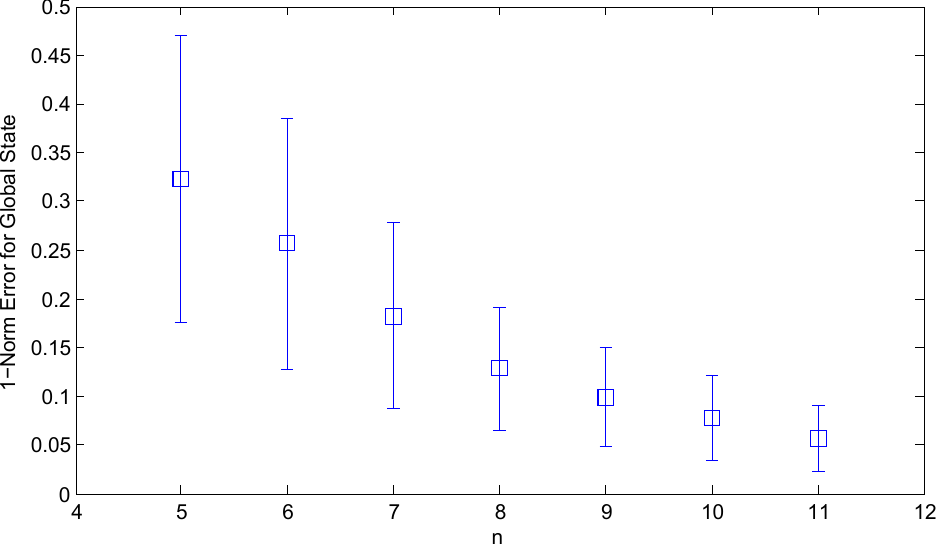}
\caption{$1$--norm difference between the reduced density operator and the reduced global Gibbs state for a system of $n$ qubits.  The squares represent the ensemble means and the error bars give the standard deviations of the differences between the Gibbs state and the subsystem trace.  The data was collected for $\beta=0.1$ and $\delta=0.02n$ and $400$ random Hamiltonians were considered for each $n$. \label{fig:global}}
\end{minipage}
\hspace{0.5cm}
\begin{minipage}[t]{0.45\textwidth}
\includegraphics[width=\textwidth]{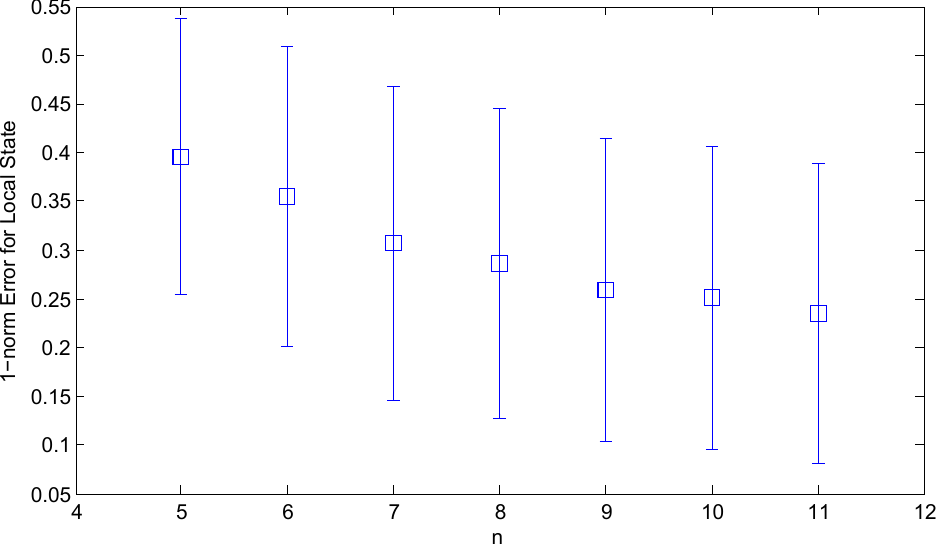}
\caption{$1$--norm difference between the reduced density operator and the reduced local Gibbs state for a system of $n$ qubits.  The squares and error bars are defined identically to those in Figure~\ref{fig:global}.  The data was collected for $\beta=0.1$ and $\delta=0.02n$ and $400$ random Hamiltonians were considered for each $n$. \label{fig:local}}
\end{minipage}
\end{figure}

The data in Figure~\ref{fig:global} shows that the distance between the reduced density matrix of the pure state and the Gibbs state shrinks as $n$ increases, roughly as $\mathcal{O}(1/n)$.  The error bars (representing the standard deviation of the discrepancy with the canonical state) also shrink as $n$ increases, illustrating that almost all such random translationally invariant $2$--local  Hamiltonians agree with the predictions of Theorem~\ref{TheCanonicalTypicalityPeriodic} and  in turn that there is a strong correspondence between the subsystem traces of the global Gibbs state and $|\psi\rangle\!\langle \psi|$.  

On the other hand, Figure~\ref{fig:local} shows that substituting the local Gibbs state for the subsystem trace of the global Gibbs state causes this correspondence to break down.  In particular, we see no clear evidence  that the ensemble mean of the differences between ${\rm Tr}_{\Lambda_n\setminus\Lambda_m} |\psi\rangle\!\langle\psi|$ and the local Gibbs state approaches zero as $n$ increases;  more tellingly, the standard deviation of the differences does not seem to decrease with $n$.  These results suggest that even as $n$ increases, ${\rm Tr}_{\Lambda_n\setminus\Lambda_m} |\psi\rangle\!\langle\psi|$ remains distinct from the local Gibbs state.  Thus the correspondence suggested by~Theorem~\ref{TheCanonicalTypicalityPeriodic} is correct and the na\"ive correspondence between the local Gibbs state and ${\rm Tr}_{\Lambda_n\setminus\Lambda_m} |\psi\rangle\!\langle\psi|$ is incorrect.

Regarding dynamical thermalization,
there are two caveats that we need to check in order to justify the applicability of Theorem~\ref{TheThermalizationPeriodic}.
First, we need to ensure that almost all Hamiltonians drawn from this random ensemble are non--degenerate,
in order to ensure thermalization for arbitrary initial states with maximal population entropy. Figure~\ref{fig:gap} shows that the probability of small eigenvalue gaps is suppressed, hence Hamiltonians that are typical of the random local Hamiltonian ensemble will be non--degenerate.  Second, we need to show that the gap degeneracy $D_G(H_{\Lambda_n}^p)$ is not too large.
Figure~\ref{fig:gapspace} shows that, with high probability, the eigenvalue gaps between any two energy levels will be distinct from any other such gap in the system, hence $D_G(H_{\Lambda_n}^p)=1$ with high probability.

These results illustrate the application of our results to a wide range of physically realistic random $2$--local Hamiltonians.  It is further reasonable to expect that broad classes of physically realistic closed quantum systems will agree with the canonical distribution, illuminating the mechanism by which thermodynamics emerges for macroscopic closed quantum systems.

\begin{figure}[t]
\begin{minipage}[t]{0.45\textwidth}
\includegraphics[width=\textwidth]{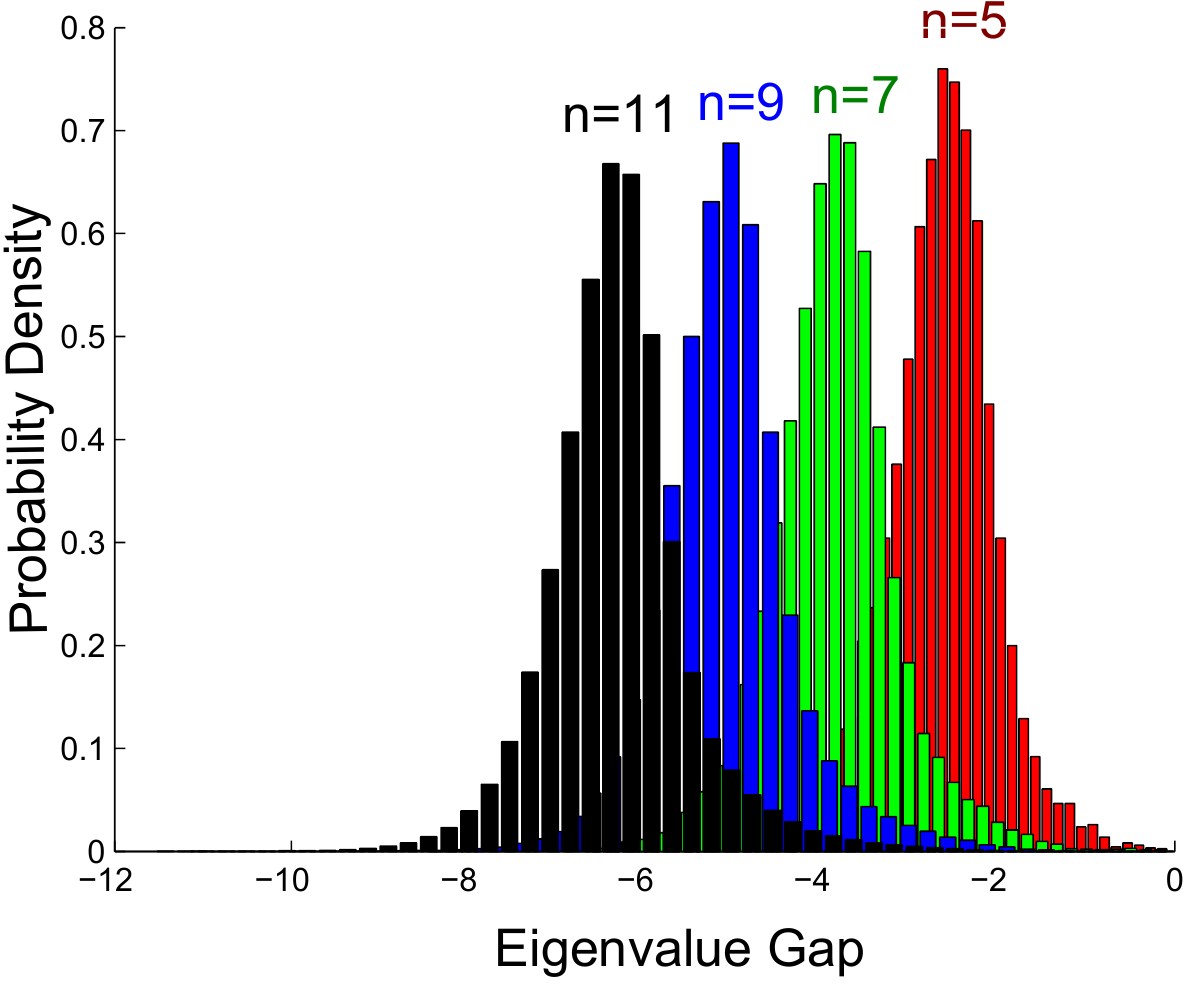}
\caption{Probability density of eigenvalue gaps for random Hamiltonians with $n=5,7,9$ and $11$ qubits.  The $x$--axis is $\log_{10}({\rm gap})$ for $100$ random Hamiltonians.  No degenerate eigenvalues were ever detected in this sample within numerical error.\label{fig:gap}}
\end{minipage}
\hspace{0.5cm}
\begin{minipage}[t]{0.45\textwidth}
\includegraphics[width=\textwidth]{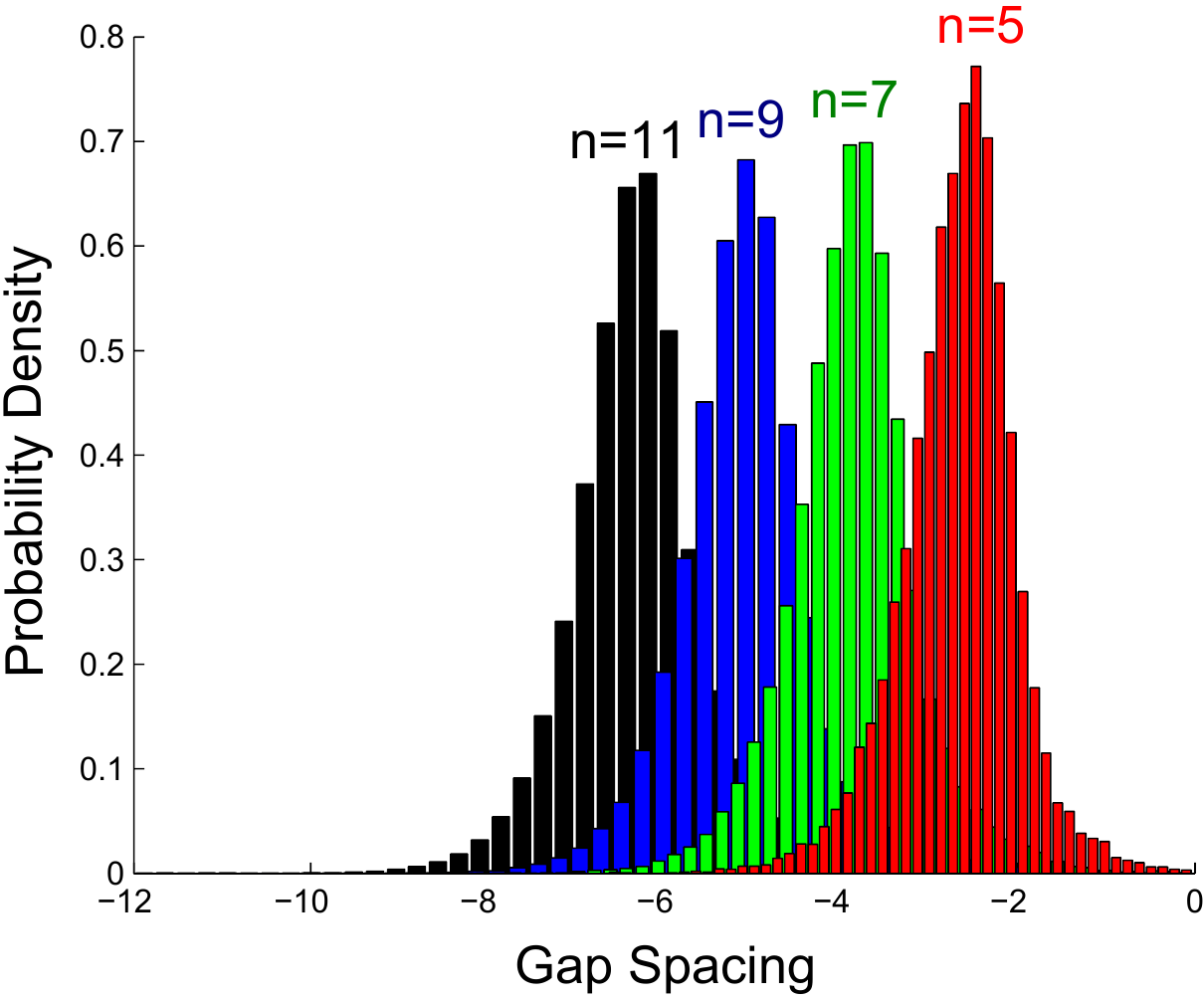}
\caption{Probability density of for the eigenvalue gap spacings for random Hamiltonians with $n=5,7,9$ and $11$ qubits.  The $x$--axis is $\log_{10}({\rm gap}(\rm{gap}))$ for $100$ random Hamiltonians.  No degenerate eigenvalue gaps were ever detected in this sample within numerical error.\label{fig:gapspace}}
\end{minipage}
\end{figure}

\subsection{Local diagonality of energy eigenstates}
\label{SecLocalDiag}

A strong sense in which the eigenstates of a local Hamiltonian $H$ could thermalize is that their reduced density matrix of a region $\Lambda$ (much smaller than the full lattice $\Lambda_{\rm lattice} =\Lambda_n$) is approximately equal to a thermal state in that region,
\begin{equation}\label{nt}
  {\rm Tr}_{\bar \Lambda} |E\rangle\! \langle E| 
  \ \approx\ 
  \frac{e^{-\beta H_\Lambda} }{{\rm tr}\, e^{-\beta H_\Lambda}}\ ,
\end{equation}
where ${\rm Tr}_{\bar \Lambda}$ denotes trace on the Hilbert space associated to the complementary region $\bar\Lambda= \Lambda_{\rm lattice} \setminus \Lambda$, and $H_\Lambda$ is the sum of all terms of $H$ which are fully contained in the region $\Lambda$. The inverse temperature $\beta$ should be chosen such that $\langle E| H_\Lambda |E\rangle = {\rm tr}( H_\Lambda e^{-\beta H_\Lambda}) /{\rm tr}\, e^{-\beta H_\Lambda}$ holds.

\begin{figure}[!hbt]
\begin{center}
\includegraphics[angle=0, width=4cm]{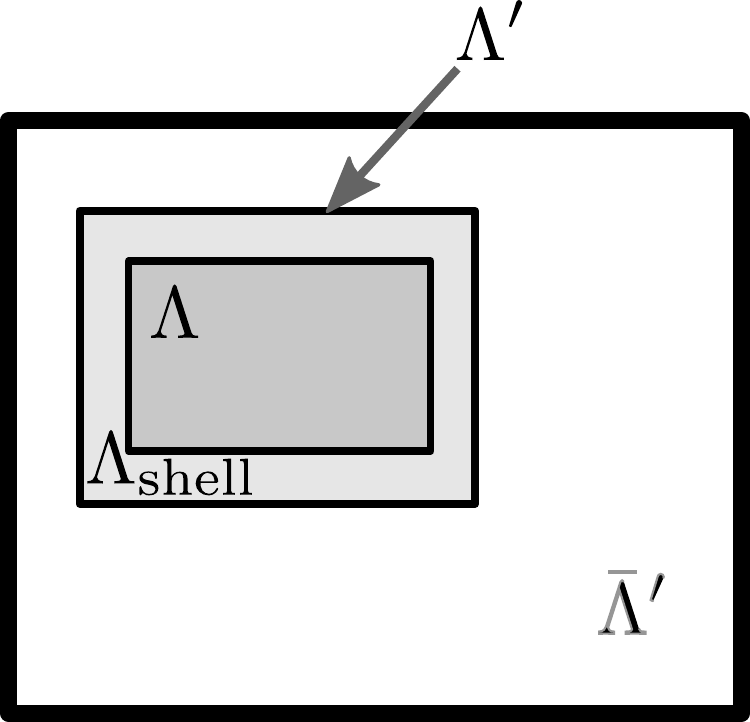}
\caption{Subdivision of the whole lattice, $\Lambda_{\rm lattice}=\Lambda_n$, into regions as used in this subsection. We have $\Lambda'=\Lambda\cap\Lambda_{\rm shell}$.}
\label{fig_regions}
\end{center}
\end{figure}
A possible concern is that the Hamiltonian $H_\Lambda$ has open boundary conditions, hence we expect boundary effects in the eigenstates of $H_\Lambda$ which are not present in ${\rm Tr}_{\bar \Lambda} |E\rangle\! \langle E|$; and this makes unlikely that relation~\eqref{nt} holds. 
A way to get rid of the boundary effects is by defining a slightly larger region $\Lambda'$ which includes a shell of width $l$ around $\Lambda$, cf.\ Figure~\ref{fig_regions}; that is
\begin{equation}\label{lp}
  \Lambda' := \{ x\in \Lambda_{\rm lattice} : \exists y\in \Lambda : {\rm dist}(x,y) \leq l \}\ .
\end{equation}
If instead of~\eqref{nt} we consider the thermal state in $\Lambda'$ and trace out the shell $\Lambda_\mathrm{shell}:= \Lambda' \setminus \Lambda$, then the approximate equality
\begin{equation}\label{yt}
  {\rm Tr}_{\bar \Lambda} |E\rangle\! \langle E| 
  \ \approx\ 
  {\rm Tr}_{\Lambda_\mathrm{shell}}\!
  \left(
  \frac{e^{-\beta H_{\Lambda'}} }{{\rm tr}\, e^{-\beta H_{\Lambda'}}} \right)
\end{equation}
is more likely to hold in generic systems, because by tracing out the shell we may eliminate the boundary effects of the eigenvectors of $H_{\Lambda'}$. (As before, we denote by $H_{\Lambda'}$ the sum of all terms in $H$ which are fully contained in $\Lambda'$.)

It is expected that the relation~\eqref{yt} holds for generic local Hamiltonians, but not for all local Hamiltonians. For example, consider the translational-invariant quantum Ising Hamiltonian in one dimension that we analyzed in Subsection~\ref{SecIsing}. This is a Hamiltonian without interaction terms, such as $H_{\Lambda_{\rm shell}}=\sum_{i=1}^n h_i$ for
$\Lambda_{\rm shell}=[1,n]$, with constant single-site
terms $h_i$. If, for example, $h_i=\left(\begin{array}{cc} 1 & 0 \\ 0 & -1 \end{array}\right)$, the computational basis vectors $|E\rangle=|x_1 x_2\ldots x_n\rangle$ with $x_i\in\{0,1\}$
are energy eigenstates. Even for those eigenstates that correspond to finite energies $E>0$ with corresponding inverse temperature $\beta<\infty$, the local reduced state on
$\Lambda=[1,m]$, $m\ll n$, is $\Tr_{\bar\Lambda}|E\rangle\langle E|=|x_1\ldots x_m\rangle\langle x_1\ldots x_m|$. This is a pure state, far away from any thermal state
of temperature $\beta$. Thus, \eqref{yt} does not hold for the Ising model.

In summary, extra conditions are necessary for~\eqref{yt} to hold. Folk wisdom tells us that such conditions could be along the lines of non-integrability, although this is not yet a clear and mathematically well-defined concept within quantum theory. In this work, we follow a different approach: instead of looking for additional conditions, we relax the statement~\eqref{yt}. One way to do this is by noticing that the state $e^{-\beta H_{\Lambda'}} /{\tr}\, e^{-\beta H_{\Lambda'}}$ is diagonal in the eigenbasis of $H_{\Lambda'}$. Our weakened statement is informally the following:

\begin{quote}{\em
For any eigenvalue $E$ of $H$ there is a density matrix $\omega_E$ defined in the extended region $\Lambda'$ which is weakly diagonal in the eigenbasis of $H_{\Lambda'}$ and satisfies}
\[
  {\rm Tr}_{\bar \Lambda} |E\rangle\! \langle E| 
  \ \approx\ 
  {\rm Tr}_{\Lambda_\mathrm{shell}}
  \omega_E \ .
\]
\end{quote}
The meaning of weakly diagonal will be made precise in the statement of the theorem below. But before, let us specify the type of systems that we are considering.
Exactly as explained at the beginning of Subsection~\ref{SecEquiv}, we consider local Hamiltonians on a cubic lattice, with a finite-dimensional Hilbert space at each site.
By local we mean that the Hamiltonian $H$ has finite interaction range $r$. This means that if we write it as
\[
  H = \sum_{{\cal X} \subseteq \Lambda_\mathrm{lattice}} \Phi({\cal X})\ ,
\]
where $\Phi({\cal X})$ has only support on the region ${\cal X}$, then for any region ${\cal X} \subseteq \Lambda_\mathrm{lattice}$ such that ${\rm diam}\,{\cal X}:=\max_{x,x' \in {\cal X}} \mathrm{dist}(x,x') >r$ we have $\Phi({\cal X}) =0$ (the definition of ${\rm dist}$ is given in~(\ref{eqDefBoundary})).
However, in contrast to the previous subsection, we do not need to assume that the interaction is translation-invariant.
This type of Hamiltonian satisfies a Lieb-Robinson bound~\cite{LiebRobinson,Bravyi}
(see~\cite{Masanes} for a simpler proof). That is, let $X,Y$ be two matrices acting non-trivially in the regions ${\cal X,Y} \subseteq \Lambda_{\rm lattice}$ which are separated by a distance ${\rm dist}({\cal X,Y})$, and let $X(t)= e^{iHt}X e^{-iHt}$. There are positive constants $C,c,v$ such that
\begin{equation}\label{lrb}
  \| [X(t),Y] \|_\infty \ \leq\ 
  C\, \|X\|_\infty \|Y\|_\infty \min\{{\cal |X|,|Y|}\}\, e^{-c[{\rm dist}({\cal X,Y}) -v |t|]}\ .
\end{equation}
The constants $C,c,v$ only depend on coarse features of the lattice and the Hamiltonian, like the interaction length, and the $\|\cdot \|_\infty$-norm of the local terms in the Hamiltonian. The constant $v$ is called the Lieb-Robinson velocity, and it is an upper-bound for the speed at which information travels through the lattice.

\vspace{5mm}\noindent
\begin{theorem}[Weak local diagonality]
\label{TheWeakLocalDiagonality}
Let $\Phi$ be any finite-range interaction (not necessarily translation-invariant),
let $\Lambda \subseteq \Lambda_{\rm lattice}$ be any region of the lattice, and let $\Lambda' \subseteq \Lambda_{\rm lattice}$ be the set of points at distance not larger than $l$ from $\Lambda$, as defined in~\eqref{lp}.
Define the regions $\Lambda_\mathrm{shell}= \Lambda' \setminus \Lambda$ and the complements $\bar\Lambda= \Lambda_{\rm lattice} \setminus \Lambda$ and $\bar\Lambda' = \Lambda_{\rm lattice} \setminus \Lambda'$. Let $H$ be a local Hamiltonian as defined above, with finite interaction range $r\leq l$.
For each eigenvector $|E\rangle$ of the Hamiltonian $H$ we define the state $\omega_E$ in the region $\Lambda'$ as
\[
  \omega_E := \int_{-\infty}^\infty\!\! dt\, g(t)\, e^{-i H_{\Lambda'} t}\, {\rm Tr}_{\bar \Lambda'} (|E\rangle\! \langle E|)\,  e^{i H_{\Lambda'} t}\ ,
\]
where $g(t) = (2\pi \sigma^2)^{-1/2}\, e^{-t^2/(2 \sigma^2)}$
and $\sigma^2 = (l-r)/(4cv^2)$.
The state $\omega_E$ is weakly diagonal in the eigenbasis of $H_{\Lambda'}$, denoted $|e\rangle$, in the sense that
\begin{equation}\label{ineq1}
  \left| \langle e_1 | \omega_E |e_2 \rangle \right|
\leq
  e^{-(l-r) (e_1 -e_2)^2 /(8cv^2)}
  \ .
\end{equation}
The state $\omega_E$ is almost indistinguishable from $|E\rangle\! \langle E|$ inside the region $\Lambda$, that is
\begin{equation}\label{t2}
  \left\| {\rm Tr}_{\Lambda_\mathrm{shell}} (\omega_E) - 
  {\rm Tr}_{\bar\Lambda}(|E\rangle\! \langle E|) \right\|_1
\ \leq \frac 2 {\sqrt{2\pi}} A J \sigma(CA+2)e^{-c(l-r)/2},
\end{equation}
where $A$ is the number of subsets $X$ with $\Phi(X)\neq 0$ that have
non-empty intersection with both $\Lambda'$ and $\bar\Lambda'$.
\end{theorem}

Note that the number $A$ quantifies the size of the boundary of $\Lambda'$; so for a three-dimensional lattice, $A$ is an area. Also, we stress the fact that closeness in $\|\cdot \|_1$-norm is a very strong feature, and it really implies that the two states in the left-hand side of~\eqref{t2} are almost indistinguishable. The right-hand side of~\eqref{t2} can be made small by choosing the thickness of the shell to be
\[
  l \gtrsim \frac{6}{c} \log A +r\ .
\]
Still, for large regions $\Lambda$, the relative volume of the shell $lA/|\Lambda|$ vanishes. 

If the local dimension is $d$, then the dimension of the Hilbert space associated to the region $\Lambda'$ is $d^{|\Lambda'|}$. Hence, the expected size of the entries of $\omega_E$ is of the order of $d^{-|\Lambda'|}$, which is very small. This may rise the concern that bound~\eqref{ineq1} is trivial. To see that this is not the case, we note that the largest entry of $\omega_E$ is at least $d^{-|\Lambda'|}$. Also, since $H_{\Lambda'}$ is a local Hamiltonian, the range of energies is $\Delta e \sim J |\Lambda'|$. This implies that the exponent of~\eqref{ineq1} is proportional to $|\Lambda'|^2$, while the exponent of the largest entry is proportional to $|\Lambda'|$, which is much smaller. In summary, for large enough regions $|\Lambda|$,
the bound~\eqref{ineq1} is non-trivial. It is a consequence of the locality of interactions as expressed by the Lieb-Robinson bound.

\begin{proof}
Using the fact that the $|e\rangle$ are the eigenvectors of $H_{\Lambda'}$ we obtain
\begin{eqnarray*}
&&
  \langle e_1 | \omega_E |e_2 \rangle
=
  \int\!\! dt\, g(t)\, e^{-i (e_1 - e_2) t}\,   
  \langle e_1 | {\rm Tr}_{\bar \Lambda'} (|E\rangle\! \langle E|) |e_2 \rangle
=
  e^{-(e_1 -e_2)^2 \sigma^2/2}
  \langle e_1 | {\rm Tr}_{\bar \Lambda'} (|E\rangle\! \langle E|) |e_2 \rangle\ ,
\end{eqnarray*}
which implies~\eqref{ineq1}. Using the triangle inequality for the norm $\|\cdot\|_1$ we obtain
\begin{eqnarray}
\nonumber
  \left\| {\rm Tr}_{\Lambda_\mathrm{shell}} (\omega_E) - 
  {\rm Tr}_{\bar\Lambda}(|E\rangle\! \langle E|) \right\|_1 &=&
  \left\|   \int\!\! dt\, g(t)\,
  {\rm Tr}_{\bar \Lambda}\! \left( e^{-i H_{\Lambda'} t} |E\rangle\! \langle E|e^{i H_{\Lambda'} t} - |E\rangle\! \langle E| \right) \right\|_1
\\ 
&\leq &
  \int\!\! dt\, g(t)\, \left\|   
  {\rm Tr}_{\bar \Lambda}\! \left( e^{-i H_{\Lambda'} t} |E\rangle\! \langle E|e^{i H_{\Lambda'} t} - |E\rangle\! \langle E| \right) \right\|_1\ .
  \label{E7} 
\end{eqnarray}
Next, we use the identity $\|Y\|_1= \max_X |{\rm Tr}(XY)|$, where the maximum is over all Hermitian matrices $X$ which satisfy
$-\mathbf{1}\leq X\leq \mathbf{1}$. Since we apply this to an observable on $\Lambda$, it follows that $X$ is fully supported on $\Lambda$.
We also use the fact that $e^{i H t} |E\rangle\! \langle E| e^{-i H t}= |E\rangle\! \langle E|$ for any $t$, obtaining
\begin{eqnarray}
\nonumber
  \left\|   
  {\rm Tr}_{\bar \Lambda}\! \left( e^{-i H_{\Lambda'} t} |E\rangle\! \langle E|e^{i H_{\Lambda'} t} - |E\rangle\! \langle E| \right) \right\|_1
&=&
  \max_X \left|   
  {\rm Tr}\! \left[ X\left( e^{-i H_{\Lambda'} t}  e^{i H t} |E\rangle\! \langle E| e^{-i H t} e^{i H_{\Lambda'} t} - |E\rangle\! \langle E| \right) \right] \right|
\\
&= &
  \max_X \left|   
  \langle E| e^{-i H t} e^{i H_{\Lambda'} t} X e^{-i H_{\Lambda'} t} e^{i H t} - X |E\rangle \right|.
  \label{E9}
\end{eqnarray}
Now we use the inequality $|\langle\alpha| Y |\beta\rangle |\leq \| Y\|_\infty$ for any pair of unit vectors $|\alpha\rangle, |\beta\rangle$. Also, we use the fact that
$[X,H_{\bar\Lambda'}]=[H_{\Lambda'},H_{\bar\Lambda'}]=0$, and define $H_A := H- H_{\bar\Lambda'} -H_{\Lambda'}$. We obtain
\begin{equation}\label{E10}
\left|   
  \langle E| e^{-i H t} e^{i H_{\Lambda'} t} X e^{-i H_{\Lambda'} t} e^{i H t} - X |E\rangle \right|
\ \leq\  
  \left\|   
  e^{-i H t} e^{i (H-H_A) t} X e^{-i (H-H_A) t} e^{i H t} - X\right\|_\infty
\end{equation}
Next, we use the matrix identity $M(t)- M(0) = \int_0^t dt_1 \frac{\partial}{\partial t_1} M(t_1)$, the triangle inequality, and the unitary invariance of the operator norm,
$\| e^{-iHt_1} Y e^{iHt_1} \|_\infty = \| Y\|_\infty$. If $t\geq 0$ then
\begin{eqnarray*}
  \left\|   
  e^{-i H t} e^{i (H-H_A) t} X e^{-i (H-H_A) t} e^{i H t} - X\right\|_\infty  &=&
  \left\| \int_0^t \!\! dt_1 \frac{\partial}{\partial t_1} \left(   
  e^{-i H t_1} e^{i (H-H_A) t_1} X e^{-i (H-H_A) t_1} e^{i H t_1} 
  \right) \right\|_\infty
\\ 
&\leq &
  \int_0^{|t|} \!\! dt_1 \left\|  \left[   
  H_A\, , e^{i H_{\Lambda'} t_1} X e^{-i H_{\Lambda'} t_1}
  \right] \right\|_\infty\ .
\end{eqnarray*}
If $t< 0$, then the substitution $t_2 := -t_1$ in the integral yields
\begin{eqnarray*}  
\left\|   
  e^{-i H t} e^{i (H-H_A) t} X e^{-i (H-H_A) t} e^{i H t} - X\right\|_\infty=
&&
  \left\| \int_0^{|t|} \!\! dt_2 \frac{\partial}{\partial t_2} \left(   
  e^{i H t_2} e^{-i (H-H_A) t_2} X e^{i (H-H_A) t_2} e^{-i H t_2} 
  \right) \right\|_\infty
\\
&\leq &
  \int_0^{|t|} \!\! dt_2 \left\|  \left[   
  H_A\, , e^{-i H_{\Lambda'} t_2} X e^{i H_{\Lambda'} t_2}
  \right] \right\|_\infty\ .
\end{eqnarray*}
In both cases, we can apply the Lieb-Robinson bound to the two regions ${\cal X} = \Lambda$ and $\cal Y$ the support region of $H_A$ (covering the
boundary of $\Lambda'$ and of $\bar\Lambda'$). For all $t\in\R$, we get
\[
  \left\|  \left[   
  H_A\, , e^{i H_{\Lambda'} t} X e^{-i H_{\Lambda'} t}
  \right] \right\|_\infty
\ \leq\  
  \| H_A \|_\infty \min\! \left\{2,
  C A\, e^{-c(l-r) +cv |t|}
  \right\},
\]
which implies
\begin{eqnarray*}
     \left\|e^{-i H t} e^{i (H-H_A) t} X e^{-i (H-H_A) t} e^{i H t} - X\right\|_\infty &\leq& \|H_A\|_\infty \min\left\{\int_0^{|t|} dt_1\cdot 2, \enspace \int_0^{|t|} dt_1\, C A e^{-c(l-r)+cv|t_1|}\right\}\\
     &\leq& \|H_A\|_\infty\min\left\{2|t|,\enspace CA|t| e^{-c(l-r)+cv|t|}\right\}.
\end{eqnarray*}
Combining this with~\eqref{E7}, \eqref{E9}, \eqref{E10}, and dividing the integration \eqref{E7} into two intervals, we get for $t_0\geq 0$
\begin{eqnarray*}
  \left\| {\rm Tr}_{\Lambda_\mathrm{shell}} (\omega_E) - 
  {\rm Tr}_{\bar\Lambda}(|E\rangle\! \langle E|) \right\|_1
  &\leq& 2 \|H_A\|_\infty \int_0^\infty dt\, g(t) \min\{2t,\enspace CA t \, e^{-c(l-r)+cvt}\}\\
&\leq&
  2\| H_A \|_\infty \left( \int_0^{t_0}\!\! dt\, g(t)\, 
   CAt\, e^{-c(l-r) +cvt} + \int_{t_0}^\infty \!\! dt\, g(t)\, 
  2t\right)
\\ 
&\leq &
  2 \| H_A \|_\infty \left( 
  C A\, e^{-c(l-r)+c v t_0} \int_0^\infty dt\, g(t) t + 
  \frac{\sigma}{\sqrt{2\pi}}2\, e^{-t_0^2/(2\sigma^2)} \right)
\\
&\leq &
  2 \| H_A \|_\infty \left( 
  \frac{\sigma}{\sqrt{2\pi}} C A e^{-c(l-r)+cv t_0} + 
  \frac{\sigma}{\sqrt{2\pi}}2\, e^{-t_0^2/(2\sigma^2)} \right).
\end{eqnarray*}
Now choose $t_0:=(l-r)/(2v)$ such that $-c(l-r)+cv t_0 = -t_0^2/(2\sigma^2)$, and use $\sigma^2 = (l-r)/(4cv^2)$.
Furthermore,
\[
   H_A=H_{\Lambda_{\rm lattice}}-H_{\bar\Lambda'}-H_{\Lambda'}=\sum_{X\subset\Lambda_{\rm lattice}:\, X\cap\Lambda'\neq\emptyset\mbox{ and }
   X\cap\bar\Lambda'\neq\emptyset} \Phi(X),
\]
such that $\|H_A\|_\infty\leq A\,J$, where $J=\max_{X\subset\Z^\nu}\|\Phi(X)\|_\infty$, and $A$ is the number of subsets $X$ with $\Phi(X)\neq 0$ that have
non-empty intersection with both $\Lambda'$ and $\bar\Lambda'$.
\end{proof}

\section{Conclusions}
Our work provides a significant step towards a rigorous understanding for how closed quantum systems thermalize.  Our key innovations come from combining 
methods from quantum information theory and from more traditional mathematical physics techniques to address the problem.  Through this approach, we find that small subsystems of closed translation-invariant quantum systems
with finite-range interaction thermalize, in the sense that they relax towards the reduction of the global Gibbs state. In doing so, we not only provide a rigorous explanation for how a wide
class of physically significant Hamiltonians thermalize, but also show that the correct correspondence is with a reduction of the global system's Gibbs state,
not its local Gibbs state.

This work opens a number of interesting avenues for future work. One open problem is to obtain more explicit finite-size bounds, but these may
well depend on details of the specific model or interaction. Similarly, an interesting open question is whether $\omega_E$ in Theorem~\ref{TheETHMain}
has Boltzmann weights on its diagonal.
However, rigorously answering this question in the affirmative, and thus proving a complete version of the eigenstate thermalization hypothesis,
seems to require additional assumptions along the lines of nonintegrability. Thus, one may hope that attempts to prove the ETH for quantum lattice systems
will also lead to a better understanding and rigorous mathematical definition of the notion of integrability in the quantum case. We further believe that the methodology 
we provide will lead to further applications to be discovered in the future.  In particular, it may turn out that giving finite versions of asymptotic mathematical
physics results will prove to be as promising as using asymptotic results to prove statements on finite systems, which was the approach taken in this paper.

\section*{Acknowledgments}
MM would like to thank Jens Eisert, Joe Emerson, Patrick Hayden, and Sandu Popescu for discussions in early stages of this project, and Oscar Dahlsten for comments.
Research at Perimeter Institute is supported by the Government of Canada through Industry Canada and by the Province of Ontario through the Ministry of Research and Innovation.
LM acknowledges support from the EU ERC Advanced Grant NLST (PHYS RQ8784), EU Qessence project, EPSRC and the Templeton Foundation.
This work was partially supported by the COST Action MP1209.

\bigskip

\section{Correction (added March 30, 2021)}
This correction has also been published here:

M.\ P.\ M\"uller, E.\ Adlam, Ll.\ Masanes, and N.\ Wiebe, \emph{Correction to: Thermalization and Canonical Typicality in Translation-Invariant Quantum Lattice Systems}, Commun.\ Math.\ Phys.\ (2021), \href{https://doi.org/10.1007/s00220-021-04014-0}{DOI:10.1007/s00220-021-04014-0}

\bigskip

In our result about dynamical thermalization, the proof of the upper bound on the time average of the distance between the local evolved state $\rho^{(n)}(t)$ and the time-averaged state $\rho_{\rm avg}^{(n)}$ is wrong. While it is correct that this distance tends to zero for block size $|\Lambda_n|\to\infty$ (see corrected proof below), it is unclear whether it can be shown that this happens \emph{exponentially fast} in $|\Lambda_n|$. This affects Theorem 31, and hence also Theorem 3 (the summary of Theorem 31) and Theorem 33 (a small modification of Theorem 31).

This mistake is due to an error in Ref.~[C3] which we have used in our proof of Lemma 30. Ref.~[C3] claims that the R\'enyi entropy $H_q$ is convex in its parameter $q$, which is incorrect. This claim has been corrected in an erratum published on the author's homepage~[C4], but we became aware of this only recently.\\

We give a corrected version of Theorem 31 of our paper~[C1] in Theorem~\ref{TheCorrection} below. Its summary (and hence the correction of Theorem 3 of our paper) reads as follows.

\begin{theorem}[Correction of~{[C1, Theorem 3]}]
\label{TheMainCorrection}
If there is a unique equilibrium state around inverse temperature $\beta:=\lim_{n\to\infty}\beta_n$, if the (possibly pure) initial state has close to maximal population entropy, in the sense that
\[
   \bar S(\rho_0^{(n)})\geq S(\gamma_{\Lambda_n}^p (\beta_n))-o(|\Lambda_n|),
\]
and if each $H_{\Lambda_n}^p$ is non-degenerate with uniformly bounded gap degeneracy $\sup_n D_G(H_{\Lambda_n}^p)<\infty$, then unitary time evolution thermalizes the subsystem $\Lambda$ for most times $t$:
\begin{eqnarray*}
\left\langle \left\| {\rm Tr}_{\Lambda_n\setminus\Lambda} \rho^{(n)}(t)-{\rm Tr}_{\Lambda_n\setminus\Lambda} \frac{\exp(-\beta_n H_{\Lambda_n}^p)}{Z_n}\right\|_1\right\rangle &\stackrel{n\to\infty}\longrightarrow & 0.
\end{eqnarray*}
The gap degeneracy~[C5] is defined as $D_G(H_{\Lambda_n}^p):=\max_E|\{(i,j)\,\,|\,\, i\neq j, E_i-E_j=E\}|$, with $E_i$ the eigenvalues of $H_{\Lambda_n}^p$.
\end{theorem}

This formulation differs from the old one in the following two ways. First, it does not give concrete bounds on the time-averaged distance between $\rho^{(n)}(t)$ and its time average (it only says that this distance tends to zero for $n\to\infty$); second, it presumes that the gap degeneracy is uniformly bounded.

To prove its formal version (Theorem~\ref{TheCorrection} below), we need two elementary lemmas.
\begin{lemma}
\label{LemStrictConcavity}
Let $\Phi$ be a translation-invariant finite-range interaction which is not physically equivalent to zero, and let $\bar u$ be some energy density for which there is a unique Gibbs state at inverse temperature $\beta(\bar u)$. Then the real function $u\mapsto s(u)$ defined in [C1, Lemma 9] is strictly concave at $\bar u$ in the following sense: If $\bar u =\lambda u_0+(1-\lambda)u_1$ for some $u_0<u_1$ and $\lambda\in (0,1)$ then $s(\bar u)>\lambda s(u_0)+(1-\lambda)s(u_1)$.
\end{lemma}
\begin{proof}
Let $u_0<u_1$ and $u=\lambda u_0+(1-\lambda)u_1$ for some $\lambda\in (0,1)$. Let $\omega_{\beta(u_0)}$ be an arbitrary Gibbs state with energy density $u_0$ at inverse temperature $\beta(u_0)$, and similarly $\omega_{\beta(u_1)}$. Set $\omega:=\lambda\omega_{\beta(u_0)}+(1-\lambda)\omega_{\beta(u_1)}$, a translation-invariant state. Since the entropy density is affine on the translation-invariant states ([C2, Thm.\ IV.2.4]), we have
\[
   s(\omega)=\lambda\, s(\omega_{\beta(u_0)})+(1-\lambda)s(\omega_{\beta(u_1)})=\lambda s(u_0)+(1-\lambda) s(u_1).
\]
By construction, $u(\omega)=u$. Thus, due to~[C1, Lemma 9], we have $s(\omega)\leq s(u)$, hence $u\mapsto s(u)$ is concave.

Let us now apply the previous argumentation to the special case $u:=\bar u$, an energy density with a unique Gibbs state. Suppose that $s(\bar u)=s(\omega)$. Then the variational principle ([C1, Definition 6]) implies that $\omega$ is a Gibbs state at inverse temperature $\beta(\bar u)$. But the set of Gibbs states at inverse temperature $\beta(\bar u)$ is a face of the set of all translation-invariant states [C2, p.\ 348], hence $\omega_{\beta(u_0)}$ and $\omega_{\beta(u_1)}$ must both be Gibbs states at inverse temperature $\beta(\bar u)$, too. But these are distinct states, since they have different energy densities, contradicting the uniqueness of the Gibbs state at $\beta(\bar u)$. Therefore $s(\bar u)>s(\omega)$, and we get the statement of strict concavity as claimed.
\end{proof}

\begin{lemma}
\label{LemSmallOverlap}
Let $\Phi$ be a translation-invariant finite-range interaction which is not physically equivalent to zero. Suppose that the maximal energy degeneracy of $H_{\Lambda_n}^p$ grows at most subexponentially in $|\Lambda_n|$, i.e.\ $\log \max\{\tr(\pi_i^{(n)})\}=o(|\Lambda_n|)$, where $(\pi_i^{(n)})_i$ denotes the eigenprojectors of $H_{\Lambda_n}^p$. Let $(\rho^{(n)})_{n\in\mathbb{N}}$ be any sequence of $\Lambda_n$-translation-invariant states with
\[
   [\rho^{(n)},H_{\Lambda_n}^p]=0,\quad S(\rho^{(n)})\geq s\cdot |\Lambda_n|+o(|\Lambda_n|), \quad {\rm tr}(\rho^{(n)} H_{\Lambda_n}^p)=u\cdot |\Lambda_n|+o(|\Lambda_n|),
\]
where $u\in(u_{\min}(\Phi),u_{\max}(\Phi))$ is an energy density such that there is a unique Gibbs state at inverse temperature $\beta(u)$, and $s=s(u)$. Then $\max_i \tr(\rho^{(n)}\pi_i^{(n)})\stackrel{n\to\infty}\longrightarrow 0$.
\end{lemma}
\begin{proof}
We can write $u$ as some convex combination of two distinct energy densities in a small neighborhood of $u$, and then Lemma~\ref{LemStrictConcavity} implies that $s=s(u)>0$. Let us now argue by contradiction. Suppose that $\lambda^{(n)}:=\max_i \tr(\rho^{(n)}\pi_i^{(n)})$ does not converge to zero. Decompose the state $\rho^{(n)}$ as follows:
\begin{equation}
\label{eqDecomp}	
   \rho^{(n)}=\lambda^{(n)}\tau^{(n)}+(1-\lambda^{(n)})\sigma^{(n)},
\end{equation}
where $\tau^{(n)}=\pi_i^{(n)}\rho^{(n)}\pi_i^{(n)}/\lambda^{(n)}$ (note that $\lambda^{(n)}>0$), with $\pi_i^{(n)}$ the maximizing projector. If $\lambda^{(n)}\neq 1$, define $\sigma^{(n)}:=\bar\pi_i^{(n)}\rho^{(n)}\bar\pi_i^{(n)}/(1-\lambda^{(n)})$, where $\bar\pi_i^{(n)}:=\mathbf{1}-\pi_i^{(n)}$; if $\lambda^{(n)}=1$, set $\sigma^{(n)}=\bar\pi_i^{(n)}/\tr(\bar\pi_i^{(n)})$ (if $n$ is large enough, then $\pi_i^{(n)}\neq\mathbf{1}$, hence this is well-defined). It follows that $\tau^{(n)}$ and $\sigma^{(n)}$ are mutually orthogonal $\Lambda_n$-translation-invariant states that commute with $H_{\Lambda_n}^p$.

The sequences of real numbers $S(\sigma^{(n)})/|\Lambda_n|$, ${\rm tr}(\sigma^{(n)}H_{\Lambda_n}^p)/|\Lambda_n|$, ${\rm tr}(\tau^{(n)}H_{\Lambda_n}^p)/|\Lambda_n|$ and $\lambda^{(n)}$ are all bounded (the latter sequence bounded away from zero by assumption). Thus, we can find a subsequence $(n_k)_{k\in\N}$ such that
\[
   \lambda^{(n_k)}\stackrel{k\to\infty}\longrightarrow \delta>0,\quad\frac 1 {|\Lambda_{n_k}|} S(\sigma^{(n_k)}) \stackrel{k\to\infty}\longrightarrow s_1,\quad \frac 1 {|\Lambda_{n_k}|} \tr(\tau^{(n_k)}H_{\Lambda_{n_k}}^p)\stackrel{k\to\infty}\longrightarrow u_0,\quad \frac 1 {|\Lambda_{n_k}|} \tr(\sigma^{(n_k)}H_{\Lambda_{n_k}}^p)\stackrel{k\to\infty}\longrightarrow u_1,
\]
where $s_1$, $u_0$, $u_1$ are real numbers, and $0<\delta\leq 1$. Due to~(\ref{eqDecomp}), computing the von Neumann entropy, we have $S(\rho^{(n_k)})=\lambda^{(n_k)} S(\tau^{(n_k)})+(1-\lambda^{(n_k)})S(\sigma^{(n_k)})+\mathcal{O}(1)$. Since $S(\tau^{(n_k)})\leq\log\tr(\pi_i^{(n_k)})=o(|\Lambda_{n_k}|)$, this implies $s\leq (1-\delta)s_1$. Thus, $s>0$ yields $\delta<1$. Similarly, computing the energy expectation value, we obtain $u=\delta u_0+(1-\delta)u_1$. 

Suppose that $s_1\geq s(u_1)$, then $s_1-\beta u_1\geq p(\beta,\Phi)$ for $\beta:=\beta(u_1)$, hence~[C1, Lemma 8] implies that we must have equality, i.e.\ $s_1=s(u_1)$. In summary, we conclude that $s_1\leq s(u_1)$. Therefore
\[
   s(u)=s\leq (1-\delta)s_1 \leq \delta\, s(u_0)+(1-\delta)s(u_1).
\]
Since $s$ is strictly concave at $u$ due to Lemma~\ref{LemStrictConcavity} above, this is only possible if $u_0=u_1=u$. Hence
\[
   0<s(u)\leq (1-\delta)s_1 \leq (1-\delta)s(u_1)=(1-\delta)s(u)
\]
which is a contradiction.
\end{proof}
This allows us to obtain a corrected version of~[C1, Theorem 31].
\begin{theorem}[{Correction of~{[C1, Theorem 31]}: Thermalization, periodic boundary conditions}]
\label{TheCorrection}
Let $\Phi$ be a translation-invariant finite-range interaction which is not physically equivalent to zero. Suppose that the maximal energy degeneracy of $H_{\Lambda_n}^p$ grows at most subexponentially in $|\Lambda_n|$, i.e.\ $\log \max\{\tr(\pi_i^{(n)})\}=o(|\Lambda_n|)$, where $(\pi_i^{(n)})_i$ denotes the eigenprojectors of $H_{\Lambda_n}^p$, and $\sup_n D_G(H_{\Lambda_n}^p)<\infty$. Let $(\rho_0^{(n)})_{n\in\N}$ be some sequence of initial states on $\Lambda_n$ which have energy expectation value $U_n:=\tr(\rho_0^{(n)}H_{\Lambda_n}^p)$ with density $U_n/|\Lambda_n|$ converging to some value $u\in (u_{\min}(\Phi),u_{\max}(\Phi))$ as $n\to\infty$, such that there is a unique Gibbs state around inverse temperature $\beta(u)$.

Define the `population entropy'' $\bar S(\rho_0^{(n)}):=S(\lambda_1,\ldots,\lambda_N)$, where $S$ is Shannon entropy, and $\lambda_i:=\tr(\rho_0^{(n)}\pi_i^{(n)})$ is the probability that the $i$-th level is populated. Suppose that for every $n$ large enough, either $H_{\Lambda_n}^p$ is non-degenerate or every $\pi_i^{(n)}\rho_0^{(n)}\pi_i^{(n)}$ is $\Lambda_n$-translation-invariant. Then, determine the inverse temperature $\beta_n$ for which
\[
   \tr(H_{\Lambda_n}^p \gamma_{\Lambda_n}^p(\beta_n))=U_n,\quad \mbox{where }\gamma_{\Lambda_n}^p(\beta_n):=\frac{\exp(-\beta_n H_{\Lambda_n}^p)}{Z_n}.
\]
If the initial states have close to maximal population entropy in the sense that
\[
   \bar S(\rho_0^{(n)})\geq S(\gamma_{\Lambda_n}^p(\beta_n))-o(|\Lambda_n|),
\]
then unitary time evolution $\rho^{(n)}(t):=\exp(-itH_{\Lambda_n}^p)\rho_0^{(n)}\exp(it H_{\Lambda_n}^p)$ thermalizes the subsystem $\Lambda_m$ for most times $t$:
\[
   \lim_{n\to\infty}\left\langle \left\| {\rm Tr}_{\Lambda_n\setminus\Lambda_m} \rho^{(n)}(t)-{\rm Tr}_{\Lambda_n\setminus\Lambda_m} \frac{\exp(-\beta_n H_{\Lambda_n}^p)}{Z_n}\right\|_1\right\rangle =0,
\]
where $Z_n=\tr(\exp(-\beta_n H_{\Lambda_n}^p))$, and $\langle\cdot\rangle$ denotes the average over all times $t\geq 0$. Furthermore, in this statement, $\beta_n$ can be replaced by $\beta:=\beta(u)$.
\end{theorem}
\begin{proof}
The only ingredient in the proof of~[C1, Theorem 31] that has to be corrected is the argument that lower-bounds the ``effective dimension'' $d_{\rm eff}$. The old proof erroneously claimed that $d_{\rm eff}$ grows exponentially in $|\Lambda_n|$, but this relied on a wrong claim about the R\'enyi entropy of Ref.~[C3]. We now give a simple alternative argument which makes use of the R\'enyi entropy $S_\infty(\lambda_1,\ldots,\lambda_N)=-\log\max_i\lambda_i$ and the inequality $S_2\geq S_\infty$~[C4]. Namely,
\[
   d_{\rm eff}=\exp(S_2(\lambda_1,\ldots,\lambda_N))\geq \exp(S_\infty(\lambda_1,\ldots,\lambda_N))=\left(\max_i\lambda_i\right)^{-1}\stackrel{n\to\infty}\longrightarrow \infty
\]
according to Lemma~\ref{LemSmallOverlap} above, applied to the sequence of states $\bar\rho_0^{(n)}=\sum_i \pi_i^{(n)}\rho_0^{(n)}\pi_i^{(n)}$. Since we have assumed that the gap degeneracy is uniformly bounded, this is enough to show that $\rho^{(n)}(t)$ is close to its time average for most times $t$ if $n$ is large. The rest of the proof works without modification (note that $\rho(\beta_n)$ should read $\gamma_{\Lambda_n}^p(\beta_n)$).
\end{proof}
Finally, [C1, Theorem 33] has to be corrected analogously. We omit the obvious details.

\bigskip
\textbf{Acknowledgments.} We are grateful to Henrik Wilming for pointing out the mistake in Lemma 30 of the old version, and for further helpful discussions.

\noindent\makebox[\linewidth]{\resizebox{0.3333\linewidth}{1pt}{$\bullet$}}\bigskip

\begin{tabular}{p{0.05\textwidth}p{0.9\textwidth}}
[C1] & M.\ P.\ M\"uller, E.\ Adlam, Ll.\ Masaanes, and N.\ Wiebe, \emph{Thermalization and Canonical Typicality in Translation-Invariant Quantum Lattice Systems}, Commun.\ Math.\ Phys.\ \textbf{340}, 499--561 (2015).\\
\strut [C2] & B.\ Simon, \emph{The Statistical Mechanics of Lattice Gases, Vol.\ 1}, Princeton University Press, Princeton, 1993.\\
\strut [C3] & K.\ \.{Z}yczwkoski, \emph{R\'enyi extrapolation of Shannon entropy}, Open Sys.\ Inf.\ Dyn.\ \textbf{10}, 297--310 (2003).\\
\strut [C4] & K.\ \.{Z}yczwkoski, \emph{R\'enyi extrapolation of Shannon entropy}, corrigendum, \url{http://www.cft.edu.pl/~karol/pdf/Zy03b.pdf}, 2005.\\
\strut [C5] & A.\ J.\ Short and T.\ C.\ Farrelly, \emph{Quantum equilibration in finite time}, New J.\ Phys.\ \textbf{14}, 013063 (2012).
\end{tabular}

\end{document}